\documentclass[11pt]{article}
\usepackage[margin=1in]{geometry}
\usepackage[colorlinks=true,
            linkcolor=blue,
            citecolor=blue,
            urlcolor=black]{hyperref}

\usepackage[sort&compress]{natbib}
 \bibpunct[, ]{[}{]}{,}{n}{}{,}%

\usepackage{fixmath}
\usepackage{bm}
\usepackage{amsbsy}
\usepackage{color}
\usepackage{verbatim}
\usepackage{multirow}
\usepackage{amssymb}
\usepackage{amsthm}
\usepackage{amsmath}
\usepackage{array}
\usepackage{mathtools}

\usepackage{graphicx}
\usepackage{tikz}
\usetikzlibrary{positioning}
\usepackage{xcolor}

\usepackage{thm-restate}
\usepackage{algpseudocode}
\usepackage{algorithm}
\usepackage{algorithmicx}

\newtheorem{theorem}{Theorem}
\newtheorem{lemma}{Lemma}
\newtheorem{corollary}{Corollary}
\newtheorem{definition}{Definition}

\newtheorem{problem}{Problem}

\newtheorem{remark}{Remark}
\newtheorem{claim}{Claim}

\newcommand{\size}[1]{\ensuremath{|#1|}}
\newcommand{\ceil}[1]{\ensuremath{\lceil#1\rceil}}

\newcommand{\lra}[1]{\ensuremath{(#1)}}

\newcommand{\lrc}[1]{\ensuremath{\{#1\}}}

\newcommand{\lrA}[1]{\ensuremath{\left(#1\right)}}
\newcommand{\lrB}[1]{\ensuremath{\left[#1\right]}}
\newcommand{\lrC}[1]{\ensuremath{\left\{#1\right\}}}

\def\OPT{\mbox{OPT}}

\def\T{\mathcal{T}}

\def\S{\mathcal{S}}

\def\RR{\mathbb{R}}

\def\OPT{\mbox{OPT}}

\newcommand{\PP}[1]{\ensuremath{\mbox{Pr}[#1]}}

\newcommand{\EE}[1]{\ensuremath{\mathbb{E}[#1]}}
\newcommand{\EEE}[1]{\ensuremath{\mathbb{E}\lrB{#1}}}

\title{Approximation Algorithms for the Cumulative Vehicle Routing Problem with Stochastic Demands}

\author
{
Jingyang Zhao\\
University of Electronic Science and Technology of China\\
\texttt{jingyangzhao1020@gmail.com}
\and
Mingyu Xiao\\
University of Electronic Science and Technology of China\\
\texttt{myxiao@uestc.edu.cn}
}

\date{}

\begin{document}

\maketitle

\begin{abstract}
In the Cumulative Vehicle Routing Problem (Cu-VRP), we need to find a feasible itinerary for a capacitated vehicle located at the depot to satisfy customers' demand, as in the well-known Vehicle Routing Problem (VRP), but the goal is to minimize the cumulative cost of the vehicle, which is based on the vehicle's load throughout the itinerary. If the demand of each customer is unknown until the vehicle visits it, the problem is called Cu-VRP with Stochastic Demands (Cu-VRPSD). Assume that the approximation ratio of metric TSP is $1.5$. In this paper, we propose a randomized $3.456$-approximation algorithm for Cu-VRPSD, improving the best-known approximation ratio of $6$ (Discret. Appl. Math. 2020). Since VRP with Stochastic Demands (VRPSD) is a special case of Cu-VRPSD, as a corollary, we also obtain a randomized $3.25$-approximation algorithm for VRPSD, improving the best-known approximation ratio of $3.5$ (Oper. Res. 2012). For Cu-VRP, we give a randomized $3.194$-approximation algorithm, improving the best-known approximation ratio of $4$ (Oper. Res. Lett. 2013). Moreover, if each customer is allowed to be satisfied by using multiple tours, we obtain further improvements for Cu-VRPSD and Cu-VRP.
\end{abstract}

\maketitle

\section{Introduction}
In the well-known Vehicle Routing Problem (VRP)~\citep{dantzig1959truck}, we are given an undirected complete graph $G=(V, E)$ with $V=\{v_0,v_1,\dots,v_n\}$, where $v_0$ denotes the depot, and the other $n$ vertices denote $n$ customers. 
Moreover, there is a metric weight function $w$ on the edges representing the length of edges, which satisfies the triangle inequality, and a demand vector $d=(d_1,...,d_n)$ implying that each customer $v_i$ has a demand of $d_i$.
The objective is to determine an \emph{itinerary} for a vehicle with a capacity of $Q$, starting from and ending at the depot, that fulfills every customer's demand while minimizing the total weight of the edges in the itinerary.

In the Cumulative Vehicle Routing Problem (Cu-VRP)~\citep{kara2007energy,kara2008cumulative}, the goal is also to find an itinerary for the vehicle, but with the objective of minimizing the \emph{cumulative cost} of the itinerary. Here, the cumulative cost for the vehicle traveling from $u$ to $v$ carrying a load of $x_{uv}\leq Q$ units of goods is defined as $a\cdot w(u,v)+b\cdot x_{uv}\cdot w(u,v)$, where $a,b\in\mathbb{R}_{\geq 0}$ are given parameters.
Specifically, $a$ represents the cost of moving an empty vehicle per unit distance, and $b$ represents the cost of transporting one unit of goods per unit distance. Moreover, when $a=1$ and $b=0$, Cu-VRP reduces to VRP.

For example, consider an itinerary $T=v_0v_1v_2v_0$, where we let $e_i=v_{i-1}v_{i\bmod 3}$ and $w(e_i)=1$ for each $i\in\{1,2,3\}$. 
Assume that the vehicle picks up $8$ units of goods at the depot and delivers $2$ units each to $v_1$ and $v_2$ along the itinerary, respectively. 
We can obtain $x_{e_1}=8$, $x_{e_2}=6$, $x_{e_3}=4$, and thus the cumulative cost of $T$ is $a\cdot 3 + b\cdot (8+6+4)=3a+18b$. Note that the item $b\cdot x_{e_3}\cdot w(e_3)$, where $e_3$ is the last edge of $T$, represents the additional carrying cost incurred by the vehicle for goods that are picked up but not delivered during the itinerary. 

Cu-VRP captures the fuel consumption in transportation and logistics, as fuel consumption depends on both the weight of the empty vehicle and the weight of the goods being carried by the vehicle~\citep{newman1989estimating,gaur2013routing}. 
In fact, it is a simplified model for fuel consumption with a linear objective function for VRP~\citep{xiao2012development}.
Since fuel consumption can account for as much as 60\% of a vehicle's operational costs~\citep{sahin2009approach}, Cu-VRP has been studied extensively through experimental algorithms~\citep{gaur2017heuristic,wang2016novel,fernandez2023cumulative, fukasawa2016branch,mulati2022arc}. 
Specifically, \citet{gaur2017heuristic} proposed a heuristic algorithm using the column generation method.
\citet{wang2016novel} generalized Cu-VRP into the model with multi-depots. Cu-VRP with time windows was investigated in~\citep{fernandez2023cumulative}, and Cu-VRP with a limitation on the number of used tours was studied in \citep{fukasawa2016branch,mulati2022arc}.
A recent survey of Cu-VRP can be found in~\citep{corona2022vehicle}.

In VRP with Stochastic Demands (VRPSD)~\citep{bertsimas1992vehicle}, the demand of each customer is represented by an independent random variable with a known distribution, and its value is unknown until the vehicle visits the customer. 
The goal is to design a \emph{policy} such that the expected weight of the itinerary is minimized.
This model has natural applications in logistic where the demands of customers are unknown before routing the vehicle, and early surveys on this topic can be found in~\citep{gendreau1996stochastic,bertsimas1996new}.
Cu-VRP with Stochastic Demands (Cu-VRPSD) was proposed in \citep{DBLP:conf/caldam/GaurMS16}, and similarly, the goal is to design a policy such that the expected cumulative cost of the itinerary is minimized. Moreover, Cu-VRPSD reduces to VRPSD when $a=1$ and $b=0$.

We observe that in VRPSD, the vehicle can always be fully loaded before departing from the depot. However, in Cu-VRPSD, due to the fact that both higher and lower loads may result in higher cumulative costs, we need to carefully consider how the vehicle is loaded. This property makes Cu-VRPSD both more challenging and more interesting compared to VRPSD.

In each of the above problems, the \emph{splittable} (resp., \emph{unsplittable}) variant requires that the demand of each customer can be satisfied partially within the vehicle's visits (resp., must be satisfied entirely in one of the vehicle's visits).
Usually, the unsplittable variants are more difficult. For example, unsplittable VRP generalizes the bin packing problem even on a line shape graph~\citep{wu2022capacitated}, and thus cannot be approximated with an approximation ratio of less than 1.5 unless P=NP. In contrast, while splittable VRP is known to be APX-hard~\citep{AsanoKTT97+}, the best-known inapproximability bound is significantly smaller than 1.5.

\subsection{Related Works}
We mainly focus on the design of \emph{approximation algorithms}, which can generate solutions with a theoretical guarantee in polynomial time. 
For Cu-VRP, an algorithm is called a $\rho$-approximation algorithm if it can output a solution with a cumulative cost of at most $\rho\cdot\OPT$ in polynomial time, where $\OPT$ is the cumulative cost of the optimal solution. 
For Cu-VRPSD, an algorithm is called a $\rho$-approximation algorithm if it can employ a policy to obtain a solution with an expected cumulative cost of at most $\rho\cdot\OPT$ in polynomial time, where $\OPT$ is the cumulative cost of the minimum expected cumulative cost solution obtained by the optimal policy.

Let $\alpha$ denote the approximation ratio of the metric Traveling Salesman Problem (TSP). It is well-known that $\alpha\leq1.5$~\citep{christofides1976worst,serdyukov1978some}, which was slightly improved to $\alpha\leq1.5-10^{-36}$~\citep{KarlinKG21,DBLP:conf/ipco/KarlinKG23}. 
The approximation ratio $\alpha$ for TSP will be frequently used in VRP related problems.

\textbf{VRP(SD).}
For VRP, there is an $(\alpha+1)$-approximation algorithm for the splittable case~\citep{HaimovichK85}, and an $(\alpha+2)$-approximation algorithm for the unsplittable case~\citep{altinkemer1987heuristics}.
Recently, \citet{blauth2022improving} improved the approximation ratio to $\alpha+1-\varepsilon$ for the splittable case, and to $\alpha+2-2\varepsilon$ for the unsplittable case, 
and \citet{uncvrp} further improved the approximation ratio to $\alpha+1+\ln2-\varepsilon'$ for the unsplittable case using the LP rounding method, where $\varepsilon$ and $\varepsilon'$ are small positive constants related to $\alpha$. Notably, \citet{uncvrp} also gave a combinatorial $(\alpha+1.75-\varepsilon')$-approximation algorithm for the unsplittable case.
For VRPSD, \citet{bertsimas1992vehicle} proposed a randomized $(\alpha+1+o(1))$-approximation algorithm for the splittable case with identical demand distributions, and a randomized $(\alpha+Q)$-approximation algorithm for the general case. 
Later, \citet{DBLP:journals/ior/GuptaNR12} obtained a randomized $(\alpha+1)$-approximation algorithm for the splittable case, and a randomized $(\alpha+2)$-approximation algorithm for the unsplittable case. 

\textbf{Cu-VRP(SD).}
For Cu-VRP, \citet{gaur2013routing} proposed a $(1+\frac{2\alpha}{\sqrt{\alpha^2+4\alpha}-\alpha})$-approximation algorithm for the splittable case, and a $(1+\frac{2\alpha}{\sqrt{\alpha^2+6\alpha+1}-(\alpha+1)})$-approximation algorithm for the unsplittable case.
For Cu-VRPSD, \citet{DBLP:conf/caldam/GaurMS16} proposed a $2(1+\alpha)$-approximation algorithm for the splittable case, and a $7$-approximation algorithm for the unsplittable case, and later, they~\citep{gaur2020improved} further improved these approximation ratios to $\max\{1+\frac{3}{2}\alpha,3\}$ and $\max\{2+\frac{3}{2}\alpha,6\}$, respectively.

\subsection{Our Results}
In this paper, we design improved approximation algorithms for Cu-VRPSD, VRPSD, and Cu-VRP. A summary of our results under $\alpha=1.5$ can be found in Table~\ref{res}.

\begin{table}[ht]
\centering
\small
\begin{tabular}{c|c|c|c}
\hline
\multirow{2}{*}{\textbf{Problem}} & \multicolumn{2}{c|}{\textbf{Approximation Ratio}} & \multirow{2}{*}{\textbf{Reference}}\\
\cline{2-3}
&\textbf{Unsplittable}& \textbf{Splittable}&\\
\hline
\multirow{4}{*}{Cu-VRPSD} & $7$ & $5$ & \citet{DBLP:conf/caldam/GaurMS16}\\
& $6$ & $3.25$ & \citet{gaur2020improved}\\
\cline{2-4}
& $\alpha+2=\bf{3.5}$ (Thm.~\ref{unttt0}) & \multirow{2}{*}{$\alpha+1=\bf{2.5}$ (Thm.~\ref{spttt0})}   & \multirow{2}{*}{\textbf{This Paper}}\\
& $\bf{3.456}$ (Thm.~\ref{th23}) & & \\
\hline
\multirow{2}{*}{VRPSD} & $3.5$ & $2.5$ & \citet{DBLP:journals/ior/GuptaNR12}\\
\cline{2-4}
 & $\alpha+1.75=\bf{3.25}$ (Cor.~\ref{coro1}) & - & \textbf{This Paper}\\
\hline
\multirow{2}{*}{Cu-VRP} & $4$ & $3.186$ & \citet{gaur2013routing}\\
\cline{2-4}
& $\alpha+1+\ln2+\varepsilon<\bf{3.194}$ (Thm.~\ref{th7}) & $\alpha+1=\bf{2.5}$ (Cor.~\ref{spttt00}) & \textbf{This Paper}\\
\hline

\end{tabular}
\caption{A summary of the previous approximation ratios and our approximation ratios under $\alpha=1.5$.}
\label{res}
\end{table}

We mainly focus on the unsplittable variants, as we will show that by extending the methods used in the unsplittable variants, we can easily obtain the results for the splittable variants. 

\textbf{Unsplittable Cases.}
The main idea of the most recent algorithms~\citep{DBLP:journals/ior/GuptaNR12, gaur2020improved} is as follows. First, we find an \(\alpha\)-approximate TSP tour and then the vehicle satisfies customers in the order they appear on the TSP tour. Once the load is less than the serving customer's demand, the vehicle goes back to the depot to reload.
Our algorithms will also use an \(\alpha\)-approximate TSP tour and visit the customers in the order according to the TSP tour. However, we do not strictly satisfy the customers in the order. To reduce the cumulative cost, our vehicle may skip customers with large demands when visiting customers according to the TSP tour (but record their demands) and satisfy them after completing the TSP tour.

Based on the above idea, we propose two novel algorithms for unsplittable Cu-VRPSD, denoted as $ALG.1(\lambda, \delta)$ and $ALG.2(\lambda, \delta)$. 
\begin{itemize}
    \item In $ALG.1(\lambda, \delta)$, the vehicle will skip customers in $\{v_i\mid d_i>\lambda\cdot Q\}$ and then satisfy each of them by using a single tour; 
    \item In $ALG.2(\lambda, \delta)$, the vehicle will skip customers in $\{v_i\mid d_i>\delta\cdot Q\}$ and then satisfy them either by using a single tour for each or by calling an algorithm for \emph{weighted set cover}.
\end{itemize}

Furthermore, in our algorithms, we set upper and lower bounds of the load of the vehicle when traveling along the TSP tour: the load is at least $\delta\cdot Q$ and less than $\lambda\cdot Q$ for some parameters $\delta$ and $\lambda$. The lower bound can be regarded as the \emph{backup} goods that the vehicle carries. 
The idea of carrying some backup goods was inspired by a partition algorithm for the TSP tour used for unsplittable VRP~\citep{uncvrp}.
We will show that this approach can reduce the potential cumulative cost caused by visiting customers with demands at most $\delta\cdot Q$ at the expense of increasing the cumulative cost of the vehicle when traveling the TSP tour. So, we need to balance the setting of $\delta$, e.g., we may set $\delta=0$ when $a/b$ is small. 

We will first prove that even setting $\delta=0$, $ALG.1(\lambda, \delta)$ can be used to obtain a randomized $(\alpha+2)$-approximation algorithm for unsplittable Cu-VRPSD, which matches the randomized $(\alpha+2)$-approximation algorithm for unsplittable VRPSD~\citep{DBLP:journals/ior/GuptaNR12}. 

Then, by carefully optimizing the setting of $\delta$, we will prove that $ALG.1(\lambda, \delta)$ can be used to obtain a randomized algorithm with an expected approximation ratio of $10/3$ for $a/b\leq 0.375$ and $3.456$ for $0.375<a/b\leq 1.444$. 
Moreover, by using both $ALG.1(\lambda, \delta)$ and $ALG.2(\lambda, \delta)$, we can obtain a randomized $3.456$-approximation algorithm for $a/b>1.444$. 
Hence, we obtain a randomized $3.456$-approximation algorithm for unsplittable Cu-VRPSD.

Note that Cu-VRPSD reduces to VRPSD when $b=0$, and this corresponds to the case where $a/b=\infty$. As a corollary, for unsplittable VRPSD, we obtain a randomized $(\alpha+1.75)$-approximation algorithm using the randomized $3.456$-approximation algorithm for unsplittable Cu-VRPSD with $a/b>1.444$.

For unsplittable Cu-VRP, we also give two algorithms, denoted as $ALG.3(\lambda, \delta)$ and $ALG.4(\lambda)$. Since the demands of customers are known in advance, in $ALG.3(\lambda, \delta)$, we first obtain a set of tours by applying the randomized rounding method to the LP of weighted set cover, and then satisfy the remaining customers by calling $ALG.1(\lambda, \delta)$; in $ALG.4(\lambda)$, we directly call $ALG.1(\lambda, 0)$. In the tours obtained by calling $ALG.1(\lambda, \delta)$ and $ALG.1(\lambda, 0)$, the load of the vehicle may be greater than the delivered units of goods. So, we also adapt a pre-optimization step to ensure that the load of the vehicle equals the delivered units of goods, which may reduce some cumulative cost. 

We will show that for any constant $\varepsilon>0$, $ALG.3(\lambda, \delta)$ can be used to obtain a randomized $(\alpha+1+\ln2+\varepsilon)$-approximation algorithm with a running time of $n^{O(\frac{1}{\min\{a/b,1\}})}$, and thus it only works for $a/b>\gamma_0$, where $\gamma_0>0$ is any fixed constant. Moreover, $ALG.4(\lambda, \delta)$ can be used to obtain a randomized $(\alpha+1+\ln2-0.029)$-approximation algorithm for $a/b<0.285$. Hence, by combining these two algorithm, we obtain a randomized $(\alpha+1+\ln2+\varepsilon<3.194)$-approximation algorithm for unsplittable Cu-VRP.



\textbf{Splittable Cases.} By slightly modifying $ALG.1(\lambda, \delta)$ with $\delta=0$ into $ALG.S(\lambda)$, we will prove that $ALG.S(\lambda)$ can be used to obtain a randomized $(\alpha+1)$-approximation algorithm for splittable Cu-VRPSD, which matches the randomized $(\alpha+1)$-approximation algorithm for splittable VRPSD~\citep{DBLP:journals/ior/GuptaNR12}. 
Since Cu-VRP is a special case of Cu-VRPSD, as a corollary, we obtain a randomized $(\alpha+1)$-approximation algorithm for splittable Cu-VRP, which almost matches the $(\alpha+1-\varepsilon)$-approximation algorithm for splittable VRP~\citep{blauth2022improving}.

Although our algorithms are simple and neat, the analysis is technically involved. Some parts also need careful calculation.
To avoid distraction from our main discussions, the proofs of lemmas and theorems marked with `*' are omitted and they can be found in Appendix~\ref{omitted}. 



\section{Notations}
In Cu-VRP, we use $G=(V, E)$ to denote the input complete graph, where $V=\{v_0,\dots,v_n\}$.
There is a non-negative weight function $w: E\to \mathbb{R}_{\geq0}$ on the edges, where $w(u,v)$ denotes the length of edge $uv\in E$. We assume that $w$ is a \emph{metric}, i.e., it is symmetric and satisfies the triangle inequality.
Let $V'\coloneqq V\setminus\{v_0\}$. 
There is also a demand vector $d=(d_1,...,d_n)\in\mathbb{R}^{V'}_{[0,Q]}$, where $Q\in\mathbb{R}_{>0}$ is the capacity of the vehicle, and each customer $v_i$ has a required demand $d_i\in[0,Q]$.
We let $l_i\coloneqq w(v_0,v_i)$ and $[i]\coloneqq\{1,2,...,i\}$.

In Cu-VRPSD, the demand of each customer $v_i$ is represented by an independent random variable $\chi_i\in[0,Q]$, where the distribution of $\chi_i$ is usually assumed to be known in advance~\citep{bertsimas1992vehicle,DBLP:journals/ior/GuptaNR12}. Let $\chi=(\chi_1,...,\chi_n)$, where we assume that $\chi_i$ is not identically zero, as $v_i$ can be ignored in such a case. Consequently, any feasible policy must visit every customer at least once~\citep{DBLP:journals/ior/GuptaNR12}.


For any random variable $L$, we use $L\sim U[l, r)$ to indicate that $L$ is uniformly distributed over the interval $[l,r)$, where $l<r$.


A \emph{tour} $T=v_0v_{1}\dots v_{i}v_{0}$ is a directed simple cycle, which always contains the depot \(v_0\).
We use $E(T)$ to denote the (multi-)set of edges on $T$, and $V'(T)$ to denote the set of customers on $T$. 
Assume that the vehicle carries a load of $x_{eT}$ units of goods when traveling along $e\in E(T)$. 
The cumulative cost of $T$ is 
\[
Cu(T)\coloneqq a\cdot \sum_{e\in E(T)}w(e)+b\cdot \sum_{e\in E(T)}x_{eT}\cdot w(e),
\]
where $w(T)\coloneqq \sum_{e\in E(T)}w(e)$ is called the \emph{weight} of $T$, $Cu_1(T)\coloneqq a\cdot \sum_{e\in E(T)}w(e)$ is called the \emph{vehicle cost} of $T$, and $Cu_2(T)\coloneqq b\cdot \sum_{e\in E(T)}x_{eT}\cdot w(e)$ is called the \emph{cargo cost} of $T$. An itinerary $\T$ is a set of tours.
A \emph{TSP tour} is an undirected cycle that includes all customers and the depot exactly once.
The weight of a minimum weight TSP tour is denoted by $\tau$. 

\subsection{Problem Definitions}
\begin{definition}[\textbf{Cu-VRPSD}]
Given a complete graph $G=(V, E)$, a metric weight function $w$, a vehicle capacity $Q\in\mathbb{R}_{>0}$, a random demand variable vector $\chi=(\chi_1,...,\chi_n)$, and two parameters $a,b\in\RR_{\geq 0}$, we need to design a policy to find a feasible itinerary $\T$ such that
\begin{itemize}
\item the vehicle carries at most $Q$ units of goods on each tour $T\in \T$,
\item the vehicle delivers goods to customers only in $V'(T)$ on each tour $T\in \T$,
\item the sum of the delivered demand over all tours for each $v_i\in V'$ equals the demand of $v_i$, 
\end{itemize}
and $\EE{Cu(\T)}$ is minimized.
\end{definition}
Note that the demand of each customer is unknown until the vehicle visits it. 

As mentioned, in \emph{splittable} Cu-VRPSD, each customer is allowed to be satisfied by using multiple tours, and in \emph{unsplittable} Cu-VRPSD, each customer must be satisfied by using only one tour: each customer may be included in multiple tours, but its demand must be satisfied entirely within exactly one of those tours. Clearly, unsplittable Cu-VRPSD admits a feasible solution only if it holds $d_i\leq Q$ for any $i\in [n]$.

In Cu-VRP, we have $\chi=d$, where $d$ is known in advance.
Moreover, by scaling each customer’s demand $\chi_i$ to $\chi_i/Q$ and adjusting the parameter $b$ to $b\cdot Q$, without loss of generality, we assume that $Q=1$.

\subsection{The Lower Bounds}
To analyze approximation algorithms, we recall the following lower bound for Cu-VRPSD.
\begin{lemma}[\citet{gaur2020improved}]\label{lb-}
For Cu-VRPSD, it holds that $\EE{Cu(\T^*)}\geq a\cdot\max\lrc{\tau,\sum_{i\in[n]}2\cdot\EE{\chi_i}\cdot l_i}+b\cdot\sum_{i\in[n]}\EE{\chi_i}\cdot l_i$, where $\T^*$ is any optimal itinerary.
\end{lemma}

When $a=1$ and $b=0$, the lower bound in Lemma~\ref{lb-} becomes $\max\lrc{\tau,\sum_{i\in[n]}2\cdot\EE{\chi_i}\cdot l_i}$, and it was used in analyzing approximation algorithms for VRPSD in~\citep{DBLP:journals/ior/GuptaNR12}.
To analyze our algorithms, we use a stronger lower bound that was implicitly used in the proof of Lemma~\ref{lb-}.
\begin{lemma}[\citet{gaur2020improved}]\label{lb}
For unsplittable Cu-VRPSD with any demand realization vector $d\in\mathbb{R}^{V'}_{[0,1]}$, it holds that
$\EE{Cu(\T^*)\mid \chi=d }\geq LB\coloneqq a\cdot\max\{\tau,\ \eta\}+b\cdot0.5\cdot\eta$, where $\eta\coloneqq\sum_{i\in[n]}2\cdot d_i\cdot l_i$.   
\end{lemma}

Lemma~\ref{lb} is stronger than Lemma~\ref{lb-} since it holds that $\EE{\max\{X,Y\}}\geq \max\{\EE{X},\EE{Y}\}$ for any random variables $X$ and $Y$ by the property of the maximum function.

\begin{lemma}\label{mid}
An algorithm is a $\rho$-approximation algorithm for Cu-VRPSD if, for any possible demand realization vector $d\in\mathbb{R}^{V'}_{[0,1]}$, the algorithm conditioned on $\chi=d$ outputs a solution $\T$ with a cumulative cost of $\EE{Cu(\T)\mid \chi=d}\leq \rho\cdot LB$.
\end{lemma}
\begin{proof}
Since it holds that $\EE{Cu(\T)\mid \chi=d}\leq \rho\cdot LB\leq\rho\cdot\EE{Cu(\T^*)\mid \chi=d }$ for any possible demand realization vector $d$, we have $\EE{Cu(\T)\mid \chi}\leq\rho\cdot\EE{Cu(\T^*)\mid \chi}$. Therefore, we have
$\EE{Cu(\T)}=\EE{\EE{Cu(\T)\mid \chi}}\leq\rho\cdot \EE{\EE{Cu(\T^*)\mid \chi}}=\rho\cdot\EE{Cu(\T^*)}$.
\end{proof}

Hence, we may frequently analyze our algorithms conditioned on $\chi=d$, where $d\in\mathbb{R}^{V'}_{[0,1]}$ is any possible demand realization.
For the sake of analysis, we let $\gamma\coloneqq a/b$, and $\sigma\coloneqq \gamma/\eta$. Note that $b=0$ corresponds to the case where $\gamma=\infty$, which turns out to be easier, as will be shown in Corollary~\ref{coro1}. 
We also define 
\begin{equation}\label{int}
\int^r_lx^{t}dF(x)\coloneqq \frac{\sum_{v_i\in V':l<d_i\leq r}2\cdot d^t_i\cdot l_i}{\sum_{v_i\in V'}2\cdot d_i\cdot l_i},\quad\quad\mbox{where}\quad t\in\{0,1,2\}.
\end{equation}
Note that $\int^1_0xdF(x)=1$. Moreover, for any $t\in\{1,2\}$ and $0\leq l\leq r$, we have 
\begin{equation}\label{intineq}
l\cdot\int^r_lx^{t-1}dF(x)<\int^r_lx^{t}dF(x)\leq r\cdot\int^r_lx^{t-1}dF(x).
\end{equation}

We remark that for (Cu-)VRPSD, if the vehicle travels along an $\alpha$-approximate TSP tour with an empty carry to record each customer's demand, and then satisfies each customer by calling a $\rho$-approximation algorithm for (Cu-)VRP, by Lemma~\ref{lb}, this strategy yields an $(\alpha+\rho)$-approximation algorithm. Therefore, we have the following result.

\begin{theorem}\label{thm0}
Given any randomized/deterministic $\rho$-approximation algorithm for (Cu-)VRP, there is a randomized/deterministic $(\alpha+\rho)$-approximation algorithm for (Cu-)VRPSD.
\end{theorem}

Note that the method in Theorem~\ref{thm0} may only be useful to obtain some straightforward results.

\section{Two Algorithms for Unsplittable Cu-VRPSD} 

\subsection{The First Algorithm}
In this section, we will introduce our first algorithm, denoted as $ALG.1(\lambda, \delta)$, which can be used to obtain an $(\alpha+2)$-approximation algorithm for unsplittable Cu-VRPSD. 
Moreover, by appropriately randomizing the setting of $\lambda$, it can be used to obtain a $10/3$-approximation algorithm for unsplittable Cu-VRPSD with any $\gamma\in(0,0.375]$ and a $3.456$-approximation algorithm for unsplittable Cu-VRPSD with any $\gamma\in[0.375, 1.444]$. 
Here, $\lambda\in(0,1]$ and $\delta\in[0,\lambda/2]$ are parameters that will be defined later. The main ideas of $ALG.1(\lambda, \delta)$ are as follows.

Firstly, $ALG.1(\lambda,\delta)$ computes an $\alpha$-approximate TSP tour $T^*$, which will be oriented in either clockwise or counterclockwise direction. 
Assume that $T^*=v_0v_1\dots v_nv_0$ by renumbering the customers following the orientation.
Then, the vehicle in $ALG.1(\lambda,\delta)$ tries to satisfy the customers in the order of $v_1\dots v_n$ as they appear on $T^*$, where the parameters $\lambda$ and $\delta$ ensures that the load of the vehicle during its travel on each edge of $T^*$ is at least $\delta$ and less than $\lambda$. Moreover, among its load, the $\delta$ units of goods are regarded as \emph{backup} goods, and the other units of goods are regarded as \emph{normal} goods. 
Specifically, if the vehicle carries $L_{i-1}$ units of normal goods during its travel from $v_{i-1}$ to $v_i$, we have $0 \leq L_{i-1} < \lambda - \delta$ for each $i \in [n+1]$. We say that the vehicle carries $S_{i-1}=(L_{i-1}, \delta)$ units of goods to indicate that it carries $L_{i-1}$ units of normal goods and $\delta$ units of backup goods. We require that $0<\lambda\leq 1$ and $0\leq\delta\leq \lambda-\delta$, i.e., $0\leq \delta\leq \lambda/2$.

When serving a customer, the main strategy of $ALG.1(\lambda,\delta)$ is to prioritize using the normal goods first and then consider using the backup goods if the normal goods are insufficient. Note that if the backup goods are used, the vehicle must return to the depot to reload before continuing since we require that the vehicle carries at least $\delta$ units of backup goods, as it is required to carry at least $\delta$ units of backup goods while traveling on each edge of $T^*$.
Conditioned on $\chi=d$, the details are as follows.

Initially, we load the vehicle with $S_0=(L_0,\delta)$ units of goods at the depot, where $L_{0}\sim U[0,\lambda-\delta)$, i.e., $L_0$ is uniformly distributed over the interval $[0,\lambda-\delta)$. When the vehicle is about to serve $v_i$, we assume that it carries $S_{i-1}=(L_{i-1},\delta)$ units of goods, where $0\leq L_{i-1}<\lambda-\delta$. Then, we have the following three cases.

\textbf{Case~1: $d_i\leq L_{i-1}$.} In this case, the vehicle directly delivers $(d_i,0)$ units of goods for $v_i$, and then goes to the next customer. Hence, we have $S_i=(L_i, \delta)$, where $L_i\coloneqq L_{i-1}+\ceil{\frac{d_i-L_{i-1}}{\lambda-\delta}}\cdot (\lambda-\delta)-d_i= L_{i-1}- d_i$ since $d_i\leq L_{i-1}<\lambda-\delta$.

\textbf{Case~2: $L_{i-1}<d_i\leq L_{i-1}+\delta$.} The vehicle delivers $(L_{i-1},d_i-L_{i-1})$ units of goods for $v_i$, goes to the depot to reload $(L_{i-1}+\ceil{\frac{d_i-L_{i-1}}{\lambda-\delta}}\cdot (\lambda-\delta)-d_i,d_i-L_{i-1})$ units of goods, and then goes to the next customer. Hence, we have $S_i=(L_i,\delta)$, where $L_i\coloneqq L_{i-1}+\ceil{\frac{d_i-L_{i-1}}{\lambda-\delta}}\cdot (\lambda-\delta)-d_i=L_{i-1}+(\lambda-\delta)-d_i$ since $0<d_i-L_{i-1}\leq \delta\leq \lambda-\delta$.

\textbf{Case~3: $L_{i-1}+\delta<d_i\leq 1$.} In this case, we must have $d_i>\delta$. 

\textbf{Case~3.1: $\delta< d_i\leq \lambda$.}
The vehicle goes to the depot to reload $(d_i-L_{i-1}-\delta, 0)$ units of goods, goes to satisfy $v_i$, then goes to the depot to reload $(L_{i-1}+\ceil{\frac{d_i-L_{i-1}}{\lambda-\delta}}\cdot (\lambda-\delta)-d_i,\delta)$ units of goods, goes to customer $v_i$ again, and then goes to the next customer. Hence, we have $S_i=(L_i,\delta)$, where $L_i\coloneqq L_{i-1}+\ceil{\frac{d_i-L_{i-1}}{\lambda-\delta}}\cdot (\lambda-\delta)-d_i$. 
Note that, by the triangle inequality, the vehicle may directly proceed to the next customer instead of returning to $v_i$ without increasing the traveling distance. However, for the sake of analysis, we still require the vehicle to return to $v_i$.

\textbf{Case~3.2: $\lambda<d_i\leq 1$.} Since $L_{i-1}<\lambda-\delta$, we must have $L_{i-1}+\delta<d_i$. Instead of satisfying $v_i$ by returning to the depot to reload as in Case~3.1, the vehicle records its demand, skips it, and goes to the next customer. Hence, we have $S_i=(L_i,\delta)$, where $L_i\coloneqq L_{i-1}$.

An illustration of the vehicle's load when serving customer $v_i$ in each of the four cases can be found in Figure~\ref{fig1}.

\begin{figure}[ht]
    \centering
    \includegraphics[scale=0.48]{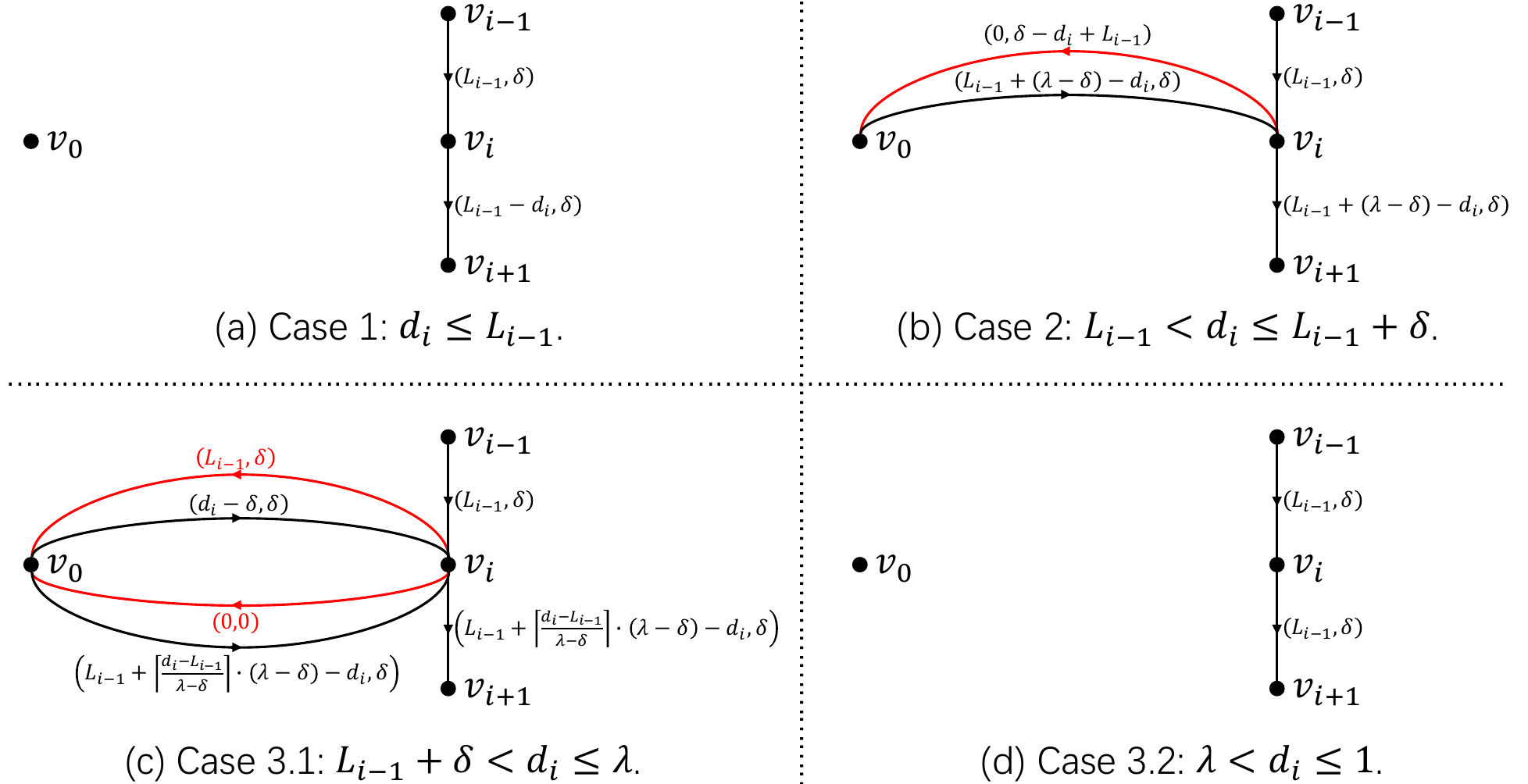}
    \caption{An illustration of the vehicle's load when serving customer $v_i$ in each of the four cases.}
    \label{fig1}
\end{figure}

After trying to satisfy all customers using the above strategy,  due to Case~3.2, there may be still a set of \emph{unsatisfied} customers $\{v_i\mid d_i>\lambda\}$. Then, for each such customer $v_i$, since its demand has been recorded, the vehicle will load exactly $d_i$ units of goods at the depot, go to satisfy $v_i$, and then return to the depot.

The details of $ALG.1(\lambda,\delta)$ is shown in Algorithm~\ref{alg1}.

\begin{algorithm}[ht]
\caption{An algorithm for unsplittable Cu-VRPSD ($ALG.1(\lambda,\delta)$)}
\label{alg1}
\small
\vspace*{2mm}
\textbf{Input:} An instance of unsplittable Cu-VRPSD, and two parameters $\lambda\in(0,1]$ and $\delta\in[0,\lambda/{2}]$.\\
\textbf{Output:} A feasible solution $\T$ to unsplittable Cu-VRPSD.

\begin{algorithmic}[1]
\State\label{tsptour} Obtain an $\alpha$-approximate TSP tour $T^*$ using an $\alpha$-approximation algorithm for metric TSP, orient $T^*$ in either clockwise or counterclockwise direction, and denote $T^*=v_0v_1v_2\dots v_n v_0$ by renumbering the customers following the direction.

\State Load the vehicle with $S_0\coloneqq (L_0,\delta)$ units of goods, including $L_0$ units of normal goods and $\delta$ units of backup goods, where $L_0\sim U[0,\lambda-\delta)$.

\State Initialize $i\coloneqq 1$ and $V^*\coloneqq\emptyset$.

\While{$i\leq n$}

\State Go to customer $v_i$.

\If {$d_{i}\leq L_{i-1}$}
\State\label{case1ofalg1} Deliver $(d_i,0)$ units of goods to $v_i$, and update $S_i\coloneqq(L_i,\delta)$, where $L_{i}\coloneqq L_{i-1}+\ceil{\frac{d_i-L_{i-1}}{\lambda-\delta}}\cdot (\lambda-\delta)-d_i= L_{i-1}-d_i$.

\ElsIf{$L_{i-1}<d_i\leq L_{i-1}+\delta$}
\State\label{case1ofalg1+}\label{addone} Deliver $(L_{i-1},d_i-L_{i-1})$ units of goods to $v_i$, goes to the depot, load the vehicle with $(L_{i-1}+\ceil{\frac{d_i-L_{i-1}}{\lambda-\delta}}\cdot (\lambda-\delta)-d_i,d_i-L_{i-1})$ units of goods, and update $S_i\coloneqq(L_i,\delta)$, where $L_{i}\coloneqq L_{i-1}+\ceil{\frac{d_i-L_{i-1}}{\lambda-\delta}}\cdot (\lambda-\delta)-d_i= L_{i-1}+(\lambda-\delta)-d_i$.

\ElsIf{$L_i+\delta<d_i\leq 1$}
\If{$\delta<d_i\leq\lambda$}
\State\label{addtwo1} Return to the depot, load the vehicle with $(d_i-L_{i-1}-\delta,0)$ units of goods, go to customer $v_i$, and deliver $(d_i-\delta,\delta)$ units of goods to $v_i$.
\State\label{addtwo2}\label{case1ofalg1++} Return to the depot, load the vehicle with $(L_{i-1}+\ceil{\frac{d_i-L_{i-1}}{\lambda-\delta}}\cdot (\lambda-\delta)-d_i,\delta)$ units of goods, go to customer $v_i$, and update $S_i\coloneqq(L_i,\delta)$, where $L_{i}\coloneqq L_{i-1}+\ceil{\frac{d_i-L_{i-1}}{\lambda-\delta}}\cdot (\lambda-\delta)-d_i$. \Comment{
The vehicle returns to $v_i$ for the sake of analysis; however, it could directly proceed to $v_{i+1}$.}

\ElsIf{$\lambda<d_i\leq 1$}
\State\label{case2ofalg1} Record $v_i$'s demand, and update $V^*\coloneqq V^*\cup \{v_i\}$ and $S_i\coloneqq(L_i,\delta)$, where $L_i\coloneqq L_{i-1}$.
\EndIf
\EndIf
\State $i \coloneqq i+1$.
\EndWhile

\State Go to the depot.

\For{$v_i\in V^*$}
\State\label{clean} Load the vehicle with $d_i$ units of goods, go to customer $v_i$, and deliver $d_i$ units of goods to $v_i$.
\State Go to the depot.
\EndFor
\end{algorithmic}
\end{algorithm}

Compared to the previous strategy in~\citep{gaur2020improved}, there are two main differences. The first is that we specially handle each customer $v_i$ with $d_i>\lambda$ in Case~3.2. Note that if we satisfy $v_i$ as the method in Case~3.1, since $L_{i-1}+\delta<\lambda$ (we will prove it in Lemma~\ref{l0}), the vehicle must incur two visits to the depot, which will cost too much. 
The second is that we ensure that the vehicle always carries $\delta$ units of backup goods when traveling along the TSP tour. The advantage is that each customer $v_i$ with $d_i\leq \delta$ incurs at most one visit to the depot while if $\delta=0$, every customer $v_i$ with $d_i\leq \lambda$ has the potential to incur two visits to the depot. However, since $\delta>0$ clearly increases the cumulative cost of the vehicle when it travels along the TSP tour, we need to carefully set the value of $\delta$. 

We also remark that the previous algorithm in~\citep{gaur2020improved} orients $T^*$ by choosing one of the two possible orientations uniformly at random, and its analysis uses this property. However, this randomness is not necessary in our analysis, which is also why $T^*$ can be oriented in either clockwise or counterclockwise direction in Line~\ref{tsptour}.

Although we require that $\lambda>0$ in $ALG.1(\lambda,\delta)$, it can be extended to the case of $\lambda=0$. In this scenario, the vehicle simply travels along the TSP tour with an empty carry to record each customer's demand, and then satisfies each customer within a single tour, as described in Case~3.2. Interestingly, if \( a = 0 \), this algorithm becomes an exact algorithm for unsplittable Cu-VRPSD, as the cumulative cost is \(b \cdot \sum_{i\in[n]}d_i \cdot l_i \), which matches the lower bound $LB$ in Lemma~\ref{mid}. 
The running time can reach $O(n)$ since all TSP tours have the same performance.
However, it may be useless for $a>0$. Hence, we mainly focus on $\lambda>0$ in the following.

\begin{lemma}\label{a=0}
Unsplittable Cu-VRPSD with $a=0$ can be solved in $O(n)$ time.
\end{lemma}

\subsubsection{The analysis}
Recall that $ALG.1(\lambda,\delta)$ carries $L_0\sim U[0,\lambda-\delta)$ units of normal goods initially. 
Next, we analyze the expected cumulative cost of $\T$ conditioned on $\chi=d$, i.e., $\EE{Cu(\T)\mid \chi=d}$.

In $ALG.1(\lambda,\delta)$, the vehicle carries $S_{i-1}=(L_{i-1},\delta)$ units of goods when traveling along the edge $v_{i-1}v_{i\bmod (n+1)}$ of the TSP tour $T^*$. For each $i\in[n]$, we let $h_i\coloneqq 1$ if $d_i\leq\lambda$, and $h_i\coloneqq 0$ otherwise. We have the following lemma.

\begin{lemma}[*]\label{l0}
For any $i\in [n+1]$, it holds $L_{i-1}=L_0 + \ceil{\frac{\sum_{j=1}^{i-1}h_j\cdot d_j-L_0}{\lambda-\delta}}\cdot (\lambda-\delta)-\sum_{j=1}^{i-1}h_j\cdot d_j$, and moreover, $L_{i-1}\sim U[0,\lambda-\delta)$, conditioned on $\chi=d$.
\end{lemma}

Lemma~\ref{l0} also implies that $0\leq L_{i-1}<\lambda-\delta$ for any $i\in[n+1]$. 

\begin{lemma}\label{l1}
In $ALG.1(\lambda,\delta)$, the expected cumulative cost conditioned on $\chi=d$ during the vehicle's travel from $v_{i-1}$ to $v_{i}$ is $a\cdot w(v_{i-1},v_{i})+b\cdot\frac{\lambda+\delta}{2}\cdot w(v_{i-1},v_{i})$.
\end{lemma}
\begin{proof}
By Lemma~\ref{l0}, the vehicle carries $(L_{i-1},\delta)$ units of goods during the vehicle's travel from $v_{i-1}$ to $v_{i}$, where $L_{i-1}\sim U[0,\lambda-\delta)$. So, $\EE{L_{i-1}+\delta\mid \chi=d}=\int_{0}^{\lambda-\delta}\frac{x+\delta}{\lambda-\delta}dx=\frac{\lambda+\delta}{2}$. Hence, the expected cumulative cost of the vehicle's travel from $v_{i-1}$ to $v_{i}$ conditioned on $\chi=d$ is 
$
a\cdot w(v_{i-1},v_{i})+b\cdot\EE{L_{i-1}+\delta\mid \chi=d}\cdot w(v_{i-1},v_{i})=a\cdot w(v_{i-1},v_{i})+b\cdot\frac{\lambda+\delta}{2}\cdot w(v_{i-1},v_{i}).
$
\end{proof}

By Line~\ref{addone} in $ALG.1(\lambda,\delta)$, if the vehicle visits $v_i$ carrying $S_{i-1}=(L_{i-1},\delta)$ units of goods, where $L_{i-1}<d_i\leq L_{i-1}+\delta$, it will first satisfy $v_i$, then proceed to the depot to reload some units of goods, and finally return to the place of $v_i$. We refer to this process as \emph{one additional visit to $v_0$}. 

By Lines~\ref{addtwo1} and \ref{addtwo2} in $ALG.1(\lambda,\delta)$, if the vehicle visits $v_i$ with $d_i\leq \lambda$ carrying $S_{i-1}=(L_{i-1},\delta)$ units of goods, where $L_{i-1}+\delta<d_i$, it will first go to the depot to reload some units of goods, then go to the place of $v_i$ to satisfy $v_i$, proceed to the depot to reload some units of goods, and finally return to the place of $v_i$ again. We refer to this process as \emph{two additional visits to $v_0$}.

When the vehicle is about to serve $v_i$ with $d_i \leq \lambda$, it may incur one additional visit or two additional visits to $v_0$, resulting in some cumulative cost. 
For each customer $v_i$ with $d_i>\lambda$, By Line~\ref{clean}, the vehicle satisfies $v_i$ using a single tour, which will also be regarded as one additional visit to $v_0$ for the sake of presentation. 
Next, we analyze the expected cumulative cost conditioned on $\chi=d$ due to the possible additional visit(s) for each customer $v_i$.

\begin{lemma}[*]\label{l2}
Conditioned on $\chi=d$, when serving each customer $v_i$ in $ALG.1(\lambda,\delta)$, the expected cumulative cost of the vehicle due to the possible additional visit(s) to $v_0$ is
\begin{itemize}
    \item $a\cdot\frac{2d_i}{\lambda-\delta}\cdot l_i+b\cdot \frac{(\lambda+\delta)\cdot d_i-d^2_i}{\lambda-\delta}\cdot l_i$ if $d_i\leq\delta$;
    \item $a\cdot\frac{4d_i-2\delta}{\lambda-\delta}\cdot l_i+b\cdot \frac{d^2_i+(\lambda-\delta)\cdot d_i}{\lambda-\delta}\cdot l_i$ if $\delta<d_i\leq\lambda-\delta$;
    \item $a\cdot\frac{2d_i+2\lambda-4\delta}{\lambda-\delta}\cdot l_i+b\cdot \frac{2d^2_i-(\lambda+\delta)\cdot d_i+\lambda^2-\delta^2}{\lambda-\delta}\cdot l_i$ if $\lambda-\delta<d_i\leq\lambda$;
    \item $a\cdot2\cdot l_i+b\cdot d_i\cdot l_i$ if $\lambda<d_i\leq 1$.
\end{itemize}
\end{lemma}

\begin{lemma}\label{t1}
For unsplittable Cu-VRPSD with any $\lambda\in(0,1]$ and $\delta\in[0,\lambda/{2}]$, conditioned on $\chi=d$, $ALG.1(\lambda,\delta)$ generates a solution $\T$ with an expected cumulative cost of
\[
\frac{A+B}{\gamma\cdot\max\lrC{\sigma,\ 1}+0.5}\cdot LB,
\]
where $A\coloneqq\gamma\cdot\lra{\alpha\cdot\sigma+\int^{\delta}_0\frac{x}{\lambda-\delta}dF(x)+\int^{\lambda-\delta}_{\delta}\frac{2x-\delta}{\lambda-\delta}dF(x)+\int^\lambda_{\lambda-\delta}\frac{x+\lambda-2\delta}{\lambda-\delta}dF(x)+\int_\lambda^11dF(x)}$, and $B\coloneqq\lra{\frac{\lambda+\delta}{2}\cdot\alpha\cdot\sigma+\int^{\delta}_0\frac{(\lambda+\delta)x-x^2}{2(\lambda-\delta)}dF(x)+\int^{\lambda-\delta}_{\delta}\frac{x^2+(\lambda-\delta)x}{2(\lambda-\delta)}dF(x)+\int^\lambda_{\lambda-\delta}\frac{2x^2-(\lambda+\delta)\cdot x+\lambda^2-\delta^2}{2(\lambda-\delta)}dF(x)+\int_\lambda^1\frac{x}{2}dF(x)}$.
\end{lemma}
\begin{proof}
By Lemma~\ref{l1}, the expected cumulative cost of the vehicle for traveling the edges on the TSP tour $T^*$ is 
$
\sum_{e\in E(T^*)}\lra{a\cdot w(e)+b\cdot\frac{\lambda+\delta}{2}\cdot w(e)}=a\cdot w(T^*)+b\cdot\frac{\lambda+\delta}{2}\cdot w(T^*)\leq a\cdot\alpha\cdot\tau+b\cdot\frac{\lambda+\delta}{2}\cdot\alpha\cdot\tau,
$
where the inequality follows from $w(T^*)\leq \alpha\cdot\tau$ since $T^*$ is an $\alpha$-approximate TSP tour.

Define the set of \emph{small customers} as $V_s\coloneqq\{v_i\mid 0<d_i\leq \delta\}$, \emph{big customers} as $V_b\coloneqq\{v_i\mid \delta<d_i\leq \lambda-\delta\}$, \emph{large customers} as $V_l\coloneqq\{v_i\mid \lambda-\delta<d_i\leq \lambda\}$, and \emph{huge customers} as $V_h\coloneqq\{v_i\mid \delta<\lambda\leq 1\}$.

By Lemma~\ref{l2}, when serving customers in $V'$, the expected cumulative cost due to the possible additional visit(s) to $v_0$ is $a\cdot\lra{\sum_{v_i\in V_s}\frac{2d_i}{\lambda-\delta}\cdot l_i+\sum_{v_i\in V_b}\frac{4d_i-2\delta}{\lambda-\delta}\cdot l_i+\sum_{v_i\in V_l}\frac{2d_i+2\lambda-4\delta}{\lambda-\delta}\cdot l_i+\sum_{v_i\in V_h}2\cdot l_i}+b\cdot\lra{\sum_{v_i\in V_s}\frac{(\lambda+\delta)\cdot d_i-d^2_i}{\lambda-\delta}\cdot l_i+\sum_{v_i\in V_b}\frac{d^2_i+(\lambda-\delta)\cdot d_i}{\lambda-\delta}\cdot l_i+\sum_{v_i\in V_l}\frac{2d^2_i-(\lambda+\delta)\cdot d_i+\lambda^2-\delta^2}{\lambda-\delta}\cdot l_i+\sum_{v_i\in V_h}d_i\cdot l_i}$.

Therefore, $ALG.1(\lambda,\delta)$ generates a solution $\T$ with an expected cumulative cost of
\begin{align*}
&\frac{a\cdot\lrA{\alpha\cdot\tau+\sum_{v_i\in V_s}\frac{2d_i}{\lambda-\delta}\cdot l_i+\sum_{v_i\in V_b}\frac{4d_i-2\delta}{\lambda-\delta}\cdot l_i+\sum_{v_i\in V_l}\frac{2d_i+2\lambda-4\delta}{\lambda-\delta}\cdot l_i+\sum_{v_i\in V_h}2\cdot l_i}}{LB}\cdot LB\\
&\quad+\frac{b\cdot\lrA{\frac{\lambda+\delta}{2}\cdot\alpha\cdot\tau+\sum_{v_i\in V_s}\frac{(\lambda+\delta)\cdot d_i-d^2_i}{\lambda-\delta}\cdot l_i+\sum_{v_i\in V_b}\frac{d^2_i+(\lambda-\delta)\cdot d_i}{\lambda-\delta}\cdot l_i}}{LB}\cdot LB\\
&\quad+\frac{b\cdot\lrA{\sum_{v_i\in V_l}\frac{2d^2_i-(\lambda+\delta)\cdot d_i+\lambda^2-\delta^2}{\lambda-\delta}\cdot l_i+\sum_{v_i\in V_h}d_i\cdot l_i}}{LB}\cdot LB\\
&=\frac{\gamma\cdot\lrA{\alpha\cdot\sigma+\int^{\delta}_0\frac{x}{\lambda-\delta}dF(x)+\int^{\lambda-\delta}_{\delta}\frac{2x-\delta}{\lambda-\delta}dF(x)+\int^\lambda_{\lambda-\delta}\frac{x+\lambda-2\delta}{\lambda-\delta}dF(x)+\int_\lambda^11dF(x)}}{\gamma\cdot\max\lrC{\sigma,\ 1}+0.5}\cdot LB\\
&\quad+\frac{\lrA{\frac{\lambda+\delta}{2}\cdot\alpha\cdot\sigma+\int^{\delta}_0\frac{(\lambda+\delta)x-x^2}{2(\lambda-\delta)}dF(x)+\int^{\lambda-\delta}_{\delta}\frac{x^2+(\lambda-\delta)x}{2(\lambda-\delta)}dF(x)}}{\gamma\cdot\max\lrC{\sigma,\ 1}+0.5}\cdot LB\\
&\quad+\frac{\lrA{\int^\lambda_{\lambda-\delta}\frac{2x^2-(\lambda+\delta)\cdot x+\lambda^2-\delta^2}{2(\lambda-\delta)}dF(x)+\int_\lambda^1\frac{x}{2}dF(x)}}{\gamma\cdot\max\lrC{\sigma,\ 1}+0.5}\cdot LB,
\end{align*}
where the equality follows from dividing both the numerator and denominator by $b\cdot\eta$, using $\gamma=a/b$, $\sigma=\tau/\eta$, and $LB=a\cdot\max\{\tau, \eta\}+b\cdot\frac{1}{2}\cdot\eta$ by Lemma~\ref{lb}, and equitation (\ref{int}).
\end{proof}

\subsubsection{The application}
Next, we use $ALG.1(\lambda,\delta)$ to deign approximation algorithms for $a,b>0$, i.e., $\gamma>0$. 

\textbf{An $(\alpha+2)$-approximation algorithm.}
When the vehicle traveling along each edge of the TSP tour in $ALG.1(\lambda,\delta)$, it always carries at least $\delta$ units of goods in total, resulting in a large cumulative cost for the case where $\gamma$ is small. 
Hence, intuitively, if $\gamma$ is small, we simply set $\delta=0$. We may first analyze $ALG.1(\lambda,\delta)$ under the simplified setting where $\delta=0$.
Surprisingly, even setting $\delta=0$, we find that $ALG.1(\lambda, \delta)$ can be used to obtain an $(\alpha+2)$-approximation algorithm for unsplittable Cu-VRPSD. The algorithm is shown in Algorithm~\ref{APP0}.

\begin{algorithm}[ht]
\caption{An $(\alpha+2)$-approximation algorithm for unsplittable Cu-VRPSD}
\label{APP0}
\small
\vspace*{2mm}
\textbf{Input:} An instance of unsplittable Cu-VRPSD. \\
\textbf{Output:} A feasible solution to unsplittable Cu-VRPSD.

\begin{algorithmic}[1]
\State Obtain a solution $\T$ by calling $ALG.1(\lambda, 0)$, where $\lambda=\min\{1,4\gamma/\alpha\}$.
\State Return $\T$.
\end{algorithmic}
\end{algorithm}


\begin{theorem}\label{unttt0}
For unsplittable Cu-VRPSD, there is a randomized $(\alpha+2)$-approximation algorithm.
\end{theorem}
\begin{proof}
Algorithm~\ref{APP0} calls $ALG.1(\lambda, 0)$. By Lemma~\ref{t1}, the approximation ratio of $ALG.1(\lambda, 0)$ is at most
\begin{equation}\label{eqs}
\begin{split}
&\frac{\gamma\cdot\lrA{\alpha\cdot\sigma+\int^{\lambda}_{0}\frac{2x}{\lambda}dF(x)+\int^1_{\lambda}1dF(x)}+\left(\frac{\lambda}{2}\cdot\alpha\cdot\sigma+\int^{\lambda}_{0}\frac{x^2+\lambda\cdot x}{2\cdot \lambda}dF(x)+\int^1_{\lambda}\frac{x}{2}dF(x)\right)}{\gamma\cdot\max\lrC{\sigma,\ 1}+0.5}\\
&\leq \frac{\gamma\cdot\lrA{\alpha\cdot\sigma+\int^{\lambda}_{0}\frac{2x}{\lambda}dF(x)+\int^1_{\lambda}\frac{2x}{\lambda}dF(x)}+\left(\frac{\lambda}{2}\cdot\alpha\cdot\sigma+\int^{\lambda}_{0}\frac{\lambda\cdot x+\lambda\cdot x}{2\cdot \lambda}dF(x)+\int^1_{\lambda}xdF(x)\right)}{\gamma\cdot\max\lrC{\sigma,\ 1}+0.5}\\
&=\frac{\gamma\cdot\lrA{\alpha\cdot\sigma+\frac{2}{\lambda}}+\left(\frac{\lambda}{2}\cdot\alpha\cdot\sigma+1\right)}{\gamma\cdot\max\lrC{\sigma,\ 1}+0.5},
\end{split}
\end{equation}
where the inequality follows from $\int^1_{\lambda}1dF(x)\leq \int^1_{\lambda}\frac{2x}{\lambda}dF(x)$, $\int^{\lambda}_{0}\frac{x^2+\lambda\cdot x}{2\cdot \lambda}dF(x)\leq\int^{\lambda}_{0}\frac{\lambda\cdot x+\lambda\cdot x}{2\cdot \lambda}dF(x)$, and $\int^1_{\lambda}\frac{x}{2}dF(x)\leq \int^1_{\lambda}xdF(x)$ by (\ref{intineq}).

Since $\lambda=\min\{1,4\gamma/\alpha\}$, we consider the following two cases.

\textbf{Case~1: $4\gamma/\alpha\geq 1$.} We have $\lambda=1$ and we further consider the following two cases. 

\textbf{Case~1.1: $\sigma\leq 1$.} By (\ref{eqs}), the approximation ratio of $ALG.1(\lambda, 0)$ is at most 
\begin{align*}
\frac{\gamma\cdot\lrA{\alpha\cdot\sigma+\frac{2}{\lambda}}+\left(\frac{\lambda}{2}\cdot\alpha\cdot\sigma+1\right)}{\gamma\cdot\max\lrC{\sigma,\ 1}+0.5}=\frac{\gamma\cdot\lrA{\alpha\cdot\sigma+2}+\left(\frac{1}{2}\cdot\alpha\cdot\sigma+1\right)}{\gamma\cdot\max\lrC{\sigma,\ 1}+0.5} &\leq\frac{\gamma\cdot\lrA{\alpha+2}+\left(\frac{1}{2}\cdot\alpha+1\right)}{\gamma+0.5}=\alpha+2,
\end{align*}
where the first inequality follows from $\sigma\leq 1$ and $\max\{\sigma,1\}=1$.

\textbf{Case~1.2: $\sigma\geq 1$.} By (\ref{eqs}), the approximation ratio of $ALG.1(\lambda, 0)$ is at most 
\begin{align*}
\frac{\gamma\cdot\lrA{\alpha\cdot\sigma+2}+\left(\frac{1}{2}\cdot\alpha\cdot\sigma+1\right)}{\gamma\cdot\max\lrC{\sigma,\ 1}+0.5} &=\frac{\gamma\cdot\lrA{\alpha\cdot\sigma+2}+\left(\frac{1}{2}\cdot\alpha\cdot\sigma\cdot\frac{1}{\sigma}+\frac{1}{2}\cdot\alpha\cdot\sigma\cdot\lrA{1-\frac{1}{\sigma}}+1\right)}{\gamma\cdot\sigma+0.5}\\
&\leq\frac{\gamma\cdot\lrA{\alpha\cdot\sigma+2}+\left(
\frac{1}{2}\cdot\alpha+\frac{2}{\alpha}\cdot\gamma\cdot\alpha\cdot\sigma\cdot\lrA{1-\frac{1}{\sigma}}+1\right)}{\gamma\cdot\sigma+0.5}\\
&=\frac{\gamma\cdot\sigma\cdot\lrA{\alpha+2}+\left(\frac{1}{2}\cdot\alpha+1\right)}{\gamma\cdot\sigma+0.5}=\alpha+2,
\end{align*}
where the first equality follows from $\sigma\geq 1$ and the first inequality from $\frac{1}{2}\leq \frac{2}{\alpha}\cdot\gamma$ by the definition of Case~1.

\textbf{Case~2: $4\gamma/\alpha\leq 1$.} We have $\lambda=4\gamma/\alpha$. Similarly, by (\ref{eqs}), the approximation ratio of $ALG.1(\lambda, 0)$ is at most 
\begin{align*}
\frac{\gamma\cdot\lrA{\alpha\cdot\sigma+\frac{2}{\lambda}}+\left(\frac{\lambda}{2}\cdot\alpha\cdot\sigma+1\right)}{\gamma\cdot\max\lrC{\sigma,\ 1}+0.5}&=\frac{\gamma\cdot\lrA{\alpha\cdot\sigma+\frac{\alpha}{2}\cdot\frac{1}{\gamma}}+\left(\frac{2}{\alpha}\cdot\gamma\cdot\alpha\cdot\sigma+1\right)}{\gamma\cdot\max\lrC{\sigma,\ 1}+0.5}\\
&=\frac{\gamma\cdot\sigma\cdot(\alpha+2)+\frac{1}{2}(\alpha+2)}{\gamma\cdot\max\lrC{\sigma,\ 1}+0.5}\leq\alpha+2,
\end{align*}
where the inequality follows from $\sigma\leq \max\{\sigma,\ 1\}$.

Therefore, Algorithm~\ref{APP0} is a randomized $(\alpha+2)$-approximation algorithm. Moreover, its running time is dominated by computing an $\alpha$-approximate TSP tour.
\end{proof}

\textbf{Further Improvements.}
Next, we aim to design better approximation algorithms. 

As shown in the proof of Theorem~\ref{unttt0}, if we call $ALG.1(\lambda,\delta)$ with $\delta=0$ and $\lambda$ being a unique value, in the worst case we may have $\int^\lambda_0x^2dF(x)=\int^\lambda_0\lambda\cdot xdF(x)$, i.e., almost all customers $v_i$ with $l_i>0$ and $d_i\leq\lambda$ have a demand of $d_i=\lambda$. Consequently, we may obtain $\int^{\theta\cdot\lambda}_0xdF(x)=0$ for any fixed $\theta\in(0,1)$, as will be shown in Lemma~\ref{a1}, and then $ALG.1(\theta\cdot \lambda,\delta)$ with $\delta=0$ and some $\theta\in(0,1)$ may generate a better solution. This suggests an improved approximation algorithm for unsplittable Cu-VRPSD shown in Algorithm~\ref{APP1}, denoted as $APPROX.1(\lambda,\theta,p)$. 

\begin{remark}
In $APPROX.1(\lambda,\theta,p)$, wejust use a two-point distributed algorithm deployment because it is easier to analyze, though it may perform worse than a multi-point distributed algorithm deployment. Further discussion on this choice will be provided at the end of this section.
\end{remark}

\begin{algorithm}[ht]
\caption{An approximation algorithm for unsplittable Cu-VRPSD ($APPROX.1(\lambda,\theta,p)$)}
\label{APP1}
\small
\vspace*{2mm}
\textbf{Input:} An instance of unsplittable Cu-VRPSD, and three parameters $\lambda\in(0,1]$, $\theta\in(0,1)$ and $p\in(0,1)$.\\
\textbf{Output:} A feasible solution to unsplittable Cu-VRPSD.
\begin{algorithmic}[1]
\State Call $ALG.1(\lambda,0)$ with a probability of $p$ and call $ALG.1(\theta\cdot\lambda,0)$ with a probability of $1-p$.
\end{algorithmic}
\end{algorithm}

Then, our goal is to find $(\lambda,\theta,p)$ minimizing the approximation ratio of $APPROX.1(\lambda,\theta,p)$.

\begin{lemma}\label{a1}
For any $\theta\in(0,1)$, we have $\int^{\theta\cdot\lambda}_0 xdF(x)\leq \frac{1}{\lambda-\theta\cdot\lambda}\lra{\int^\lambda_0\lambda\cdot xdF(x)-\int^\lambda_0x^2dF(x)}$.   
\end{lemma}
\begin{proof}
By (\ref{int}) and (\ref{intineq}), we have 
\begin{align*}
\int^\lambda_0x^2dF(x)&=\int^{\theta\cdot\lambda}_0x^2dF(x)+\int^\lambda_{\theta\cdot\lambda} x^2dF(x)\\
&\leq \theta\cdot\lambda\cdot \int^{\theta\cdot\lambda}_0 xdF(x)+\lambda\cdot\int^\lambda_{\theta\cdot\lambda} xdF(x)=\lambda\cdot\int^{\lambda}_0 xdF(x)-(\lambda-\theta\cdot\lambda)\cdot \int^{\theta\cdot\lambda}_0 xdF(x)   
\end{align*}
Hence, we have $\int^{\theta\cdot\lambda}_0 xdF(x)\leq \frac{1}{\lambda-\theta\cdot\lambda}\lra{\int^\lambda_0\lambda\cdot xdF(x)-\int^\lambda_0x^2dF(x)}$.
\end{proof}

\begin{theorem}\label{t2}
For unsplittable Cu-VRPSD, under $\alpha=1.5$, we can find $(\lambda,\theta,p)$ such that the approximation ratio of $APPROX.1(\lambda,\theta,p)$ is bounded by $10/3$ for any $\gamma\in(0,0.375]$ and $3.456$ for any $\gamma\in(0.375,1.444]$.
\end{theorem}
\begin{proof}
If we call $ALG.1(\lambda,\delta)$ with $\delta=0$ and $\lambda$ being a unique value, as shown in the proof of Theorem~\ref{unttt0}, a good choice for $\lambda$ is $\min\{1, 4\gamma/\alpha\}$. 
For the sake of analysis, we directly set $\lambda=\min\{1, 4\gamma/\alpha\}$.

If $\int^\lambda_0xdF(x)=0$, we can obtain $\int^\lambda_0x^2dF(x)\leq \lambda\cdot \int^\lambda_0xdF(x)=0$.
Hence, we define $\mu\coloneqq0$ if $\int^\lambda_0xdF(x)=0$, and $\mu\coloneqq\frac{\int^\lambda_0x^2dF(x)}{\int^\lambda_0xdF(x)}$ otherwise.

By Lemmas~\ref{mid} and \ref{t1}, the approximation ratio of $ALG.1(\lambda,0)$ is at most
\begin{align*}
&\max_{\sigma\geq0}\frac{\gamma\cdot\lrA{\alpha\cdot\sigma+\int^{\lambda}_{0}\frac{2x}{\lambda}dF(x)+\int_\lambda^1 1dF(x)}+\lrA{\frac{\lambda}{2}\cdot\alpha\cdot\sigma+\int^{\lambda}_{0}\frac{x^2+\lambda\cdot x}{2\lambda}dF(x)+\int_\lambda^1\frac{x}{2}dF(x)}}{\gamma\cdot\max\lrC{\sigma,\ 1}+0.5}\\
&=\max_{\sigma\geq0}\frac{\gamma\cdot\lrA{\alpha\cdot\sigma+\frac{2}{\lambda}\cdot\int_{0}^{\lambda}xdF(x)+\int_{\lambda}^{1}1dF(x)}+\lrA{\frac{\lambda}{2}\cdot\alpha\cdot\sigma+\frac{1+\mu/\lambda}{2}\cdot\int_{0}^{\lambda}xdF(x)+\frac{1}{2}\cdot\int_{\lambda}^{1}xdF(x)}}{\gamma\cdot\max\lrC{\sigma,\ 1}+0.5}\\
&\leq\max_{\sigma\geq0} \frac{\gamma\cdot\lrA{\alpha\cdot\sigma+\frac{2}{\lambda}\cdot\int_{0}^{\lambda}xdF(x)+\frac{1}{\lambda}\cdot\int_{\lambda}^{1}xdF(x)}+\lrA{\frac{\lambda}{2}\cdot\alpha\cdot\sigma+\frac{1+\mu/\lambda}{2}\cdot\int_{0}^{\lambda}xdF(x)+\frac{1}{2}\cdot\int_{\lambda}^{1}xdF(x)}}{\gamma\cdot\max\lrC{\sigma,\ 1}+0.5}\\
&=\max_{\sigma\geq0} \frac{\gamma\cdot\lrA{\alpha\cdot\sigma+\frac{1}{\lambda}\cdot\int_{0}^{\lambda}xdF(x)+\frac{1}{\lambda}}+\lrA{\frac{\lambda}{2}\cdot\alpha\cdot\sigma+\frac{\mu/\lambda}{2}\cdot\int_{0}^{\lambda}xdF(x)+\frac{1}{2}}}{\gamma\cdot\max\lrC{\sigma,\ 1}+0.5}\\
&\leq\max_{\sigma\geq0} \frac{\gamma\cdot\lrA{\alpha\cdot\sigma+\frac{2}{\lambda}}+\lrA{\frac{\lambda}{2}\cdot\alpha\cdot\sigma+\frac{\mu/\lambda+1}{2}}}{\gamma\cdot\max\lrC{\sigma,\ 1}+0.5},
\end{align*}
where the first equality follows from the definition of $\mu$, the second equality from $\int^1_0xdF(x)=1$ by (\ref{int}), the first inequality from $\int_{\lambda}^{1}1dF(x)\leq \frac{1}{\lambda}\cdot\int_{\lambda}^{1}xdF(x)$ by (\ref{intineq}), and the second inequality from $\int_{0}^{\lambda}xdF(x)\leq \int_{0}^{1}xdF(x)=1$.

Since $\int_{0}^{\lambda}xdF(x)\leq \int_{0}^{1}xdF(x)=1$, by Lemma~\ref{a1}, we have 
\begin{equation}\label{eq1}
\begin{split}
\int^{\theta\cdot\lambda}_0 xdF(x)&\leq \frac{1}{\lambda-\theta\cdot\lambda}\cdot\lrA{\int^\lambda_0\lambda\cdot xdF(x)-\int^\lambda_0x^2dF(x)}=\frac{\lambda-\mu}{\lambda-\theta\cdot\lambda}\cdot\int^\lambda_0xdF(x)\leq \frac{\lambda-\mu}{\lambda-\theta\cdot\lambda}  
\end{split}
\end{equation}

Similarly, the approximation ratio of $ALG.1(\theta\cdot\lambda,0)$ is at most
\begin{align*}
&\max_{\sigma\geq0}\frac{\gamma\cdot\lrA{\alpha\cdot\sigma+\int_{0}^{\theta\cdot\lambda}\frac{2x}{\theta\cdot\lambda}dF(x)+\int_{\theta\cdot\lambda}^{1}1dF(x)}+\lrA{\frac{\theta\cdot\lambda}{2}\cdot\alpha\cdot\sigma+\int_{0}^{\theta\cdot\lambda}\frac{x^2+\theta\cdot\lambda\cdot x}{2\cdot\theta\cdot\lambda}dF(x)+\int_{\theta\cdot\lambda}^{1}\frac{x}{2}dF(x)}}{\gamma\cdot\max\lrC{\sigma,\ 1}+0.5}\\
&\leq\max_{\sigma\geq0} \frac{\gamma\cdot\lrA{\alpha\cdot\sigma+\frac{2}{\theta\cdot\lambda}\cdot\int_{0}^{\theta\cdot\lambda}xdF(x)+\frac{1}{\theta\cdot\lambda}\cdot\int_{\theta\cdot\lambda}^{1}xdF(x)}+\lrA{\frac{\theta\cdot\lambda}{2}\cdot\alpha\cdot\sigma+\int_{0}^{\theta\cdot\lambda}xdF(x)+\frac{1}{2}\cdot\int_{\theta\cdot\lambda}^{1}xdF(x)}}{\gamma\cdot\max\lrC{\sigma,\ 1}+0.5}\\
&=\max_{\sigma\geq0} \frac{\gamma\cdot\lrA{\alpha\cdot\sigma+\frac{1}{\theta\cdot\lambda}\cdot\int_{0}^{\theta\cdot\lambda}xdF(x)+\frac{1}{\theta\cdot\lambda}}+\lrA{\frac{\theta\cdot\lambda}{2}\cdot\alpha\cdot\sigma+\frac{1}{2}\cdot\int_{0}^{\theta\cdot\lambda}xdF(x)+\frac{1}{2}}}{\gamma\cdot\max\lrC{\sigma,\ 1}+0.5}\\
&\leq\max_{\sigma\geq0} \frac{\gamma\cdot\lrA{\alpha\cdot\sigma+\frac{1}{\theta\cdot\lambda}\cdot\frac{\lambda-\mu}{\lambda-\theta\cdot\lambda}+\frac{1}{\theta\cdot\lambda}}+\lrA{\frac{\theta\cdot\lambda}{2}\cdot\alpha\cdot\sigma+\frac{1}{2}\cdot\frac{\lambda-\mu}{\lambda-\theta\cdot\lambda}+\frac{1}{2}}}{\gamma\cdot\max\lrC{\sigma,\ 1}+0.5}\\
&=\max_{\sigma\geq0} \frac{\gamma\cdot\lrA{\alpha\cdot\sigma+\frac{1}{\theta\cdot\lambda}\cdot\frac{2\lambda-\mu-\theta\cdot\lambda}{\lambda-\theta\cdot\lambda}}+\lrA{\frac{\theta\cdot\lambda}{2}\cdot\alpha\cdot\sigma+\frac{1}{2}\cdot\frac{2\lambda-\mu-\theta\cdot\lambda}{\lambda-\theta\cdot\lambda}}}{\gamma\cdot\max\lrC{\sigma,\ 1}+0.5},
\end{align*}
where the first inequality follows from 
(\ref{intineq}), the second inequality from (\ref{eq1}), and the first equality from $\int^1_0xdF(x)=1$ by (\ref{int}).

Recall that in $APPROX.1(\lambda,\theta,p)$, we call $ALG.1(\lambda,0)$ (resp., $ALG.1(\theta\cdot\lambda,0)$) with a probability of $p$ (resp., $1-p$). Hence, to erase the items related to $\mu$ in the numerators of the approximation ratios of $ALG.1(\lambda,0)$ and $ALG.1(\theta\cdot\lambda,0)$, we need to set $p$ such that 
\[
p\cdot\frac{1}{2}\cdot\frac{1}{\lambda}+(1-p)\cdot\lrA{\gamma\cdot\frac{1}{\theta\cdot\lambda}\cdot\frac{-1}{\lambda-\theta\cdot\lambda}+\frac{1}{2}\cdot\frac{-1}{\lambda-\theta\cdot\lambda}}=0.
\]
Then, we can obtain $p=\frac{\frac{1}{2(\lambda-\theta\cdot\lambda)}+\frac{\gamma}{\theta\cdot\lambda(\lambda-\theta\cdot\lambda)}}{\frac{1}{2\lambda}+\frac{1}{2(\lambda-\theta\cdot\lambda)}+\frac{\gamma}{\theta\cdot\lambda(\lambda-\theta\cdot\lambda)}}$. Clearly, we have $p\in[0,1]$. Hence, the approximation ratio is $\max_{\sigma\geq 0} R(\sigma)$,
where
\begin{equation}\label{ratio1}
\begin{split}
&R(\sigma)\coloneqq\frac{\gamma\cdot\lrA{\alpha\cdot\sigma+p\cdot\frac{2}{\lambda}+(1-p)\cdot\frac{1}{\theta\cdot\lambda}\cdot\frac{2\lambda-\theta\cdot\lambda}{\lambda-\theta\cdot\lambda}}+\lrA{\frac{p\cdot\lambda+(1-p)\cdot\theta\cdot\lambda}{2}\cdot\alpha\cdot\sigma+p\cdot\frac{1}{2}+(1-p)\cdot\frac{1}{2}\cdot\frac{2\lambda-\theta\cdot\lambda}{\lambda-\theta\cdot\lambda}}}{\gamma\cdot\max\lrC{\sigma,\ 1}+0.5}.
\end{split}
\end{equation}

It can be verified that $\max_{\sigma\geq 0} R(\sigma)=\max_{\sigma\geq 1} R(\sigma)$. Moreover, since the function $\frac{a'x+b'}{c'x+d'}$ with $x\geq 1$ and $a',b',c',d'>0$ attains the maximum value only if $x=1$ or $x=\infty$, we know that the approximation ratio is bounded by 
$
\max\{R(1),\ R(\infty)\}.
$
Recall that $\lambda=\min\{1, 4\gamma/\alpha\}$ and $p=\frac{\frac{1}{2(\lambda-\theta\cdot\lambda)}+\frac{\gamma}{\theta\cdot\lambda(\lambda-\theta\cdot\lambda)}}{\frac{1}{2\lambda}+\frac{1}{2(\lambda-\theta\cdot\lambda)}+\frac{\gamma}{\theta\cdot\lambda(\lambda-\theta\cdot\lambda)}}$. Assume $\alpha=1.5$, then we have $\alpha/4=0.375$. By calculation, we have the following results.
\begin{itemize}
    \item When $\gamma\in(0,0.375]$, setting $\theta=0.5$, we have $\max\{R(1),\ R(\infty)\}\equiv 10/3$;
    \item When $\gamma\in[0.375,1.444]$, setting $\theta=0.6677$, we have $\max\{R(1),\ R(\infty)\}\leq 3.456$.
\end{itemize}

The values of $\max\{R(1),\ R(\infty)\}$ under $\lambda=\min\{1,4/\gamma/\alpha\}$, $p=\frac{\frac{1}{2(\lambda-\theta\cdot\lambda)}+\frac{\gamma}{\theta\cdot\lambda(\lambda-\theta\cdot\lambda)}}{\frac{1}{2\lambda}+\frac{1}{2(\lambda-\theta\cdot\lambda)}+\frac{\gamma}{\theta\cdot\lambda(\lambda-\theta\cdot\lambda)}}$, and different values of $\theta$ can be found in Figure~\ref{fig2}.\footnote{We computed the values and generated this figure (along with others in the paper) using Python. The source code is available at \url{https://github.com/JingyangZhao/Cu-VRPSD}.} 

\begin{figure}[ht]
    \centering
    \includegraphics[scale=0.6]{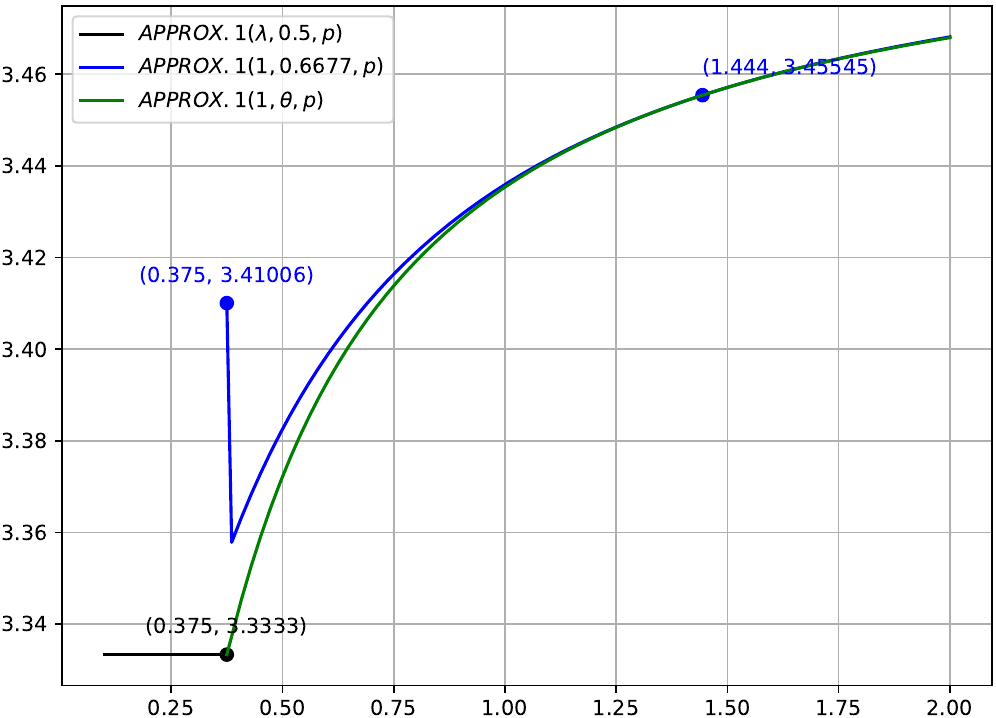}
    \caption{The approximation ratio of $APPROX.1(\lambda,\theta,p)$ under $\gamma\in(0,2)$, i.e., the value of $\max\{R(1),R(\infty)\}$, where the black line denotes the approximation ratio of $APPROX.1(\lambda,0.5,p)$ with $\lambda=4\gamma/1.5$ and $p=\frac{\frac{1}{2(\lambda-\theta\cdot\lambda)}+\frac{\gamma}{\theta\cdot\lambda(\lambda-\theta\cdot\lambda)}}{\frac{1}{2\lambda}+\frac{1}{2(\lambda-\theta\cdot\lambda)}+\frac{\gamma}{\theta\cdot\lambda(\lambda-\theta\cdot\lambda)}}$, the blue line denotes the approximation ratio of $APPROX.1(\lambda,0.6677,p)$ with $\lambda=1$ and the same $p$, and the green line denotes the approximation ratio of $APPROX.1(\lambda,\theta,p)$ with $\lambda=1$ and the same $p$ but with a more suitable $\theta$ related to $\gamma$. Specifically, for any $\gamma$, we choose the most suitable $\theta$ by enumerating the values in $\{\frac{i}{10000}\mid i\in[10000]\}$.}
    \label{fig2}
\end{figure}

The result for $\gamma\in(0,0.375]$ may be surprising. We give the details of its proof.
\begin{claim}
When $\gamma\in(0,0.375]$, setting $\theta=0.5$, we have $\max\{R(1),\ R(\infty)\}\equiv 10/3$.
\end{claim}
\begin{proof}[Claim proof]
Note that $\lambda=\min\{1, 4\gamma/\alpha\}=4\gamma/\alpha$ and $\alpha=1.5$. Setting $\theta=0.5$, we can obtain that $p=\frac{\frac{1}{2(\lambda-\theta\cdot\lambda)}+\frac{\gamma}{\theta\cdot\lambda(\lambda-\theta\cdot\lambda)}}{\frac{1}{2\lambda}+\frac{1}{2(\lambda-\theta\cdot\lambda)}+\frac{\gamma}{\theta\cdot\lambda(\lambda-\theta\cdot\lambda)}}=\frac{5}{6}$. Hence, under $\sigma\geq 1$, by (\ref{ratio1}), we have 
\begin{align*}
&R(\sigma)=\frac{3/2\cdot\gamma\cdot\sigma+5/8+3/8+11/6\cdot\gamma\cdot\sigma+2/3}{\gamma\cdot\sigma+0.5}=\frac{10}{3}\cdot\frac{\gamma\cdot\sigma+0.5}{\gamma\cdot\sigma+0.5}=\frac{10}{3}.
\end{align*}
Hence, we have $\max\{R(1),\ R(\infty)\}\equiv 10/3$.
\end{proof}

This finishes the proof.
\end{proof}

We mention that the approximation ratio of $APPROX.1$ may still achieve $\alpha+2$ when $\gamma=\infty$, which is consistent with the result in Theorem~\ref{unttt0}. Hence, it can not improve the current best approximation algorithm for unsplittable VRPSD~\citep{DBLP:journals/ior/GuptaNR12}.

\subsection{The Second Algorithm}
In this section, we will introduce our second algorithm, denoted as $ALG.2(\lambda, \delta)$, which can be used to obtain a $3.456$-approximation algorithm for unsplittable Cu-VRPSD with any $\gamma\in[1.444,\infty)$, and an $(\alpha+1.75)$-approximation algorithm for unsplittable VRPSD. Here, we may require $\lambda\in(0,1]$, $\delta\in(0,\lambda/2]$, and $1/\delta\in\mathbb{N}$. 

$ALG.2(\lambda, \delta)$ is based on $ALG.1(\lambda, \delta)$. 
The vehicle will skip customers $v_i$ with $d_i>\lambda$ when it travels along the TSP tour in $ALG.1(\lambda, \delta)$, and then satisfy each of them using a single tour at last. In $ALG.2(\lambda, \delta)$, the main difference is that the vehicle will skip customers $v_i$ with $d_i>\delta$, and at last use the better method from either satisfying each of them using a single tour or solving the weighted $(1-\delta)/\delta$-set cover problem as shown below. 

Given a \emph{feasible} set of unsatisfied customers $S$ such that the total demand of all customers in $S$ is at most $1$, we know that the number of customers $\size{S}$ is at most $(1-\delta)/\delta$ since each unsatisfied customer $v_i$ has a demand of $d_i>\delta$. 
Then, we can optimally compute a tour with a minimum cumulative cost $Cu(S)$ for all customers in $S$ in $O(\size{S}!)$ time.
There are at most $n^{O(1/\delta)}$ number of feasible sets since $\size{S}\leq(1-\delta)/\delta$.
Therefore, to satisfy customers $v_i$ with $d_i>\delta$, we can obtain an instance of weighted $(1-\delta)/\delta$-set cover by taking each unsatisfied customer as an element, and each feasible set $S$ of unsatisfied customers as a set with a weight of $Cu(S)$ in polynomial time.
By calling a $\rho$-approximation algorithm for weighted $(1-\delta)/\delta$-set cover~\citep{DBLP:conf/sosa/0001L023}, we can obtain a set of tours satisfying all customers $v_i$ with $d_i>\delta$.

According to the two methods, there are two set of tours $\T_1$ and $\T_2$, and their cumulative cost can be computed in polynomial time. Hence, we route the vehicle according to the tours in $\T'$, where $\T'\coloneqq\T_1$ if $Cu(\T_1)\leq Cu(\T_2)$ and $\T'\coloneqq\T_2$ otherwise. 

The details of $ALG.2(\lambda,\delta)$ is shown in Algorithm~\ref{alg2}.
\begin{algorithm}[ht]
\caption{An algorithm for unsplittable Cu-VRPSD ($ALG.2(\lambda,\delta)$)}
\label{alg2}
\small
\vspace*{2mm}
\textbf{Input:} An instance of unsplittable Cu-VRPSD, and two parameters $\lambda\in(0,1]$, $\delta\in(0,\lambda/2]$, and $1/\delta\in\mathbb{N}$.\\
\textbf{Output:} A feasible solution $\T$ to unsplittable Cu-VRPSD.

\begin{algorithmic}[1]
\State Obtain an $\alpha$-approximate TSP tour $T^*=v_0v_1v_2\dots v_n v_0$, as Line~\ref{tsptour} in $ALG.1(\lambda, \delta)$.

\State Load the vehicle with $S_0\coloneqq (L_0,\delta)$ units of goods, including $L_0$ units of normal goods and $\delta$ units of backup goods, where $L_0\sim U[0,\lambda-\delta)$.

\State Initialize $i\coloneqq 1$ and $V^*\coloneqq\emptyset$.

\While{$i\leq n$}

\State Go to customer $v_i$.

\If{$\delta<d_i\leq 1$}
\State Record $v_i$'s demand, and let $V^*\coloneqq V^*\cup \{v_i\}$ and $S_i\coloneqq(L_i,\delta)$, where $L_i\coloneqq L_{i-1}$.

\ElsIf {$d_{i}\leq L_{i-1}$}
\State Deliver $(d_i,0)$ units of goods to $v_i$, and update $S_i\coloneqq(L_i,\delta)$, where $L_{i}\coloneqq L_{i-1}+\ceil{\frac{d_i-L_{i-1}}{\lambda-\delta}}\cdot (\lambda-\delta)-d_i= L_{i-1}-d_i$.

\Else \Comment{Since $d_i\leq\delta$, we must have $L_{i-1}<d_i\leq L_{i-1}+\delta$}
\State Deliver $(L_{i-1},d_i-L_{i-1})$ units of goods to $v_i$, goes to the depot, load the vehicle with $(L_{i-1}+\ceil{\frac{d_i-L_{i-1}}{\lambda-\delta}}\cdot (\lambda-\delta)-d_i,d_i-L_{i-1})$ units of goods, and update $S_i\coloneqq(L_i,\delta)$, where $L_{i}\coloneqq L_{i-1}+\ceil{\frac{d_i-L_{i-1}}{\lambda-\delta}}\cdot (\lambda-\delta)-d_i= L_{i-1}+(\lambda-\delta)-d_i$.
\EndIf
\State $i \coloneqq i+1$.
\EndWhile

\State Go to the depot.

\State\label{tour1} Consider a set of tours $\T_1$ by obtaining a single tour as in Line~\ref{clean} for each $v_i\in V^*$.

\State\label{tour2} Consider a set of tours $\T_2$ by calling a $\rho$-approximation algorithm for weighted $\frac{1-\delta}{\delta}$-set cover~\citep{DBLP:conf/sosa/0001L023}, where the instance is constructed as follows: 
\begin{enumerate}
\item Obtain all possible feasible sets $S$ of customers in $V^*$ such that the total demand of all customers in $S$ is at most 1; 
\item For each feasible set $S$, compute a tour with a minimum cumulative cost $Cu(S)$ for all customers in $S$;
\item Get an instance of weighted $\frac{1-\delta}{\delta}$-set cover by taking each customer in $V^*$ as an element, and each feasible set $S$ as a weighted set with a weight of $Cu(S)$. 
\end{enumerate}

\State\label{tour'} Let $\T'\coloneqq\T_1$ if $Cu(\T_1)\leq Cu(\T_2)$ and $\T'\coloneqq\T_2$ otherwise. 
\State Route the vehicle according to the tours in $\T'$.
\end{algorithmic}
\end{algorithm}

\begin{lemma}\label{t3}
For unsplittable Cu-VRPSD with any $\lambda\in(0,1]$, $\delta\in(0,\lambda/2]$, and $1/\delta\in\mathbb{N}$, conditioned on $\chi=d$, $ALG.2(\lambda,\delta)$ outputs a solution $\T$ with an expected cumulative cost of
\[
\frac{\gamma\cdot\lrA{\alpha\cdot\sigma+\int^{\delta}_0\frac{x}{\lambda-\delta}dF(x)}+\lrA{\frac{\lambda+\delta}{2}\cdot\alpha\cdot\sigma+\int^{\delta}_0\frac{(\lambda+\delta)x-x^2}{2(\lambda-\delta)}dF(x)}}{\gamma\cdot\max\lrC{\sigma,\ 1}+0.5}\cdot LB+Cu(\T'),
\]
where
\[
Cu(\T')\leq\min\lrC{\frac{\int^{1}_\delta\frac{2\gamma+x}{2}dF(x)}{\gamma\cdot\max\lrC{\sigma,\ 1}+0.5}\cdot LB,\ \rho\cdot Cu(\T^*)}.
\]
\end{lemma}
\begin{proof}
$ALG.2(\lambda, \delta)$ in Algorithm~\ref{alg2} outputs a set of tours $\T$, where $\T\setminus\T'$ satisfies all customers $v_i$ with $d_i\leq\delta$ and $\T'$ satisfies all customers $v_i$ with $d_i>\delta$.

By the analysis of $ALG.1(\lambda,\delta)$, and Lemma~\ref{t1}, we know that the expected cumulative cost of $\T\setminus\T'$ is at most
\[
\frac{\gamma\cdot\lrA{\alpha\cdot\sigma+\int^{\delta}_0\frac{x}{\lambda-\delta}dF(x)}+\lrA{\frac{\lambda+\delta}{2}\cdot\alpha\cdot\sigma+\int^{\delta}_0\frac{(\lambda+\delta)\cdot x-x^2}{2(\lambda-\delta)}dF(x)}}{\gamma\cdot\max\lrC{\sigma,\ 1}+0.5}\cdot LB.
\]

By Lines~\ref{tour1}-\ref{tour'}, we know that $Cu(\T')\leq \min\{Cu(\T_1),\ Cu(\T_2)\}$. By the proof of Lemma~\ref{t1}, we can obtain 
\[
Cu(\T_1)=\frac{\gamma\cdot\int^{1}_\delta1dF(x)+\int^{1}_\delta\frac{x}{2}dF(x)}{\gamma\cdot\max\lrC{\sigma,\ 1}+0.5}\cdot LB=\frac{\int^{1}_\delta\frac{2\gamma+x}{2}dF(x)}{\gamma\cdot\max\lrC{\sigma,\ 1}+0.5}\cdot LB.
\]

Since $\T_2$ is obtained by calling a $\rho$-approximation algorithm for weighted $\frac{1-\delta}{\delta}$-set cover, we have $Cu(\T_2)\leq\rho\cdot SC$, where $SC$ is the weight of the optimal solution for weighted $\frac{1-\delta}{\delta}$-set cover.
Recall that $\T^*$ is the solution satisfying all customers obtained by the optimal policy, whose cumulative cost is clearly an upper bound of $SC$, the cumulative cost of the optimal solution for satisfying $V^*$, by the triangle inequality. Hence, we have 
\[
Cu(\T_2)\leq \rho\cdot Cu(\T^*).
\]

Conditioned on the demand realization, $Cu(\T')$, $Cu(\T')$, and $Cu(\T')$ are all fixed. Hence, the lemma holds.
\end{proof}

\subsubsection{The applications}
Our goal is to obtain a $3.456$-approximation algorithm for unsplittable Cu-VRPSD with any $\gamma\in[1.444,\infty)$. As a byproduct, we will also obtain an $(\alpha+1.75)$-approximation algorithm for unsplittable VRPSD.

Since in $APPROX.1(\lambda,\theta,p)$, $ALG.1(\lambda, \delta)$ sets $\lambda=1$ for any $\gamma>0.375$, we also set $\lambda=1$ in $ALG.2(\lambda, \delta)$ for the sake of analysis.
Moreover, since weighted 2-set cover~\citep{DBLP:conf/sosa/0001L023} can be solved optimally in polynomial time, i.e., $\rho=1$ when $\delta=1/3$, we set $\delta=1/3$ in $ALG.2(\lambda, \delta)$.

According to Lemmas~\ref{t1} and \ref{t3}, we will show that $ALG.2(\lambda, \delta)$ can be used to make a trade-off with $ALG.1(\lambda, \delta)$.
We use the approximation algorithm for unsplittable Cu-VRPSD shown in Algorithm~\ref{APP2}, denoted as $APPROX.2$.

\begin{algorithm}[ht]
\caption{An approximation algorithm for unsplittable Cu-VRPSD ($APPROX.2$)}
\label{APP2}
\small
\vspace*{2mm}
\textbf{Input:} An instance of unsplittable Cu-VRPSD.\\
\textbf{Output:} A feasible solution to unsplittable Cu-VRPSD.
\begin{algorithmic}[1]
\State 
Call $ALG.1(1,1/3)$ with a probability of $0.5$ and call $ALG.2(1,1/3)$ with a probability of $0.5$.
\end{algorithmic}
\end{algorithm}

\begin{theorem}\label{t4}
For unsplittable Cu-VRPSD, under $\alpha=1.5$, $APPROX.2$ is a randomized $3.456$-approximation algorithm for any $\gamma\in[1.444,\infty)$.
\end{theorem}
\begin{proof}
By Lemma~\ref{t1}, the approximation ratio of $ALG.1(1,1/3)$ is at most
\begin{align*}
&\frac{\gamma\cdot\lrA{\alpha\cdot\sigma+\int^{\frac{1}{3}}_0\frac{3x}{2}dF(x)+\int^\frac{2}{3}_{\frac{1}{3}}\frac{6x-1}{2}dF(x)+\int^1_\frac{2}{3}\frac{3x+1}{2}dF(x)}}{\gamma\cdot\max\lrC{\sigma,\ 1}+0.5}\\
&\quad+\frac{\lrA{\frac{2}{3}\cdot\alpha\cdot\sigma+\int^{\frac{1}{3}}_0\frac{4x-3x^2}{4}dF(x)+\int^\frac{2}{3}_{\frac{1}{3}}\frac{3x^2+2x}{4}dF(x)+\int^1_\frac{2}{3}\frac{9x^2-6x+4}{6}dF(x)}}{\gamma\cdot\max\lrC{\sigma,\ 1}+0.5}\\
&\leq \frac{\gamma\cdot\lrA{\alpha\cdot\sigma+\int^{\frac{1}{3}}_0\frac{3x}{2}dF(x)+\int^1_{\frac{1}{3}}\frac{6x-1}{2}dF(x)}+\lrA{\frac{2}{3}\cdot\alpha\cdot\sigma+\int^{\frac{1}{3}}_0xdF(x)+\int^1_{\frac{1}{3}}\frac{3x^2+2x}{4}dF(x)}}{\gamma\cdot\max\lrC{\sigma,\ 1}+0.5},
\end{align*}
where the inequality follows from $\frac{4x-3x^2}{4}\leq x$ for any $x\in[0,\frac{1}{3}]$, and $\frac{3x+1}{2}\leq \frac{6x-1}{2}$ and $\frac{9x^2-6x+4}{6}\leq \frac{3x^2+2x}{4}$ for any $x\in[\frac{2}{3},1]$.

Assume that $\gamma\geq0.1667>\frac{1}{6}$, and then we have the following inequality.
\begin{equation}\label{eq2}
\begin{split}
R(\T')&\coloneqq\min\lrC{\frac{\int^{1}_\frac{1}{3}\frac{2\gamma+x}{2}dF(x)}{\gamma\cdot\max\lrC{\sigma,\ 1}+0.5},\ 1}\\
&\leq \lrA{\frac{1}{2}+\frac{1}{6\gamma+1}}\cdot\frac{\int^{1}_\frac{1}{3}\frac{2\gamma+x}{2}dF(x)}{\gamma\cdot\max\lrC{\sigma,\ 1}+0.5}+\lrA{\frac{1}{2}-\frac{1}{6\gamma+1}}\\
&\leq \frac{1}{2}\cdot\frac{\int^{1}_\frac{1}{3}\frac{2\gamma+x}{2}dF(x)}{\gamma\cdot\max\lrC{\sigma,\ 1}+0.5}+\frac{1}{6\gamma+1}\cdot\frac{\int^{1}_\frac{1}{3}\frac{6\gamma\cdot x+x}{2}dF(x)}{\gamma\cdot\max\lrC{\sigma,\ 1}+0.5}+\lrA{\frac{1}{2}-\frac{1}{6\gamma+1}}\\
&=\frac{\int^{1}_\frac{1}{3}\frac{2\gamma+x}{4}dF(x)}{\gamma\cdot\max\lrC{\sigma,\ 1}+0.5}+\frac{\int^{1}_\frac{1}{3}\frac{x}{2}dF(x)}{\gamma\cdot\max\lrC{\sigma,\ 1}+0.5}+\frac{6\gamma-1}{12\gamma+2}\\
&=\frac{\int^{1}_\frac{1}{3}\frac{2\gamma+3x}{4}dF(x)}{\gamma\cdot\max\lrC{\sigma,\ 1}+0.5}+\frac{6\gamma-1}{12\gamma+2},
\end{split}
\end{equation}
where the second inequality follows from $\int^1_\frac{1}{3}1dF(x)\leq \int^1_\frac{1}{3}3xdF(x)$.

By Lemma~\ref{t3}, the approximation ratio of $ALG.2(1,1/3)$ is at most
\begin{align*}
&\frac{\gamma\cdot\lrA{\alpha\cdot\sigma+\int^{\frac{1}{3}}_0\frac{3x}{2}dF(x)}+\lrA{\frac{2}{3}\cdot\alpha\cdot\sigma+\int^{\frac{1}{3}}_0\frac{4x-3x^2}{4}dF(x)}}{\gamma\cdot\max\lrC{\sigma,\ 1}+0.5}+R(\T')\\
&\leq \frac{\gamma\cdot\lrA{\alpha\cdot\sigma+\int^{\frac{1}{3}}_0\frac{3x}{2}dF(x)+\int^{1}_\frac{1}{3}\frac{1}{2}dF(x)}+\lrA{\frac{2}{3}\cdot\alpha\cdot\sigma+\int^{\frac{1}{3}}_0xdF(x)+\int^{1}_\frac{1}{3}\frac{3x}{4}dF(x)}}{\gamma\cdot\max\lrC{\sigma,\ 1}+0.5}+\frac{6\gamma-1}{12\gamma+2},
\end{align*}
where the inequality follows from (\ref{eq2}), and $\frac{4x-3x^2}{4}\leq x$ for any $x\in[0,\frac{1}{3}]$.

Recall that in $APPROX.2$ we call $ALG.1(1,1/3)$ (resp., $ALG.2(1,1/3)$) with a probability of $0.5$ (resp., $0.5$). The approximation ratio of $APPROX.2$ is at most
\begin{align*}
&\frac{\gamma\cdot\lrA{\alpha\cdot\sigma+\int^{\frac{1}{3}}_0\frac{3x}{2}dF(x)+\int^{1}_\frac{1}{3}\frac{3x}{2}dF(x)}+\lrA{\frac{2}{3}\cdot\alpha\cdot\sigma+\int^{\frac{1}{3}}_0xdF(x)+\int^{1}_\frac{1}{3}\frac{3x^2+5x}{8}dF(x)}}{\gamma\cdot\max\lrC{\sigma,\ 1}+0.5}+\frac{6\gamma-1}{24\gamma+4}\\
&\leq\frac{\gamma\cdot\lrA{\alpha\cdot\sigma+1.5}+\lrA{\frac{2}{3}\cdot\alpha\cdot\sigma+1}}{\gamma\cdot\max\lrC{\sigma,\ 1}+0.5}+\frac{6\gamma-1}{24\gamma+4},
\end{align*}
where the inequality follows from $\frac{3x^2+5x}{8}\leq x$ for any $x\in[\frac{2}{3},1]$, and $\int^1_0xdF(x)=1$. 

Under $\gamma\in[0.1667,\infty)$, the approximation ratio is 
\[
\max_{\sigma\geq 0} R(\sigma),\quad\quad\mbox{where}\quad\quad R(\sigma)\coloneqq\frac{\gamma\cdot\lrA{\alpha\cdot\sigma+1.5}+\lrA{\frac{2}{3}\cdot\alpha\cdot\sigma+1}}{\gamma\cdot\max\lrC{\sigma,\ 1}+0.5}+\frac{6\gamma-1}{24\gamma+4}.
\]

As in the proof of Theorem~\ref{t2}, we have $\max_{\sigma\geq 0} R(\sigma)=\max\{R(1),\ R(\infty)\}$. Moreover, it can be verified that $R(1)\geq R(\infty)$ holds for any $\gamma\geq 0.5$. Hence, we have $R(\sigma)=R(1)$ for any $\gamma\geq 0.5$. Assume that $\alpha=1.5$. By calculation, we have the following result.
\begin{itemize}
    \item When $\gamma\in[1.444,\infty)$, we have $\max\{R(1),\ R(\infty)\}=R(1)\leq 3.456$.
\end{itemize}

Hence, $APPROX.2$ is a randomized $3.456$-approximation algorithm for unsplittable Cu-VRPSD with any $\gamma\in[1.444,\infty)$.
The approximation ratios of $APPROX.1(\lambda,\theta,p)$ and $APPROX.2$ under $\gamma\in(0,2)$ can be found in Figure~\ref{fig3}.

\begin{figure}[ht]
    \centering
    \includegraphics[scale=0.6]{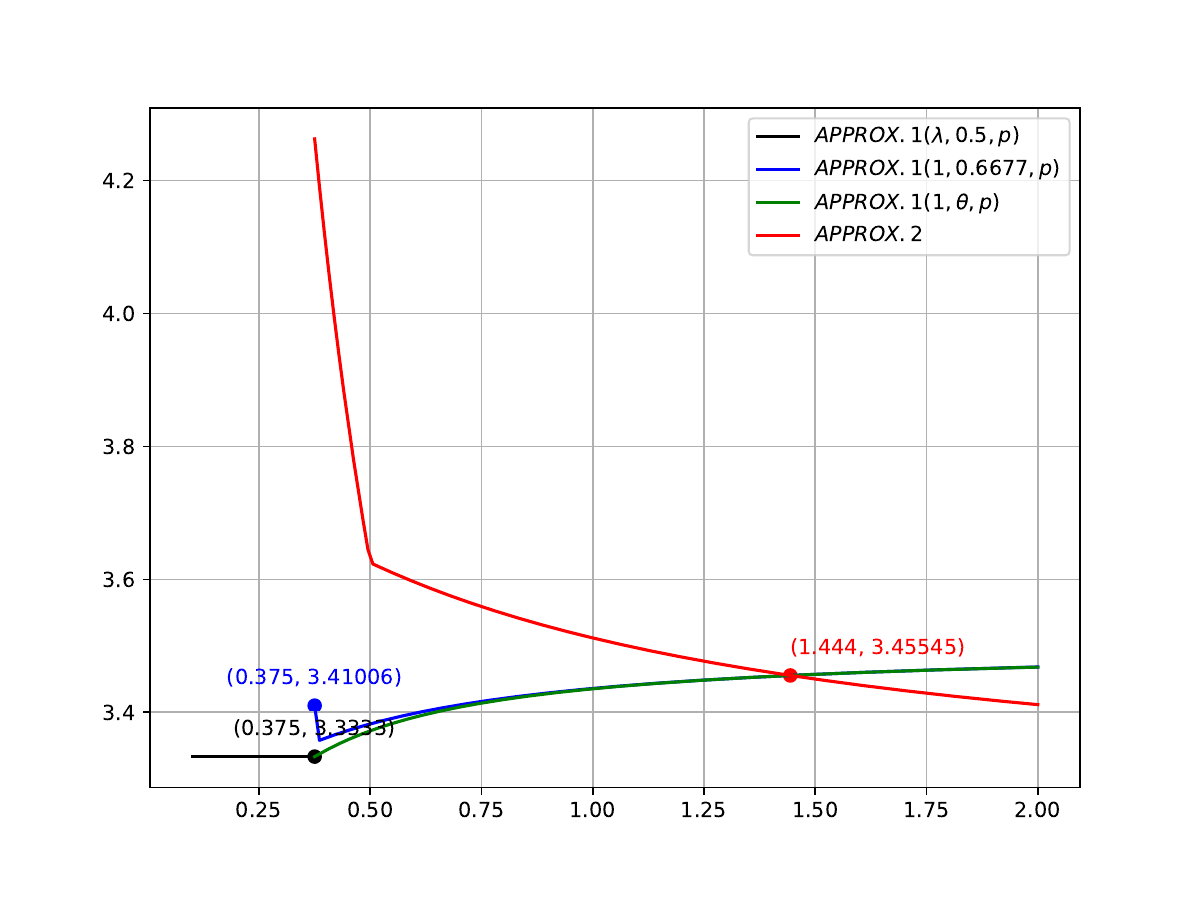}
    \caption{The approximation ratios of $APPROX.1(\lambda,\theta,p)$ and $APPROX.2$ under $\gamma\in(0,2)$, where the red line denotes the approximation ratio of $APPROX.2$.}
    \label{fig3}
\end{figure}
\end{proof}

Combining the results in Lemma~\ref{a=0} and Theorems~\ref{t2} and \ref{t4}, we obtain the following result.

\begin{theorem}\label{th23}
For unsplittable Cu-VRPSD, under $\alpha=1.5$, there is a randomized $3.456$-approximation algorithm.  
\end{theorem}

\begin{remark}
We believe that the analysis of our algorithm $APPROX.2$ is not tight. We proved an approximation ratio of $3.456$ for $APPROX.2$ when $\gamma\geq 1.444$, where in the worst case we have $\gamma=1.444$. However, numerical experiments show that the approximation ratio may achieve $3.408$ when $\gamma=1.444$. Moreover, by calling $APPROX.1(\lambda,\theta,p)$ when $\gamma\leq 1.033$ and $APPROX.2$ otherwise, we may obtain an approximation ratio of $3.438$ for unsplittable Cu-VRPSD. 
The details of the numerical experiments are put in Appendix~\ref{tightness}.
\end{remark}

We note that to obtain a better approximation algorithm for unsplittable Cu-VRPSD, a possible way is to run $ALG.1$ (resp., $ALG.2$) with a probability of $p_\gamma$ (resp., $1-p_\gamma$), where $p_\gamma$ is a function dependent on $\gamma$. Moreover, when executing $ALG.1$ or $ALG.2$, the parameters $\lambda$ and $\delta$ follow certain distributions related to $\gamma$, i.e., a multi-point distributed algorithm deployment as mentioned before.
Although this method is simple, conducting a tight analysis requires significant effort. We leave this for future work.

\begin{problem}
Determine the best possible approximation ratio that can be achieved using the above framework by optimizing the probability $p_\gamma$ and the parameter distributions of $\lambda$ and $\delta$ as functions of $\gamma$.
\end{problem}


Recall that Cu-VRPSD reduces to VRPSD when $b=0$, which corresponds to the case where $\gamma=\infty$. By the proof of Theorem~\ref{t4}, the approximation ratio of $APPROX.2$ is at most $\max_{\sigma\geq 0} \lrC{\frac{\alpha\cdot\sigma+1.5}{\max\lrC{\sigma,\ 1}}+\frac{1}{4}}\leq\alpha+1.75$.
Thus, $APPROX.2$ is a randomized $(\alpha+1.75)$-approximation algorithm for unsplittable VRPSD.

\begin{corollary}\label{coro1}
For unsplittable VRPSD, there is a randomized $(\alpha+1.75)$-approximation algorithm.
\end{corollary}

\section{Two Algorithms for Unsplittable Cu-VRP}
In this section, we give an $(\alpha+1+\ln2+\varepsilon<3.194)$-approximation algorithm for unsplittable Cu-VRP, where $\varepsilon>$ is any small positive constant. 

\subsection{The First Algorithm}
Based on the well-known randomized rounding method for weighted $k$-set cover, we propose an $(\alpha+1+\ln2+\varepsilon<3.194)$-approximation algorithm, denoted as $ALG.3(\lambda,\delta)$, for unsplittable Cu-VRP with any $\gamma>\gamma_0$, where $\gamma_0>0$ is any fixed constant.

Recall that $ALG.2(\lambda,\delta)$ first satisfies customers $v_i$ with $d_i\leq\delta$ by traveling along the TSP tour and then customers $v_i$ with $d_i>\delta$ by possibly solving weighted $\frac{1-\delta}{\delta}$-set cover. However, it may only be used for $\delta=1/3$ since the best-known approximation ratio of weighted $3$-set cover is about $1.79$~\citep{DBLP:conf/sosa/0001L023}, which is already too large. 
In $ALG.3(\lambda,\delta)$, since the demands of customers are known in advance for unsplittable Cu-VRP, we first try to satisfy customers in $V^*\coloneqq\{v_i\in V'\mid d_i>\delta\}$ by solving weighted $\frac{1-\delta}{\delta}$-set cover using the randomized rounding method, which may yield better performance. However, due to the randomness, some customers in $V^*$ may still be unsatisfied. Then, we satisfy all remaining customers by calling $ALG.1(\lambda,\delta)$. 
This method was used to obtain an $(\alpha+1+\ln2-\varepsilon')$-approximation algorithm for some small constant $\varepsilon'>0$ for unsplittable VRP~\citep{uncvrp}. The details are shown as follows. 

To obtain an instance of weighted $\frac{1-\delta}{\delta}$-set cover, we use the method in Line~\ref{tour2} of $ALG.2(\lambda,\delta)$. Now, we have obtained a set of feasible sets $\S$, and each $S\in\S$ has a weight of $Cu(S)$. Then, we obtain the linear relaxation of weighted set cover as shown in $(\ref{lp})$, and it can be solved in $n^{O(1/\delta)}$ since $\size{\S}=n^{O(1/\delta)}$. In the randomized rounding method, we select each $S\in\S$ with a probability of $\min\{\ln2\cdot x_S,1\}$. Denote the set of selected sets by $\S'$, which corresponds to a set of tours $\T'$ satisfying a subset of customers $V^{**}\subseteq V^*$. Note that $Cu(\T')\leq Cu(\S')$ since 
we may perform shortcutting to ensure that each customer appears in only one tour and it does not increase the cumulative cost by the triangle inequality.
At last, we call $ALG.1(\lambda,\delta)$ to obtain a set of tours $\T''$ to satisfy the left customers in $V'\setminus V^{**}$. Due to the stochastic demands in unsplittable Cu-VRPSD, the load of the vehicle in $ALG.1(\lambda,\delta)$ may be greater than the delivered units of goods in each tour of $\T''$. In unsplittable Cu-VRP, we can optimize the tours in $\T''$ so that the load equals the delivered units of goods. Moreover, for each tour $T=v_0v_{i_1}\dots v_{i_T}v_0\in \T''$, we consider another tour with the opposite direction, i.e., $v_0v_{i_T}\dots v_{i_1}v_0$, and choose the better one into our final solution.
\begin{alignat}{3}
\text{minimize} & \quad & \sum_{S\in\S} Cu(S)\cdot x_S \notag\\
\text{subject to} & \quad & \sum_{\substack{S\in \S: v\in S}}x_S \geq\ & 1, \quad && \forall\ v \in V^*,\label{lp}\\
&& x_S \geq\ & 0, \quad && \forall\ S \in \S. \notag
\end{alignat}

The details of $ALG.3(\lambda,\delta)$ is shown in Algorithm~\ref{alg3}.
\begin{algorithm}[ht]
\caption{An algorithm for unsplittable Cu-VRPSD ($ALG.3(\lambda,\delta)$)}
\label{alg3}
\small
\vspace*{2mm}
\textbf{Input:} An instance of unsplittable Cu-VRPSD, and two parameters $\lambda\in(0,1]$, $\delta\in(0,\lambda/2]$, and $1/\delta\in\mathbb{N}$.\\
\textbf{Output:} A feasible solution $\T$ to unsplittable Cu-VRPSD.

\begin{algorithmic}[1]
\State Get an instance $(V^*,\S)$ of weighted $\frac{1-\delta}{\delta}$-set cover using Line~\ref{tour2} in $ALG.2(\lambda,\delta)$.

\State Solve the linear program of weighted set cover in (\ref{lp}).

\State Select each $S\in\S$ with a probability of $\min\{\ln2\cdot x_S, 1\}$. 
Denote the set of selected sets by $\S'$, corresponding tours by $\T'$, and satisfied customers by $V^{**}$.
\State Call $ALG.1(\lambda,\delta)$ to obtain a set of tours $\T''$ to satisfy the customers in $V'\setminus V^{**}$.
\State\label{opt} For each tour in $\T''$, ensure the load of the vehicle is the delivered units of goods, obtain another tour with the opposite direction, and choose the better one into $\T'''$.
\State Return $\T'\cup\T'''$.
\end{algorithmic}
\end{algorithm}

\subsubsection{The analysis}
\begin{lemma}[*]\label{pr}
It holds that $\EE{Cu(\S')}\leq \ln2\cdot Cu(\T^*)$, $\PP{v\notin V^{**}}=1$ for any $v\in V\setminus V^*$, and $\PP{v\notin V^{**}}\leq 1/2$ for any $v\in V^*$.
\end{lemma}



\begin{lemma}[*]\label{t5}
For unsplittable Cu-VRP with any $\lambda\in(0,1]$, $\delta\in(0,\lambda/2]$, and $1/\delta\in\mathbb{N}$, $ALG.3(\lambda,\delta)$ generates a solution $\T$ with an expected cumulative cost of
\[
\ln2\cdot Cu(\T^*)+\frac{\gamma\cdot\lrA{\alpha\cdot\sigma+\frac{1}{\lambda-\delta}}+\frac{\lambda}{2}\cdot\lrA{\alpha\cdot\sigma+\frac{1}{\lambda-\delta}}}{\gamma\cdot\max\lrC{\sigma,\ 1}+0.5}\cdot LB.
\]
\end{lemma}

\begin{theorem}\label{t6}
For unsplittable Cu-VRP with any constant $\gamma_0>0$ and  some small $\varepsilon>0$, there is a randomized $(\alpha+1+\ln2+\varepsilon<3.194)$-approximation algorithm for $\gamma>\gamma_0$.
\end{theorem}
\begin{proof}
Let $\lambda=\min\{1,2\gamma/\alpha\}$. 
By Lemma~\ref{t5}, for any $\lambda\in(0,1]$ and constant $\varepsilon>0$, setting $0<\delta\leq\frac{\lambda\cdot\varepsilon'}{1+\varepsilon'}$ satisfying $(1+\alpha)\cdot \varepsilon'\leq \varepsilon$, $\delta\in(0,\lambda/2]$ and $1/\delta\in\mathbb{N}$, $ALG.3(\lambda,\delta)$ generates a solution $\T$ with an expected approximation ratio of
\begin{equation}\label{eqt}
\begin{split}
\ln2+\frac{\gamma\cdot\lrA{\alpha\cdot\sigma+\frac{1}{\lambda-\delta}}+\frac{\lambda}{2}\cdot\lrA{\alpha\cdot\sigma+\frac{1}{\lambda-\delta}}}{\gamma\cdot\max\lrC{\sigma,\ 1}+0.5}&\leq \ln2+\frac{\gamma\cdot\lrA{\alpha\cdot\sigma+\frac{1}{\lambda}}+\gamma\cdot\frac{\varepsilon'}{\lambda}+\frac{\lambda}{2}\cdot\lrA{\alpha\cdot\sigma+\frac{1}{\lambda}}+\frac{\lambda}{2}\cdot\frac{\varepsilon'}{\lambda}}{\gamma\cdot\max\lrC{\sigma,\ 1}+0.5}\\
&\leq \ln2+\frac{\gamma\cdot\lrA{\alpha\cdot\sigma+\frac{1}{\lambda}}+\frac{\lambda}{2}\cdot\lrA{\alpha\cdot\sigma+\frac{1}{\lambda}}}{\gamma\cdot\max\lrC{\sigma,\ 1}+0.5}+\varepsilon,
\end{split}
\end{equation}
where the first inequality follows from $\frac{1}{\lambda-\delta}\leq \frac{1}{\lambda}+\frac{\varepsilon'}{\lambda}$ since $\delta\leq\frac{\lambda\cdot\varepsilon'}{1+\varepsilon'}$, and the second from $\lra{\frac{\gamma}{\lambda}+0.5}\cdot\varepsilon'\leq \lra{\gamma\cdot \max\{\sigma,1\}+0.5}\cdot\lra{1+\alpha}\cdot\varepsilon'\leq \varepsilon$ since $\lambda=\min\{1,2\gamma/\alpha\}$ and $(1+\alpha)\cdot \varepsilon'\leq \varepsilon$.

Next, by using an analysis similar to the one in the proof of Theorem~\ref{unttt0},  we show that the approximation ratio is bounded by $\alpha+1+\ln2+\varepsilon$.

\textbf{Case~1: $\gamma\geq\alpha/2$.} We have $\lambda=1$ and we further consider the following two cases.

\textbf{Case~1.1: $\sigma\leq1$.} By (\ref{eqt}), the approximation ratio is at most 
\begin{align*}
&\ln2+\frac{\gamma\cdot\lrA{\alpha+1}+\frac{1}{2}\cdot\lrA{\alpha+1}}{\gamma+0.5}+\varepsilon=\alpha+1+\ln2+\varepsilon.
\end{align*}

\textbf{Case~1.2: $\sigma\geq1$.} By (\ref{eqt}), the approximation ratio is at most 
\begin{align*}
\ln2+\frac{\gamma\cdot\lrA{\alpha\cdot\sigma+1}+\frac{1}{2}\cdot\lrA{\alpha\cdot\sigma+1}}{\gamma\cdot\sigma+0.5}+\varepsilon&= \ln2+\frac{\gamma\cdot\lrA{\alpha\cdot\sigma+1}+\frac{1}{2}\cdot\lrA{\alpha\cdot\sigma\cdot\lrA{\frac{1}{\sigma}+1-\frac{1}{\sigma}}+1}}{\gamma\cdot\sigma+0.5}+\varepsilon\\
&\leq\ln2+\frac{\gamma\cdot\lrA{\alpha\cdot\sigma+1}+\frac{1}{2}\cdot\alpha+\gamma\cdot\sigma\cdot\lrA{1-\frac{1}{\sigma}}+0.5}{\gamma\cdot\sigma+\frac{1}{2}}+\varepsilon\\
&=\alpha+1+\ln2+\varepsilon,
\end{align*}
where the inequality follows from $\alpha/2\leq \gamma$ by the definition of Case~1.

\textbf{Case~2: $\gamma\geq\alpha/2$.} We have $\lambda=2\gamma/\alpha$. By (\ref{eqt}), the approximation ratio is at most 
\begin{align*}
\ln2+\frac{\gamma\cdot\lrA{\alpha\cdot\sigma+\frac{1}{2\gamma/\alpha}}+\frac{2\gamma/\alpha}{2}\cdot\lrA{\alpha\cdot\sigma+\frac{1}{2\gamma/\alpha}}}{\gamma\cdot\max\{\sigma,\ 1\}+0.5}+\varepsilon&=\ln2+\frac{\gamma\cdot\alpha\cdot\sigma+\alpha/2+\gamma\cdot\sigma+1/2}{\gamma\cdot\max\{\sigma,\ 1\}+0.5}+\varepsilon\\
&=\ln2+(\alpha+1)\cdot\frac{\gamma\cdot\sigma+1/2}{\gamma\cdot\max\{\sigma,\ 1\}+0.5}+\varepsilon\\
&\leq\alpha+1+\ln2+\varepsilon.
\end{align*}

Let $\varepsilon$ be a sufficiently small constant and $\alpha=\frac{3}{2}$. We can set $\delta=\frac{1}{\ceil{\frac{1+\varepsilon'}{\lambda\cdot\varepsilon'}}}$, where $\varepsilon'=\frac{\varepsilon}{1+\alpha}$. Since $\lambda\leq 1$, the running time is $n^{O(1/\delta)}=n^{O(1/(\min\{\gamma,1\}\cdot \varepsilon))}$, dominated by solving the linear program of weighted set cover, and the approximation ratio is at most $2.5+\ln2+\varepsilon<3.194$. 
\end{proof}

\subsection{The Second Algorithm}
In this section, we propose an $(\alpha+1+\ln2-0.029)$-approximation algorithm for unsplittable Cu-VRP with $\gamma\in(0,0.285]$, denoted as $ALG.4(\lambda)$.
Combing with Lemma~\ref{a=0} and Theorem~\ref{t6}, $ALG.4(\lambda)$ implies an $(\alpha+1+\ln2+\varepsilon<3.194)$-approximation algorithm for unsplittable Cu-VRP.

In $ALG.4(\lambda)$, we call $ALG.1(\lambda,0)$ to obtain a set of tours $\T'$ to satisfy all customers, and then we optimize each tour in $\T'$ as Line~\ref{opt} in $ALG.3$.

\begin{algorithm}[ht]
\caption{An algorithm for unsplittable Cu-VRPSD ($ALG.4(\lambda)$)}
\label{alg4}
\small
\vspace*{2mm}
\textbf{Input:} An instance of unsplittable Cu-VRPSD, and two parameters $\lambda\in(0,1]$, $\delta\in(0,\lambda/2]$, and $1/\delta\in\mathbb{N}$.\\
\textbf{Output:} A feasible solution $\T$ to unsplittable Cu-VRPSD.

\begin{algorithmic}[1]
\State Call $ALG.1(\lambda,0)$ to obtain a set of tours $\T'$ to satisfy all customers.
\State\label{opt+} For each tour in $\T'$, ensure the load of the vehicle is the delivered units of goods, obtain another tour with the opposite direction, and choose the better one into $\T$.
\State Return $\T$.
\end{algorithmic}
\end{algorithm}

\begin{lemma}[*]\label{t7}
For unsplittable Cu-VRP with any $\lambda\in(0,1]$, $ALG.4(\lambda)$ generates a solution $\T$ with an expected cumulative cost of
\[
\frac{\gamma\cdot\lrA{\alpha\cdot\sigma+\int^{\lambda}_{0}\frac{2x}{\lambda}dF(x)+\int_\lambda^1 1dF(x)}+\lrA{\frac{\lambda}{2}\cdot\alpha\cdot\sigma+\int^{\lambda}_{0}\frac{x^2/2+\lambda\cdot x}{2\lambda}dF(x)+\int_\lambda^1\frac{x}{2}dF(x)}}{\gamma\cdot\max\lrC{\sigma,\ 1}+0.5}\cdot LB.
\]
\end{lemma}

Similarly, we use $ALG.4(\lambda)$ to design an algorithm for unsplittable Cu-VRP shown in Algorithm~\ref{APP4}.

\begin{algorithm}[ht]
\caption{An approximation algorithm for unsplittable Cu-VRPSD ($APPROX.4(\lambda,\theta,p)$)}
\label{APP4}
\small
\vspace*{2mm}
\textbf{Input:} An instance of unsplittable Cu-VRP, and three parameters $\lambda\in(0,1]$, $\theta\in(0,1)$ and $p\in(0,1)$.\\
\textbf{Output:} A feasible solution to unsplittable Cu-VRPSD.
\begin{algorithmic}[1]
\State 
Call $ALG.4(\lambda)$ with a probability of $p$ and call $ALG.4(\theta\cdot\lambda)$ with a probability of
$1-p$.
\end{algorithmic}
\end{algorithm}

\begin{theorem}[*]\label{t8}
For unsplittable Cu-VRP, when $\alpha=1.5$, we can find $(\lambda,\theta,p)$ such that the approximation ratio of $APPROX.4(\lambda,\theta,p)$ is bounded by $3.163$ for any $\gamma\in(0,0.428]$. 
Moreover, for any $\alpha\in[1,1.5]$, we can find $(\lambda,\theta,p)$ such that the approximation ratio of $APPROX.4(\lambda,\theta,p)$ is bounded by $\alpha+1+\ln2-0.029$ for any $\gamma\in(0,0.285]$.
\end{theorem}

By Theorems~\ref{t6} and \ref{t8}, we have the following result.

\begin{theorem}\label{th7}
For unsplittable Cu-VRP with any small constant $\varepsilon>0$, there is a randomized $(\alpha+1+\ln2+\varepsilon<3.194)$-approximation algorithm.
\end{theorem}

Our approximation ratio for unsplittable Cu-VRP almost matches the current best-known approximation ratio of $\alpha+1+\ln2-\varepsilon'$ for unsplittable VRP~\citep{blauth2022improving,uncvrp}, where $\varepsilon'$ is a small constant related to $\alpha$.
A natural question is whether we can slightly improve our approximation ratio by using the methods in~\citep{blauth2022improving,uncvrp}. 
However, we remark that even with a stronger lower bound of $LB_\varepsilon\coloneqq a\cdot\max\{\tau,\ (1+\varepsilon)\cdot\eta\}+b\cdot0.5\cdot\eta$ for some constant $\varepsilon>0$, we still cannot achieve a better approximation ratio (for general $\gamma>0$) by using the analysis in Theorem~\ref{t6}.

\section{An Algorithm for Splittable Cu-VRPSD}
In this section, we propose a randomized $(\alpha+1)$-approximation algorithm for splittable Cu-VRPSD.

\subsection{The Algorithm}
In the splittable case, if the vehicle cannot fully satisfy the visiting customer, it can partially satisfy customers by delivering all remaining goods. Thus, we may no longer need to consider using the idea of backup goods. We will modify $ALG.1(\lambda, \delta)$ with $\delta=0$ into $ALG.S(\lambda)$, where we will not distinguish between backup goods and normal goods.

Initially, we load the vehicle with $S_0=L_0$ units of goods at the depot, where $L_{0}\sim U[0,\lambda)$. Then, $ALG.S(\lambda)$ adapts the following greedy strategy: the vehicle will return to the depot to reload $\lambda$ units of goods whenever its current load is less than the demand of the severing customer. 

The details of $ALG.S(\lambda)$ can be found in Algorithm~\ref{alg5}.

\begin{algorithm}[ht]
\caption{An algorithm for splittable Cu-VRPSD ($ALG.S(\lambda)$)}
\label{alg5}
\small
\vspace*{2mm}
\textbf{Input:} An instance of splittable Cu-VRPSD, and one parameter $\lambda\in(0,1]$.\\
\textbf{Output:} A feasible solution $\T$ to splittable Cu-VRPSD.

\begin{algorithmic}[1]
\State Obtain an $\alpha$-approximate TSP tour $T^*$ using an $\alpha$-approximation algorithm for metric TSP, orient $T^*$ in either clockwise or counterclockwise direction, and denote $T^*=v_0v_1v_2\dots v_n v_0$ by renumbering the customers following the direction.

\State Load the vehicle with $S_0\coloneqq L_0$ units of goods, where $L_0\sim U[0,\lambda)$.

\State Initialize $i\coloneqq 1$ and $V^*\coloneqq\emptyset$.

\While{$i\leq n$}

\State Go to customer $v_i$.

\If {$d_{i}\leq L_{i-1}$}
\State\label{case1ofalg5} Deliver $(d_i,0)$ units of goods to $v_i$, and update $S_i\coloneqq L_i$, where $L_{i}\coloneqq L_{i-1}+\ceil{\frac{d_i-L_{i-1}}{\lambda}}\cdot \lambda-d_i= L_{i-1}-d_i$.
\State $i \coloneqq i+1$.
\Else
\State Deliver $L_{i-1}$ units of goods to $v_i$, and update $d_i\coloneqq d_i - L_{i-1}$. 
\State\label{splitaddition} Return to the depot, load the vehicle with $\lambda$ units of goods, and update $L_{i-1}\coloneqq \lambda$.
\EndIf
\EndWhile
\State Go to the depot.
\end{algorithmic}
\end{algorithm}

\subsubsection{The analysis}
Similarly, for each $i\in[n+1]$, the vehicle  in $ALG.S(\lambda)$ carries $S_{i-1}=L_{i-1}$ units of goods when traveling along the edge $v_{i-1}v_{i\bmod (n+1)}$ of the TSP tour $T^*$. We have the following lemma.

\begin{lemma}\label{sl0}
For any $i\in [n+1]$, it holds $L_{i-1}=L_0 + \ceil{\frac{\sum_{j=1}^{i-1}d_j-L_0}{\lambda}}\cdot \lambda-\sum_{j=1}^{i-1} d_j$, and moreover, $L_{i-1}\sim U[0,\lambda)$, conditioned on $\chi=d$.
\end{lemma}
\begin{proof}
By treating each customer $v_i$ with $d_i>\lambda$ as a set of customers with a total demand of $d_i$, where each customer has a demand of at most $\lambda$, one can view an instance of splittable Cu-VRPSD as an instance of unsplittable Cu-VRPSD, where each customer has a demand of at most $\lambda$. 
Then, this lemma follows directly from the proof of Lemma~\ref{l0}.
\end{proof}

\begin{lemma}\label{sl1}
In $ALG.S(\lambda)$, the expected cumulative cost conditioned on $\chi=d$ during the vehicle's travel from $v_{i-1}$ to $v_{i}$ is $a\cdot w(v_{i-1},v_{i})+b\cdot\frac{\lambda}{2}\cdot w(v_{i-1},v_{i})$.
\end{lemma}
\begin{proof}
It can be obtained by using an analysis similar to the one in the proof of Lemma~\ref{l1}. 
\end{proof}

By Line~\ref{splitaddition} in $ALG.1(\lambda)$, if the vehicle visits $v_i$ carrying $L_{i-1}<d_i$ units of goods, it will first deliver $L_{i-1}$ units of goods for $v_i$, and then proceed to the depot to reload $\lambda$ units of goods to deliver the remaining required demand for $v_i$. We refer to this process as an \emph{additional visit to $v_0$}. Note that if $d_i\leq\lambda$, the vehicle can satisfy $v_i$ using at most one additional visit to $v_0$, as after one additional visit to $v_0$, the vehicle carries $\lambda$ units of goods. However, if $\lambda<d_i$, the vehicle may incur more than one additional visits to $v_0$.

\begin{lemma}[*]\label{sl2}
Conditioned on $\chi=d$, when serving each customer $v_i$ in $ALG.S(\lambda)$, the expected cumulative cost of the vehicle due to the possible additional visit(s) to $v_0$ is $a\cdot 2\cdot \frac{d_i}{\lambda}\cdot l_i+b\cdot d_i\cdot l_i$.
\end{lemma}

\begin{lemma}[*]\label{t9}
For splittable Cu-VRPSD with any $\lambda\in(0,1]$, conditioned on $\chi=d$, $ALG.S(\lambda)$ generates a solution $\T$ with an expected cumulative cost of
\[
\frac{\gamma\cdot\lrA{\alpha\cdot\sigma+\frac{1}{\lambda}}+\frac{\lambda}{2}\cdot\lrA{\alpha\cdot\sigma+\frac{1}{\lambda}}}{\gamma\cdot\max\lrC{\sigma,\ 1}+0.5}\cdot LB.
\]
\end{lemma}

\subsubsection{The application}
Next, we use $ALG.S(\lambda)$ to deign an $(\alpha+1)$-approximation algorithm for splittable Cu-VRPSD. The algorithm is shown in Algorithm~\ref{APP5}.

\begin{algorithm}[ht]
\caption{An $(\alpha+1)$-approximation algorithm for splittable Cu-VRPSD}
\label{APP5}
\small
\vspace*{2mm}
\textbf{Input:} An instance of splittable Cu-VRPSD. \\
\textbf{Output:} A feasible solution to splittable Cu-VRPSD.

\begin{algorithmic}[1]
\State Obtain a solution $\T$ by calling $ALG.S(\lambda)$, where $\lambda=\min\{1,2\gamma/\alpha\}$.
\State Return $\T$.
\end{algorithmic}
\end{algorithm}


\begin{theorem}\label{spttt0}
For splittable Cu-VRPSD, there is a randomized $(\alpha+1)$-approximation algorithm.
\end{theorem}
\begin{proof}
By Lemma~\ref{t9}, the expected approximation ratio of $ALG.S(\lambda)$ is at most $\frac{\gamma\cdot\lrA{\alpha\cdot\sigma+\frac{1}{\lambda}}+\frac{\lambda}{2}\cdot\lrA{\alpha\cdot\sigma+\frac{1}{\lambda}}}{\gamma\cdot\max\lrC{\sigma,\ 1}+0.5}$.
By the proof of Theorem~\ref{t6}, under $\lambda=\min\{1,2\gamma/\alpha\}$, this expression has a value of at most $\alpha+1$.
Therefore, Algorithm~\ref{APP5} is a randomized $(\alpha+1)$-approximation algorithm. 
\end{proof}

Since Cu-VRP is a special case of Cu-VRPSD, we obtain the following corollary.

\begin{corollary}\label{spttt00}
For splittable Cu-VRP, there is a randomized $(\alpha+1)$-approximation algorithm.
\end{corollary}

Similarly, our approximation ratio for splittable Cu-VRP almost matches the current best-known approximation ratio of $\alpha+1-\varepsilon$ for splittable VRP~\citep{blauth2022improving}, where $\varepsilon>\frac{1}{3000}$ when $\alpha=3/2$.
It remains unclear whether the methods in~\citep{blauth2022improving} can be used to achieve a better approximation ratio.

Moreover, Algorithm~\ref{APP5} can be derandomized as follows.
Since Algorithm~\ref{APP5} generates a solution based on partitioning an $\alpha$-approximate TSP tour, by the above corollary, there exists a partition such that the corresponding solution is an $(\alpha+1)$-approximate solution for splittable Cu-VRP. Given the TSP tour, based on the dynamic programming approach, the optimal partition with the minimum cumulative cost can be found in $O(n^2)$~\cite{gaur2013routing}. 
Hence, we can obtain a deterministic $(\alpha+1)$-approximation algorithm.

\section{Conclusion}
By using the idea of skipping customers with large demands during the TSP tour and satisfying them later, combined with careful analysis, we improve the approximation ratios for Cu-VRPSD, VRPSD, and Cu-VRP. Whether this idea is useful in designing practical algorithms for these problems is worthy of further study.

For (Cu-)VRPSD, our approximation algorithms are all randomized. Theorem~\ref{thm0} shows a method to obtain some straightforward deterministic algorithms by employing deterministic algorithms for (Cu-)VRP. 
It might be possible to obtain better deterministic approximation algorithms by exploring alternative techniques or by imposing additional restrictions on the distribution of each customer's demand.
Also, a significant lower bound on the approximation ratio is currently known only for unsplittable VRP. It would be interesting to establish stronger lower bounds for (Cu-)VRP(SD).

Moreover, since our numerical experiments show that based on $APPROX.1(\lambda,\theta,p)$ and $APPROX.2$ we may obtain an approximation ratio of $3.438$ for unsplittable Cu-VRPSD, it would be interesting to formally establish this approximation guarantee. Notably, both algorithms rely on a two-point distributed algorithm deployment. We suspect that a much better approximation ratio could be attained by using a multi-point distributed algorithm deployment, and we leave this as an open question for future research. 

\appendix
\section{Omitted Proofs}\label{omitted}
\subsection{Proof of Lemma~\ref{l0}}
\begingroup
\def\thelemma{\ref{l0}}
\begin{lemma}
For any $i\in [n+1]$, it holds $L_{i-1}=L_0 + \ceil{\frac{\sum_{j=1}^{i-1}h_j\cdot d_j-L_0}{\lambda-\delta}}\cdot (\lambda-\delta)-\sum_{j=1}^{i-1}h_j\cdot d_j$, and moreover, $L_{i-1}\sim U[0,\lambda-\delta)$, conditioned on $\chi=d$.
\end{lemma}
\endgroup
\begin{proof}
Since $L_0\sim U[0,\lambda-\delta)$, the lemma holds for $i=1$. Assume that the equality holds for $i=i'\geq 1$, i.e., $L_{i'-1}=L_0 + \ceil{\frac{\sum_{j=1}^{i'-1}h_j\cdot d_j-L_0}{\lambda-\delta}}\cdot (\lambda-\delta)-\sum_{j=1}^{i'-1}h_j\cdot d_j$. Note that we have $0\leq L_{i'-1}<\lambda-\delta$. Next, we consider $L_{i'}$.

\textbf{Case~1: $d_{i'}\leq \lambda$.} We have $h_{i'}=1$. By Lines~\ref{case1ofalg1}, \ref{case1ofalg1+}, and \ref{case1ofalg1++}, we have $L_{i'}=L_{i'-1}+\ceil{\frac{d_{i'}-L_{i'-1}}{\lambda-\delta}}\cdot (\lambda-\delta)-d_{i'}$.
Hence, we have $L_{i'}\geq 0$ and $L_{i'}<\lambda-\delta$.
Therefore, we have $L_0 + \ceil{\frac{\sum_{j=1}^{i'-1}h_j\cdot d_j-L_0}{\lambda-\delta}}\cdot(\lambda-\delta)-\sum_{j=1}^{i'}h_j\cdot d_j+\ceil{\frac{d_{i'}-L_{i'-1}}{\lambda-\delta}}\cdot (\lambda-\delta)\geq 0$ and $L_0 + \ceil{\frac{\sum_{j=1}^{i'-1}h_j\cdot d_j-L_0}{\lambda-\delta}}\cdot(\lambda-\delta)-\sum_{j=1}^{i'}h_j\cdot d_j+\ceil{\frac{d_{i'}-L_{i'-1}}{\lambda-\delta}}\cdot (\lambda-\delta)< \lambda-\delta$. Alternatively, we have $\ceil{\frac{\sum_{j=1}^{i'-1}h_j\cdot d_j-L_0}{\lambda-\delta}}+\ceil{\frac{d_{i'}-L_{i'-1}}{\lambda-\delta}}-1< \frac{\sum_{j=1}^{i'}h_j\cdot d_j-L_0}{\lambda-\delta}\leq \ceil{\frac{\sum_{j=1}^{i'-1}h_j\cdot d_j-L_0}{\lambda-\delta}}+\ceil{\frac{d_{i'}-L_{i'-1}}{\lambda-\delta}}$, and hence $\ceil{\frac{\sum_{j=1}^{i'}h_j\cdot d_j-L_0}{\lambda-\delta}}=\ceil{\frac{\sum_{j=1}^{i'-1}h_j\cdot d_j-L_0}{\lambda-\delta}}+\ceil{\frac{d_{i'}-L_{i'-1}}{\lambda-\delta}}$. Therefore, we have $L_{i'}=L_{i'-1}+\ceil{\frac{d_{i'}-L_{i'-1}}{\lambda-\delta}}\cdot (\lambda-\delta)-d_{i'}=L_0+\ceil{\frac{\sum_{j=1}^{i'-1}h_j\cdot d_j-L_0}{\lambda-\delta}}\cdot(\lambda-\delta)-\sum_{j=1}^{i'}h_j\cdot d_j+\ceil{\frac{d_{i'}-L_{i'-1}}{\lambda-\delta}}\cdot (\lambda-\delta)=L_0 + \ceil{\frac{\sum_{j=1}^{i'}h_j\cdot d_j-L_0}{\lambda-\delta}}\cdot(\lambda-\delta)-\sum_{j=1}^{i'}h_j\cdot d_j$.

\textbf{Case~2: $d_i>\lambda$.} We have $h_i=0$, and hence $\sum_{j=1}^{i'-1}h_j\cdot d_j=\sum_{j=1}^{i'}h_j\cdot d_j$. By Line~\ref{case2ofalg1}, we have $L_{i'}=L_{i'-1}=L_0+\ceil{\frac{\sum_{j=1}^{i'-1}h_j\cdot d_j-L_0}{\lambda-\delta}}\cdot(\lambda-\delta)-\sum_{j=1}^{i'-1}h_j\cdot d_j=L_0 + \ceil{\frac{\sum_{j=1}^{i'}h_j\cdot d_j-L_0}{\lambda-\delta}}\cdot(\lambda-\delta)-\sum_{j=1}^{i'}h_j\cdot d_j$.

In both cases, we have $L_{i'}=L_0 + \ceil{\frac{\sum_{j=1}^{i'}h_j\cdot d_j-L_0}{\lambda-\delta}}\cdot(\lambda-\delta)-\sum_{j=1}^{i'}h_j\cdot d_j$. By induction, the equality holds for any $i\in [n+1]$. 

For any $i\in [n+1]$, we have $L_{i-1}=L_0 + \ceil{\frac{\sum_{j=1}^{i-1}h_j\cdot d_j-L_0}{\lambda-\delta}}\cdot(\lambda-\delta)-\sum_{j=1}^{i-1}h_j\cdot d_j$. 
Assume that $(\sum_{j=1}^{i-1}h_j\cdot d_j)\bmod(\lambda-\delta)=L'$, which is fixed conditioned on $\chi=d$. We have $L_{i-1}=\lambda-\delta+L_0-L'\in [\lambda-\delta-L',\lambda-\delta)$ when $L_0\in[0,L')$, and $L_{i-1}=L_0-L'\in [0, \lambda-\delta-L')$ when $L_0\in [L',\lambda-\delta)$. The relationship between $L_0$ and $L_{i-1}$ is bijective. Since $L_0\sim U[0,\lambda-\delta)$, we also obtain $L_{i-1}\sim U[0,\lambda-\delta)$, conditioned on $\chi=d$.
\end{proof}

\subsection{Proof of Lemma~\ref{l2}}
\begingroup
\def\thelemma{\ref{l2}}
\begin{lemma}
Conditioned on $\chi=d$, when serving each customer $v_i$ in $ALG.1(\lambda,\delta)$, the expected cumulative cost of the vehicle due to the possible additional visit(s) to $v_0$ is
\begin{itemize}
    \item $a\cdot\frac{2d_i}{\lambda-\delta}\cdot l_i+b\cdot \frac{(\lambda+\delta)\cdot d_i-d^2_i}{\lambda-\delta}\cdot l_i$ if $d_i\leq\delta$;
    \item $a\cdot\frac{4d_i-2\delta}{\lambda-\delta}\cdot l_i+b\cdot \frac{d^2_i+(\lambda-\delta)\cdot d_i}{\lambda-\delta}\cdot l_i$ if $\delta<d_i\leq\lambda-\delta$;
    \item $a\cdot\frac{2d_i+2\lambda-4\delta}{\lambda-\delta}\cdot l_i+b\cdot \frac{2d^2_i-(\lambda+\delta)\cdot d_i+\lambda^2-\delta^2}{\lambda-\delta}\cdot l_i$ if $\lambda-\delta<d_i\leq\lambda$;
    \item $a\cdot2\cdot l_i+b\cdot d_i\cdot l_i$ if $\lambda<d_i\leq 1$.
\end{itemize}
\end{lemma}
\endgroup
\begin{proof}
If $d_i>\lambda$, By Lines~\ref{case2ofalg1} and \ref{clean}, the vehicle incurs one additional visit to $v_0$, where the vehicle carries $d_i$ units of goods from $v_0$ to $v_i$ and $0$ units of goods from $v_i$ to $v_0$. So, the expected cumulative cost is $a\cdot2\cdot l_i+b\cdot d_i\cdot l_i$.
Next, we consider $d_i\leq\lambda$. 

By Lemma \ref{l0}, the vehicle carries $(L_{i-1},\delta)$ units of goods when traveling along $v_{i-1}v_i$ in $ALG.1(\lambda,\delta)$, and it holds that $L_{i-1}\sim U[0,\lambda-\delta)$.
Hence, the vehicle incurs one additional visits to $v_0$ with a probability of $\PP{d_i-\delta\leq L_{i-1}=x<d_i}$, and incurs two additional visits to $v_0$ with a probability of $\PP{0\leq L_{i-1}=x<d_i-\delta}$ (recall Figure~\ref{fig1}). We consider the following three cases.

\textbf{Case~1: $d_i\leq\delta$.} The vehicle incurs at most one additional visit to $v_0$.
If the vehicle incurs one additional visit to $v_0$, By Line~\ref{addone}, the vehicle carries $(0,\delta-(d_i-L_{i-1}))$ units of goods from $v_i$ to $v_0$, and $(L_{i-1}+\ceil{\frac{d_i-L_{i-1}}{\lambda-\delta}}\cdot (\lambda-\delta)-d_i,\delta)=(L_{i-1}+(\lambda-\delta)-d_i,\delta)$ units of goods from $v_0$ to $v_i$. 
So, the cumulative cost is $a\cdot 2\cdot l_i + b\cdot (\delta-(d_i-L_{i-1})+L_{i-1}+(\lambda-\delta)-d_i+\delta)\cdot l_i=a\cdot 2\cdot l_i + b\cdot (2L_{i-1}-2d_i+\lambda+\delta)\cdot l_i$. Since $L_{i-1}\sim U[0,\lambda-\delta)$ and $d_i\leq\delta\leq\lambda-\delta$, the expected cumulative cost is
\begin{align*}
&\int_{0}^{\min\{d_i, \lambda-\delta\}}\frac{a\cdot 2\cdot l_i + b\cdot (2x-2d_i+\lambda+\delta)\cdot l_i}{\lambda-\delta}dx\\
&=\int_{0}^{d_i}\frac{a\cdot 2\cdot l_i + b\cdot (2x-2d_i+\lambda+\delta)\cdot l_i}{\lambda-\delta}dx=a\cdot\frac{2d_i}{\lambda-\delta}\cdot l_i+b\cdot \frac{(\lambda+\delta)\cdot d_i-d^2_i}{\lambda-\delta}\cdot l_i.    
\end{align*}

\textbf{Case~2: $\delta<d_i\leq\lambda-\delta$.} The vehicle incurs at most two additional visits to $v_0$. Similarly, if the vehicle incurs one additional visit to $v_0$, By Line~\ref{addone}, the vehicle carries $(0,\delta-(d_i-L_{i-1}))$ units of goods from $v_i$ to $v_0$, and $(L_{i-1}+\ceil{\frac{d_i-L_{i-1}}{\lambda-\delta}}\cdot (\lambda-\delta)-d_i,\delta)=(L_{i-1}+(\lambda-\delta)-d_i,\delta)$ units of goods from $v_0$ to $v_i$, and the cumulative cost is $a\cdot 2\cdot l_i + b\cdot (\delta-(d_i-L_{i-1})+L_{i-1}+(\lambda-\delta)-d_i+\delta)\cdot l_i=a\cdot 2\cdot l_i + b\cdot (2L_{i-1}-2d_i+\lambda+\delta)\cdot l_i$. If the vehicle incurs two additional visits to $v_0$, By Lines~\ref{addtwo1} and~\ref{addtwo2}, the vehicle carries $(L_{i-1},\delta)$ units of goods from $v_i$ to $v_0$, $(d_i-\delta,\delta)$ units of goods from $v_0$ to $v_i$, $(0,0)$ units of goods from $v_i$ to $v_0$, and $(L_{i-1}+\ceil{\frac{d_i-L_{i-1}}{\lambda-\delta}}\cdot (\lambda-\delta)-d_i,\delta)=(L_{i-1}+(\lambda-\delta)-d_i,\delta)$ units of goods from $v_0$ to $v_i$, and the cumulative cost is $a\cdot 4\cdot l_i + b\cdot (L_{i-1}+\delta+d_i+0+L_{i-1}+(\lambda-\delta)-d_i+\delta)\cdot l_i=a\cdot 4\cdot l_i + b\cdot (2L_{i-1}+\lambda+\delta)\cdot l_i$. Hence, the expected cumulative cost is
\begin{align*}
&\int_{d_i-\delta}^{\min\{d_i, \lambda-\delta\}}\frac{a\cdot 2\cdot l_i + b\cdot (2x-2d_i+\lambda+\delta)\cdot l_i}{\lambda-\delta}dx+\int_{0}^{d_i-\delta}\frac{a\cdot 4\cdot l_i + b\cdot (2x+\lambda+\delta)\cdot l_i}{\lambda-\delta}dx\\
&=\int_{d_i-\delta}^{d_i}\frac{a\cdot 2\cdot l_i + b\cdot (2x-2d_i+\lambda+\delta)\cdot l_i}{\lambda-\delta}dx+\int_{0}^{d_i-\delta}\frac{a\cdot 4\cdot l_i + b\cdot (2x+\lambda+\delta)\cdot l_i}{\lambda-\delta}dx\\
&=\frac{a\cdot2\cdot\delta\cdot l_i+b\cdot\lambda\cdot\delta\cdot l_i}{\lambda-\delta}+\frac{a\cdot4\cdot(d_i-\delta)\cdot l_i+b\cdot(d^2_i+(\lambda-\delta)\cdot d_i-\delta\cdot\lambda)\cdot l_i}{\lambda-\delta}\\
&=a\cdot\frac{4d_i-2\delta}{\lambda-\delta}\cdot l_i+b\cdot \frac{d^2_i+(\lambda-\delta)\cdot d_i}{\lambda-\delta}\cdot l_i.    
\end{align*}

\textbf{Case~3: $\lambda-\delta<d_i\leq\lambda$.} The vehicle incurs at most two additional visits to $v_0$. If the vehicle incurs one additional visit to $v_0$, By Line~\ref{addone}, the vehicle carries $(0,\delta-(d_i-L_{i-1}))$ units of goods from $v_i$ to $v_0$, and $(L_{i-1}+\ceil{\frac{d_i-L_{i-1}}{\lambda-\delta}}\cdot (\lambda-\delta)-d_i,\delta)$ units of goods from $v_0$ to $v_i$, and the cumulative cost is $a\cdot 2\cdot l_i + b\cdot (\delta-(d_i-L_{i-1})+L_{i-1}+\ceil{\frac{d_i-L_{i-1}}{\lambda-\delta}}\cdot(\lambda-\delta)-d_i+\delta)\cdot l_i=a\cdot 2\cdot l_i + b\cdot (2L_{i-1}-2d_i+2\delta+\ceil{\frac{d_i-L_{i-1}}{\lambda-\delta}}\cdot(\lambda-\delta))\cdot l_i$. If the vehicle incurs two additional visits to $v_0$, By Lines~\ref{addtwo1} and~\ref{addtwo2}, the vehicle carries $(L_{i-1},\delta)$ units of goods from $v_i$ to $v_0$, $(d_i-\delta,\delta)$ units of goods from $v_0$ to $v_i$, $(0,0)$ units of goods from $v_i$ to $v_0$, and $(L_{i-1}+\ceil{\frac{d_i-L_{i-1}}{\lambda-\delta}}\cdot (\lambda-\delta)-d_i,\delta)$ units of goods from $v_0$ to $v_i$, and the cumulative cost is $a\cdot 4\cdot l_i + b\cdot (L_{i-1}+\delta+d_i+0+L_{i-1}+\ceil{\frac{d_i-L_{i-1}}{\lambda-\delta}}\cdot(\lambda-\delta)-d_i+\delta)\cdot l_i=a\cdot 4\cdot l_i + b\cdot (2L_{i-1}+2\delta+\ceil{\frac{d_i-L_{i-1}}{\lambda-\delta}}\cdot(\lambda-\delta))\cdot l_i$. Hence, the expected cumulative cost is
\begin{align*}
&\int_{d_i-\delta}^{\min\{d_i, \lambda-\delta\}}\frac{a\cdot 2\cdot l_i + b\cdot (2x-2d_i+2\delta+\ceil{\frac{d_i-x}{\lambda-\delta}}\cdot(\lambda-\delta))\cdot l_i}{\lambda-\delta}dx\\
&\quad+\int_{0}^{d_i-\delta}\frac{a\cdot 4\cdot l_i + b\cdot (2x+2\delta+\ceil{\frac{d_i-x}{\lambda-\delta}}\cdot(\lambda-\delta))\cdot l_i}{\lambda-\delta}dx\\
&=\int_{d_i-\delta}^{\lambda-\delta}\frac{a\cdot 2\cdot l_i + b\cdot (2x-2d_i+2\delta+1\cdot(\lambda-\delta))\cdot l_i}{\lambda-\delta}dx\\
&\quad+\int_{0}^{d_i+\delta-\lambda}\frac{a\cdot 4\cdot l_i + b\cdot (2x+2\delta+2\cdot(\lambda-\delta))\cdot l_i}{\lambda-\delta}dx\\
&\quad+\int_{d_i+\delta-\lambda}^{d_i-\delta}\frac{a\cdot 4\cdot l_i + b\cdot (2x+2\delta+1\cdot(\lambda-\delta))\cdot l_i}{\lambda-\delta}dx\\
&=\frac{a\cdot2\cdot (\lambda-d_i)\cdot l_i+b\cdot(d^2_i+(\delta-3\lambda)\cdot d_i+(2\lambda^2-\lambda\cdot\delta))}{\lambda-\delta}\\
&\quad+\frac{a\cdot 4\cdot (d_i+\delta-\lambda)\cdot l_i+b\cdot(d^2_i+2\delta\cdot d_i+\delta^2-\lambda^2)}{\lambda-\delta}\\
&\quad+\frac{a\cdot4\cdot(\lambda-2\delta)\cdot l_i+b\cdot((2\lambda-4\delta)\cdot d_i+\delta\cdot\lambda-2\delta^2)}{\lambda-\delta}\\
&=a\cdot\frac{2d_i+2\lambda-4\delta}{\lambda-\delta}\cdot l_i+b\cdot \frac{2d^2_i-(\lambda+\delta)\cdot d_i+\lambda^2-\delta^2}{\lambda-\delta}\cdot l_i,  
\end{align*}
where the first equality follows from that $\ceil{\frac{d_i-x}{\lambda-\delta}}=1$ if $d_i+\delta-\lambda\leq x\leq \lambda-\delta$ and $\ceil{\frac{d_i-x}{\lambda-\delta}}=2$ if $0\leq x<d_i+\delta-\lambda$ since in our setting $\delta\leq \lambda-\delta$ and in this case $\lambda-\delta<d_i\leq\lambda$.
\end{proof}

\subsection{Proof of Lemma~\ref{pr}}
\begingroup
\def\thelemma{\ref{pr}}
\begin{lemma}
It holds that $\EE{Cu(\S')}\leq \ln2\cdot Cu(\T^*)$, $\PP{v\notin V^{**}}=1$ for any $v\in V\setminus V^*$, and $\PP{v\notin V^{**}}\leq 1/2$ for any $v\in V^*$.
\end{lemma}
\endgroup
\begin{proof}
Since the linear program is a relaxation of unsplittable Cu-VRP, we have $\sum_{S\in\S}Cu(S)\cdot x_S\leq\OPT$. Moreover, since each $S\in\S$ is selected into $\S'$ with a probability of $\min\{\ln2\cdot x_S,1\}$, we can obtain $\EE{Cu(\S')}=\sum_{S\in\S}Cu(S)\cdot \min\{\ln2\cdot x_S,1\}\leq\ln2\cdot\sum_{S\in\S}Cu(S)\cdot x_S\leq\ln2\cdot\OPT$.

Since randomized rounding is used only for customers in $V^{*}$, we have $\PP{v\notin V^{**}}=1$ for any $v\in V\setminus V^*$. Next, we consider $v\in V^*$.

Since $\PP{v\notin V^{**}}=\prod_{S\in \S: v\in S}\PP{S\notin\S'}$, we may assume that $\min\{\ln2\cdot x_S,1\}=\ln2\cdot x_S$ for any $S\in \S$ with $v\in S$ since otherwise $\PP{v\notin V^{**}}=0\leq1/2$ holds trivially. Then, we have $\PP{v\notin V^{**}}=\prod_{S\in \S: v\in S}(1-\ln2\cdot x_S)\leq e^{-\sum_{S\in \S: v\in S}\ln2\cdot x_S}\leq e^{-\ln2}=1/2$ since $1-x\leq e^{-x}$ for any $x\in[0,1]$ and we have $\sum_{\substack{S\in \S: v\in S}}x_S \geq 1$ for any $v\in V^*$ by the linear program.
\end{proof}

\subsection{Proof of Lemma~\ref{t5}}
\begingroup
\def\thelemma{\ref{t5}}
\begin{lemma}
For unsplittable Cu-VRP with any $\lambda\in(0,1]$, $\delta\in(0,\lambda/2]$, and $1/\delta\in\mathbb{N}$, $ALG.3(\lambda,\delta)$ generates a solution $\T$ with an expected cumulative cost of
\[
\ln2\cdot Cu(\T^*)+\frac{\gamma\cdot\lrA{\alpha\cdot\sigma+\frac{1}{\lambda-\delta}}+\frac{\lambda}{2}\cdot\lrA{\alpha\cdot\sigma+\frac{1}{\lambda-\delta}}}{\gamma\cdot\max\lrC{\sigma,\ 1}+0.5}\cdot LB.
\]
\end{lemma}
\endgroup
\begin{proof}
Recall that $\T=\T'\cup\T'''$. The set of tours $\T'$ are obtained by the selected sets in $\S'$ and $Cu(\T')\leq Cu(\S')$ by the triangle inequality. By Lemma~\ref{pr}, we have $\EE{Cu(\T')}\leq \ln2\cdot\OPT$. Next, we consider $\T''$, which is obtained by calling $ALG.1(\lambda,\delta)$ for customers in $V'\setminus V^{**}$.

Define small customers $V_s\coloneqq\{v_i\in V'\mid 0<d_i\leq \delta\}$, big customers $V_b\coloneqq\{v_i\in V'\mid \delta<d_i\leq \lambda-\delta\}$, large customers $V_l\coloneqq\{v_i\in V'\mid \lambda-\delta<d_i\leq \lambda\}$, and huge customers $V_h\coloneqq\{v_i\in V'\mid \delta<\lambda\leq 1\}$. Note that $V^*=V_b\cup V_l\cup V_h$. 

By Lemma~\ref{pr}, we know that $\PP{v\in V_s\cap (V'\setminus V^{**})}=\PP{v\notin V^{**}}=1$ for each $v\in V_s$, and
$\PP{v\in V_b\cap (V'\setminus V^{**})}=\PP{v\notin V^{**}}\leq 1/2$ for each $v\in V_b$. Similarly, we have $\PP{v\in V_l\cap (V'\setminus V^{**})}\leq1/2$ for each $v\in V_l$ and $\PP{v\in V_h\cap (V'\setminus V^{**})}\leq1/2$ for each $v\in V_h$.

By $ALG.1(\lambda,\delta)$, there are two kinds of tours in $\T''$, where the first kind of tour satisfying customers in Lines~\ref{case1ofalg1}, \ref{case1ofalg1+}, \ref{addtwo1}, and \ref{addtwo2} delivers at most $\lambda$ units of goods, and the second tour satisfying a customer $v_i$ in Line~\ref{clean} delivers exactly $d_i$ units of goods. Denote the set of the first kind and the second kind of tours by $\T''_1$ and $\T''_2$, respectively. By Lemma~\ref{t1}, the expected vehicle cost of $\T''_1$ is
\begin{equation}\label{eq3}
\begin{split}
&\EEE{Cu_1(\T''_1)}\\
&=\frac{\EEE{a\cdot\lrA{\alpha\cdot\tau+\sum_{v_i\in V_s\cap (V'\setminus V^{**})}\frac{2d_i}{\lambda-\delta}\cdot l_i+\sum_{v_i\in V_b\cap (V'\setminus V^{**})}\frac{4d_i-2\delta}{\lambda-\delta}\cdot l_i+\sum_{v_i\in V_l\cap (V'\setminus V^{**})}\frac{2d_i+2\lambda-4\delta}{\lambda-\delta}\cdot l_i}}}{LB}\cdot LB\\
&\leq\frac{{a\cdot\lrA{\alpha\cdot\tau+\sum_{v_i\in V_s}\frac{2d_i}{\lambda-\delta}\cdot l_i+\sum_{v_i\in V_b}\frac{2d_i-\delta}{\lambda-\delta}\cdot l_i+\sum_{v_i\in V_l}\frac{d_i+\lambda-2\delta}{\lambda-\delta}\cdot l_i}}}{LB}\cdot LB\\
&=\frac{\gamma\cdot\lrA{\alpha\cdot\sigma+\int^{\delta}_0\frac{x}{\lambda-\delta}dF(x)+\int^{\lambda-\delta}_{\delta}\frac{x-\delta/2}{\lambda-\delta}dF(x)+\int^\lambda_{\lambda-\delta}\frac{x/2+\lambda/2-\delta}{\lambda-\delta}dF(x)}}{\gamma\cdot\max\lrC{\sigma,\ 1}+0.5}\cdot LB\\
&\leq\frac{\gamma\cdot\lrA{\alpha\cdot\sigma+\int^{\delta}_0\frac{x}{\lambda-\delta}dF(x)+\int^{\lambda}_{\delta}\frac{x}{\lambda-\delta}dF(x)}}{\gamma\cdot\max\lrC{\sigma,\ 1}+0.5}\cdot LB=\frac{\gamma\cdot\lrA{\alpha\cdot\sigma+\int^{\lambda}_0\frac{x}{\lambda-\delta}dF(x)}}{\gamma\cdot\max\lrC{\sigma,\ 1}+0.5}\cdot LB,
\end{split}
\end{equation}
where the first inequality follows from $\PP{v\in V_s\cap (V'\setminus V^{**})}=1$ for each $v\in V_s$,
$\PP{v\in V_b\cap (V'\setminus V^{**})}\leq 1/2$ for each $v\in V_b$, and $\PP{v\in V_l\cap (V'\setminus V^{**})}\leq 1/2$ for each $v\in V_l$ by the previous analysis.

The expected cumulative cost of $\T''_2$ is
\begin{equation}\label{eq4}
\begin{split}
&\EEE{Cu(\T''_2)}=\frac{\EEE{a\cdot{\sum_{v_i\in V_h}2\cdot l_i}+b\cdot\sum_{v_i\in V_h}d_i\cdot l_i}}{LB}\cdot LB\\
&\leq\frac{{a\cdot{\sum_{v_i\in V_h}l_i}+b\cdot\sum_{v_i\in V_h}\frac{ d_i}{2}\cdot l_i}}{LB}\cdot LB=\frac{\gamma\cdot\int^1_{\lambda}\frac{1}{2}dF(x)+\int^1_{\lambda}\frac{x}{4}dF(x)}{\gamma\cdot\max\lrC{\sigma,\ 1}+0.5}\cdot LB.
\end{split}
\end{equation}

The set of tours $\T'''$ is obtained by optimizing the set of tours $\T''$. 
Denote the set of the optimized tours in $\T''_1$ by $\T'''_1$ and in $\T''_2$ by $\T'''_2$. Thus, we have $\T'''=\T'''_1\cup \T'''_2$.

\begin{claim}\label{claim}
It holds that $Cu(\T'''_2)=Cu(\T''_2)$, and $Cu(\T'''_1)\leq \lrA{1+\frac{\lambda}{2\gamma}}\cdot Cu_1(\T''_1)$.
\end{claim}
\begin{proof}[Claim proof]
The second kind of tour in $\T''_2$ can not be improved using the optimization method in Line~\ref{opt} since it is a single tour satisfying a single customer. Hence, we have $Cu(\T'''_2)=Cu(\T''_2)$. Next, we consider the first kind of tour.

For each tour $T=v_0v_{1}\dots v_{k}v_0\in \T''_1$, Line~\ref{opt} ensures that the load of the vehicle equals the delivered units of goods on $T$. Hence, the vehicle carries $\sum_{j=1}^{i-1}d_j$ units of goods when traveling along the edge $v_{i-1}v_{i\bmod (k+1)}$ for each $i\in[k+1]$, and then we have 
\[
Cu(T)=a\cdot\sum_{i=1}^{k+1}w(v_{i-1},v_{i\bmod (k+1)})+b\cdot\sum_{i=1}^{k+1}\lrA{\sum_{j=1}^{i-1}d_j}\cdot w(v_{i-1},v_{i\bmod (k+1)}).
\]

Consider the tour with the opposite direction, i.e., $T'=v_0v_{k}\dots v_{i}v_0\in \T''$. we have 
\[
Cu(T')=a\cdot\sum_{i=1}^{k+1}w(v_{i-1},v_{i\bmod (k+1)})+b\cdot\sum_{i=1}^{k+1}\lrA{\sum_{j=1}^{k}d_j-\sum_{j=1}^{i-1}d_j}\cdot w(v_{i-1},v_{i\bmod (k+1)}).
\]

Therefore, the better one has a cumulative cost of at most 
\begin{align*}
&a\cdot\sum_{i=1}^{k+1}w(v_{i-1},v_{i\bmod (k+1)})+b\cdot\frac{1}{2}\cdot\sum_{j=1}^{k}d_j\cdot\sum_{i=1}^{k+1}w(v_{i-1},v_{i\bmod (k+1)})\\
&\leq a\cdot\sum_{i=1}^{k+1}w(v_{i-1},v_{i\bmod (k+1)})+b\cdot\frac{\lambda}{2}\cdot\sum_{i=1}^{k+1}w(v_{i-1},v_{i\bmod (k+1)})=\lrA{1+\frac{\lambda}{2\gamma}}\cdot Cu_1(T),
\end{align*}
where the inequality follows from that the first kind of tour delivers at most $\lambda$ units of goods, and the equality from $\gamma=a/b$ and $Cu_1(T)=a\cdot\sum_{i=1}^{k+1}w(v_{i-1},v_{i\bmod (k+1)})$.

Therefore, we have $Cu(\T'''_1)\leq \lrA{1+\frac{\lambda}{2\gamma}}\cdot Cu_1(\T''_1)$.
\end{proof}

By (\ref{eq3}), (\ref{eq4}), and the Claim~\ref{claim}, the expected cumulative cost of $\T'''$ is at most 
\begin{align*}
&\lrA{1+\frac{\lambda}{2\gamma}}\cdot\frac{\gamma\cdot\lrA{\alpha\cdot\sigma+\int^{\lambda}_0\frac{x}{\lambda-\delta}dF(x)}}{\gamma\cdot\max\lrC{\sigma,\ 1}+0.5}\cdot LB+\frac{\gamma\cdot\int^1_{\lambda}\frac{1}{2}dF(x)+\int^1_{\lambda}\frac{x}{4}dF(x)}{\gamma\cdot\max\lrC{\sigma,\ 1}+0.5}\cdot LB\\
&=\frac{\gamma\cdot\lrA{\alpha\cdot\sigma+\int^{\lambda}_0\frac{x}{\lambda-\delta}dF(x)+\int^1_{\lambda}\frac{1}{2}dF(x)}+\frac{\lambda}{2}\cdot\lrA{\alpha\cdot\sigma+\int^{\lambda}_0\frac{x}{\lambda-\delta}dF(x)+\int^1_{\lambda}\frac{\lambda\cdot x}{2}dF(x)}}{\gamma\cdot\max\lrC{\sigma,\ 1}+0.5}\cdot LB\\
&\leq \frac{\gamma\cdot\lrA{\alpha\cdot\sigma+\int^{\lambda}_0\frac{x}{\lambda-\delta}dF(x)+\int^1_{\lambda}\frac{x}{\lambda-\delta}dF(x)}+\frac{\lambda}{2}\cdot\lrA{\alpha\cdot\sigma+\int^{\lambda}_0\frac{x}{\lambda-\delta}dF(x)+\int^1_{\lambda}\frac{x}{\lambda-\delta}dF(x)}}{\gamma\cdot\max\lrC{\sigma,\ 1}+0.5}\cdot LB\\
&=\frac{\gamma\cdot\lrA{\alpha\cdot\sigma+\frac{1}{\lambda-\delta}}+\frac{\lambda}{2}\cdot\lrA{\alpha\cdot\sigma+\frac{1}{\lambda-\delta}}}{\gamma\cdot\max\lrC{\sigma,\ 1}+0.5}\cdot LB,
\end{align*}
where the inequality follows from $\frac{\lambda\cdot x}{2}\leq \frac{1}{2}\leq \frac{x}{\lambda-\delta}$ for any $x\in[\lambda, 1]$.
\end{proof}
\subsection{Proof of Lemma~\ref{t7}}
\begingroup
\def\thelemma{\ref{t7}}
\begin{lemma}
For unsplittable Cu-VRP with any $\lambda\in(0,1]$, $ALG.4(\lambda)$ generates a solution $\T$ with an expected cumulative cost of
\[
\frac{\gamma\cdot\lrA{\alpha\cdot\sigma+\int^{\lambda}_{0}\frac{2x}{\lambda}dF(x)+\int_\lambda^1 1dF(x)}+\lrA{\frac{\lambda}{2}\cdot\alpha\cdot\sigma+\int^{\lambda}_{0}\frac{x^2/2+\lambda\cdot x}{2\lambda}dF(x)+\int_\lambda^1\frac{x}{2}dF(x)}}{\gamma\cdot\max\lrC{\sigma,\ 1}+0.5}\cdot LB.
\]
\end{lemma}
\endgroup
\begin{proof}
Firstly, $ALG.4(\lambda)$ calls $ALG.1(\lambda,0)$ to obtain a set of tours $\T'$. As in the proof of Lemma~\ref{t5}, there are two kinds of tours in $\T'$, where the first kind of tour satisfying customers in Lines~\ref{case1ofalg1}, \ref{case1ofalg1+}, \ref{addtwo1}, and \ref{addtwo2} delivers at most $\lambda$ units of goods, and the second tour satisfying a customer $v_i$ in Line~\ref{clean} delivers exactly $d_i$ units of goods. Denote the set of the first kind and the second kind of tours by $\T'_1$ and $\T'_2$, respectively.

By Lemma~\ref{t1}, we can obtain 
\begin{align*}
\EEE{Cu(\T'_1)}\leq\frac{\gamma\cdot\lrA{\alpha\cdot\sigma+\int^{\lambda}_{0}\frac{2x}{\lambda}dF(x)}+\lrA{\frac{\lambda}{2}\cdot\alpha\cdot\sigma+\int^{\lambda}_{0}\frac{x^2+\lambda\cdot x}{2\lambda}dF(x)}}{\gamma\cdot\max\lrC{\sigma,\ 1}+0.5}\cdot LB,
\end{align*}
and 
\begin{equation}\label{eq5}
\EEE{Cu(\T'_2)}\leq\frac{\gamma\cdot\int_\lambda^1 1dF(x)+\int_\lambda^1\frac{x}{2}dF(x)}{\gamma\cdot\max\lrC{\sigma,\ 1}+0.5}\cdot LB.
\end{equation}

The set of tours $\T$ is obtained by optimizing the set of tours $\T'$. 
Denote the set of the optimized tours in $\T'_1$ by $\T_1$ and in $\T'_2$ by $\T_2$. 

\begin{claim}\label{claim+}
It holds that $Cu(\T_2)=Cu(\T'_2)$ and 
\[\EEE{Cu(\T'_1)}\leq\frac{\gamma\cdot\lrA{\alpha\cdot\sigma+\int^{\lambda}_{0}\frac{2x}{\lambda}dF(x)}+\lrA{\frac{\lambda}{2}\cdot\alpha\cdot\sigma+\int^{\lambda}_{0}\frac{x^2/2+\lambda\cdot x}{2\lambda}dF(x)}}{\gamma\cdot\max\lrC{\sigma,\ 1}+0.5}\cdot LB.
\]
\end{claim}
\begin{proof}[Claim proof]
By the proof of Claim~\ref{claim}, we can obtain $Cu(\T_2)=Cu(\T'_2)$. To prove the inequality, we need a different method from that used in Claim~\ref{claim}.

Recall Case~2 in the proof Lemma~\ref{l2}.
Since $\delta=0$, if the vehicle incurs some additional visits to $v_0$ when serving a customer $v_i$, it must incur two additional visits to $v_0$. Moreover, the trips by the additional visits to $v_0$ belong to some tours in $\T'_1$.
If the vehicle incurs two additional visits, 
the vehicle carries $(L_{i-1},0)$ units of goods from $v_i$ to $v_0$, $(d_i,0)$ units of goods from $v_0$ to $v_i$, $(0,0)$ units of goods from $v_i$ to $v_0$, and $(L_{i-1}+\lambda-d_i,0)$ units of goods from $v_0$ to $v_i$, and the cumulative cost is $a\cdot 4\cdot l_i + b\cdot (L_{i-1}+d_i+0+L_{i-1}+\lambda-d_i)\cdot l_i=a\cdot 4\cdot l_i + b\cdot (2L_{i-1}+\lambda)\cdot l_i$. 
Specifically, we can assume that the vehicle carries $(0,0)$ instead of $(L_{i-1},0)$ units of goods from $v_i$ to $v_0$ in the first visit to $v_0$, as we will optimize the tours in $\T'_1$ so that the load of the vehicle equals the delivered units of goods. Hence, the cumulative cost reduces to $a\cdot 4\cdot l_i + b\cdot (0+d_i+0+L_{i-1}+\lambda-d_i)\cdot l_i=a\cdot 4\cdot l_i + b\cdot (L_{i-1}+\lambda)\cdot l_i$. Then, following the proof of Lemma~\ref{t1}, we can obtain 
\[\EEE{Cu(\T_1)}\leq\frac{\gamma\cdot\lrA{\alpha\cdot\sigma+\int^{\lambda}_{0}\frac{2x}{\lambda}dF(x)}+\lrA{\frac{\lambda}{2}\cdot\alpha\cdot\sigma+\int^{\lambda}_{0}\frac{x^2+2\lambda\cdot x}{4\lambda}dF(x)}}{\gamma\cdot\max\lrC{\sigma,\ 1}+0.5}\cdot LB.
\]
This finishes the proof.
\end{proof}
By (\ref{eq5}) and Claim~\ref{claim+}, the theorem is proved.
\end{proof}

\subsection{Proof of Theorem~\ref{t8}}
\begingroup
\def\thetheorem{\ref{t8}}
\begin{theorem}
For unsplittable Cu-VRP, when $\alpha=1.5$, we can find $(\lambda,\theta,p)$ such that the approximation ratio of $APPROX.4(\lambda,\theta,p)$ is bounded by $3.163$ for any $\gamma\in(0,0.428]$. 
Moreover, for any $\alpha\in[1,1.5]$, we can find $(\lambda,\theta,p)$ such that the approximation ratio of $APPROX.4(\lambda,\theta,p)$ is bounded by $\alpha+1+\ln2-0.029$ for any $\gamma\in(0,0.285]$.
\end{theorem}
\endgroup
\begin{proof}
By Lemma~\ref{t7}, the approximation ratio of $ALG.4(\lambda,0)$ is at most
\begin{align*}
&\frac{\gamma\cdot\lrA{\alpha\cdot\sigma+\int^{\lambda}_{0}\frac{2x}{\lambda}dF(x)+\int_\lambda^1 1dF(x)}+\lrA{\frac{\lambda}{2}\cdot\alpha\cdot\sigma+\int^{\lambda}_{0}\frac{x^2/2+\lambda\cdot x}{2\lambda}dF(x)+\int_\lambda^1\frac{x}{2}dF(x)}}{\gamma\cdot\max\lrC{\sigma,\ 1}+0.5}\\
&=\frac{\gamma\cdot\lrA{\alpha\cdot\sigma+\frac{2}{\lambda}\cdot\int_{0}^{\lambda}xdF(x)+\int_{\lambda}^{1}1dF(x)}+\lrA{\frac{\lambda}{2}\cdot\alpha\cdot\sigma+\frac{1+\mu/2\lambda}{2}\cdot\int_{0}^{\lambda}xdF(x)+\frac{1}{2}\cdot\int_{\lambda}^{1}xdF(x)}}{\gamma\cdot\max\lrC{\sigma,\ 1}+0.5}\\
&\leq \frac{\gamma\cdot\lrA{\alpha\cdot\sigma+\frac{2}{\lambda}\cdot\int_{0}^{\lambda}xdF(x)+\frac{1}{\lambda}\cdot\int_{\lambda}^{1}xdF(x)}+\lrA{\frac{\lambda}{2}\cdot\alpha\cdot\sigma+\frac{1+\mu/2\lambda}{2}\cdot\int_{0}^{\lambda}xdF(x)+\frac{1}{2}\cdot\int_{\lambda}^{1}xdF(x)}}{\gamma\cdot\max\lrC{\sigma,\ 1}+0.5}\\
&= \frac{\gamma\cdot\lrA{\alpha\cdot\sigma+\frac{1}{\lambda}\cdot\int_{0}^{\lambda}xdF(x)+\frac{1}{\lambda}}+\lrA{\frac{\lambda}{2}\cdot\alpha\cdot\sigma+\frac{\mu/2\lambda}{2}\cdot\int_{0}^{\lambda}xdF(x)+\frac{1}{2}}}{\gamma\cdot\max\lrC{\sigma,\ 1}+0.5}\\
&\leq \frac{\gamma\cdot\lrA{\alpha\cdot\sigma+\frac{2}{\lambda}}+\lrA{\frac{\lambda}{2}\cdot\alpha\cdot\sigma+\frac{\mu/2\lambda+1}{2}}}{\gamma\cdot\max\lrC{\sigma,\ 1}+0.5},
\end{align*}
where the first equality follows from the definition of $\mu$, the second equality from $\int^1_0xdF(x)=1$ by (\ref{int}), the first inequality from $\int_{\lambda}^{1}1dF(x)\leq \frac{1}{\lambda}\cdot\int_{\lambda}^{1}xdF(x)$, and the second inequality from $\int_{0}^{\lambda}xdF(x)\leq \int_{0}^{1}xdF(x)=1$.

By Lemma~\ref{t7}, the approximation ratio of $ALG.4(\theta\cdot\lambda,0)$ is at most
\begin{align*}
&\frac{\gamma\cdot\lrA{\alpha\cdot\sigma+\int_{0}^{\theta\cdot\lambda}\frac{2x}{\theta\cdot\lambda}dF(x)+\int_{\theta\cdot\lambda}^{1}1dF(x)}+\lrA{\frac{\theta\cdot\lambda}{2}\cdot\alpha\cdot\sigma+\int_{0}^{\theta\cdot\lambda}\frac{x^2/2+\theta\cdot\lambda\cdot x}{2\cdot\theta\cdot\lambda}dF(x)+\int_{\theta\cdot\lambda}^{1}\frac{x}{2}dF(x)}}{\gamma\cdot\max\lrC{\sigma,\ 1}+0.5}\\
&\leq \frac{\gamma\cdot\lrA{\alpha\cdot\sigma+\frac{2}{\theta\cdot\lambda}\cdot\int_{0}^{\theta\cdot\lambda}xdF(x)+\frac{1}{\theta\cdot\lambda}\cdot\int_{\theta\cdot\lambda}^{1}xdF(x)}+\lrA{\frac{\theta\cdot\lambda}{2}\cdot\alpha\cdot\sigma+\frac{3}{4}\cdot\int_{0}^{\theta\cdot\lambda}xdF(x)+\frac{1}{2}\cdot\int_{\theta\cdot\lambda}^{1}xdF(x)}}{\gamma\cdot\max\lrC{\sigma,\ 1}+0.5}\\
&= \frac{\gamma\cdot\lrA{\alpha\cdot\sigma+\frac{1}{\theta\cdot\lambda}\int_{0}^{\theta\cdot\lambda}xdF(x)+\frac{1}{\theta\cdot\lambda}}+\lrA{\frac{\theta\cdot\lambda}{2}\cdot\alpha\cdot\sigma+\frac{1}{4}\cdot\int_{0}^{\theta\cdot\lambda}xdF(x)+\frac{1}{2}}}{\gamma\cdot\max\lrC{\sigma,\ 1}+0.5}\\
&\leq \frac{\gamma\cdot\lrA{\alpha\cdot\sigma+\frac{1}{\theta\cdot\lambda}\cdot\frac{\lambda-\mu}{\lambda-\theta\cdot\lambda}+\frac{1}{\theta\cdot\lambda}}+\lrA{\frac{\theta\cdot\lambda}{2}\cdot\alpha\cdot\sigma+\frac{1}{4}\cdot\frac{\lambda-\mu}{\lambda-\theta\cdot\lambda}+\frac{1}{2}}}{\gamma\cdot\max\lrC{\sigma,\ 1}+0.5}\\
&= \frac{\gamma\cdot\lrA{\alpha\cdot\sigma+\frac{1}{\theta\cdot\lambda}\cdot\frac{2\lambda-\mu-\theta\cdot\lambda}{\lambda-\theta\cdot\lambda}}+\lrA{\frac{\theta\cdot\lambda}{2}\cdot\alpha\cdot\sigma+\frac{1}{4}\cdot\frac{3\lambda-\mu-2\theta\cdot\lambda}{\lambda-\theta\cdot\lambda}}}{\gamma\cdot\max\lrC{\sigma,\ 1}+0.5},
\end{align*}
where the first inequality follows from $\int_{0}^{\theta\cdot\lambda}x^2dF(x)\leq \int_{0}^{\theta\cdot\lambda}\theta\cdot\lambda\cdot xdF(x)$, the second inequality from (\ref{eq1}), and the first equality from $\int^1_0xdF(x)=1$ by (\ref{int}).

Recall that in $APPROX.4(\lambda,\theta,p)$ we call $ALG.4(\lambda,0)$ (resp., $ALG.4(\theta\cdot\lambda,0)$) with a probability of $p$ (resp., $1-p$). Hence, to erase the items related to $\mu$ in the numerators of the approximation ratios of $ALG.4(\lambda,0)$ and $ALG.4(\theta\cdot\lambda,0)$, we need to set $p$ such that 
\[
p\cdot\frac{1}{4}\cdot\frac{1}{\lambda}+(1-p)\cdot\lrA{\gamma\cdot\frac{1}{\theta\cdot\lambda}\cdot\frac{-1}{\lambda-\theta\cdot\lambda}+\frac{1}{4}\cdot\frac{-1}{\lambda-\theta\cdot\lambda}}=0.
\]
Then, we can obtain $p=\frac{\frac{1}{4(\lambda-\theta\cdot\lambda)}+\frac{\gamma}{\theta\cdot\lambda(\lambda-\theta\cdot\lambda)}}{\frac{1}{4\lambda}+\frac{1}{4(\lambda-\theta\cdot\lambda)}+\frac{\gamma}{\theta\cdot\lambda(\lambda-\theta\cdot\lambda)}}$, and the approximation ratio is $\max_{\sigma\geq 0} R(\sigma)$, where
\begin{equation}\label{ratio2}
\begin{split}
&R(\sigma)\coloneqq\frac{\gamma\cdot\lrA{\alpha\cdot\sigma+p\cdot\frac{2}{\lambda}+(1-p)\cdot\frac{1}{\theta\cdot\lambda}\cdot\frac{2\lambda-\theta\cdot\lambda}{\lambda-\theta\cdot\lambda}}+\lrA{\frac{p\cdot\lambda+(1-p)\cdot\theta\cdot\lambda}{2}\cdot\alpha\cdot\sigma+p\cdot\frac{1}{2}+(1-p)\cdot\frac{1}{4}\cdot\frac{3\lambda-2\theta\cdot\lambda}{\lambda-\theta\cdot\lambda}}}{\gamma\cdot\max\lrC{\sigma,\ 1}+0.5}\\
&=\frac{\lrA{\gamma\cdot\alpha\cdot\sigma+\frac{p\cdot\lambda+(1-p)\cdot\theta\cdot\lambda}{2}\cdot\alpha\cdot\sigma}+\lrA{\gamma\cdot\lrA{p\cdot\frac{2}{\lambda}+(1-p)\cdot\frac{1}{\theta\cdot\lambda}\cdot\frac{2\lambda-\theta\cdot\lambda}{\lambda-\theta\cdot\lambda}}+\lrA{p\cdot\frac{1}{2}+(1-p)\cdot\frac{1}{4}\cdot\frac{3\lambda-2\theta\cdot\lambda}{\lambda-\theta\cdot\lambda}}}}{\gamma\cdot\max\lrC{\sigma,\ 1}+0.5}
\end{split}
\end{equation}

When $\alpha=1.5$, $\gamma\in(0,0.428]$, and $\lambda=\min\{3.5\gamma/\alpha,1\}$, we have the following result.
\begin{itemize}
    \item When $\gamma\in(0,0.428]$, setting $\theta=0.5043$, we have $\max_{\sigma\geq 0} R(\sigma)<3.163$.
\end{itemize}

\begin{claim}
When $\gamma\in(0,0.428]$, setting $\theta=0.5043$, we have $\max_{\sigma\geq 0} R(\sigma)\leq 3.163$.
\end{claim}
\begin{proof}[Claim proof]
Note that $\lambda=\min\{1, 3.5\gamma/\alpha\}=3.5\gamma/\alpha$ and $\alpha=1.5$. Setting $\theta=0.5043$, we can obtain $\gamma\cdot\alpha\cdot\sigma+\frac{p\cdot\lambda+(1-p)\cdot\theta\cdot\lambda}{2}\cdot\alpha\cdot\sigma\leq 3.163\cdot\gamma\cdot\sigma$ and $\gamma\cdot\lra{p\cdot\frac{2}{\lambda}+(1-p)\cdot\frac{1}{\theta\cdot\lambda}\cdot\frac{2\lambda-\theta\cdot\lambda}{\lambda-\theta\cdot\lambda}}+\lra{p\cdot\frac{1}{2}+(1-p)\cdot\frac{1}{4}\cdot\frac{3\lambda-2\theta\cdot\lambda}{\lambda-\theta\cdot\lambda}}\leq 1.5811<3.163/2$. Hence, by (\ref{ratio2}), we have 
\begin{align*}
&R(\sigma)\leq \frac{3.163\cdot\gamma\cdot\sigma+3.163/2}{\gamma\cdot\max\lrC{\sigma,\ 1}+0.5}\leq 3.163\cdot\frac{\gamma\cdot\sigma+0.5}{\gamma\cdot\max\lrC{\sigma,\ 1}+0.5}\leq 3.163,
\end{align*}
which finishes the proof.
\end{proof}

Note that $\frac{1}{3.5}\approx 0.2857$. We can obtain the following result for any $1\leq \alpha\leq 1.5$.
\begin{itemize}
    \item When $\gamma\in(0,0.285]$, setting $\theta=0.5$, we have $\max_{\sigma\geq 0} R(\sigma)\leq \alpha+1+\ln2-0.029$ for any $1\leq \alpha \leq 1.5$.
\end{itemize}
\begin{claim}
When $\gamma\in(0,0.285]$, setting $\theta=0.5$, we have $\max_{\sigma\geq 0} R(\sigma)\leq \alpha+1+\ln2-0.029$ for any $1\leq \alpha \leq 1.5$.
\end{claim}
\begin{proof}[Claim proof]
We set $\lambda=\min\{1, 3.5\gamma/\alpha\}=3.5\gamma/\alpha$ and $\theta=1/2$. We can obtain that
\begin{equation}\label{thep}
p=\frac{\frac{1}{4(\lambda-\theta\cdot\lambda)}+\frac{\gamma}{\theta\cdot\lambda(\lambda-\theta\cdot\lambda)}}{\frac{1}{4\lambda}+\frac{1}{4(\lambda-\theta\cdot\lambda)}+\frac{\gamma}{\theta\cdot\lambda(\lambda-\theta\cdot\lambda)}}=\frac{\frac{1}{4(1-\theta)}+\frac{\alpha}{3.5}\cdot \frac{1}{\theta(1-\theta)}}{\frac{1}{4}+\frac{1}{4(1-\theta)}+\frac{\alpha}{3.5}\cdot\frac{1}{\theta(1-\theta)}}=\frac{8\alpha+3.5}{8\alpha+5.25},
\end{equation}
where the first equality follows from $\lambda=\min\{1, 3.5\gamma/\alpha\}=3.5\gamma/\alpha$ and the second from $\theta=1/2$.

One one hand, we have
\begin{equation}\label{ratio2-1}
\begin{split}
\gamma\cdot\alpha\cdot\sigma+\frac{p\cdot\lambda+(1-p)\cdot\theta\cdot\lambda}{2}\cdot\alpha\cdot\sigma &=\gamma\cdot\alpha\cdot\sigma\cdot\lrA{1+\frac{p+(1-p)\cdot\theta}{2}\cdot\frac{3.5}{\alpha}}\\
&=\gamma\cdot\sigma\cdot\lrA{\alpha+\frac{1+\frac{8\alpha+3.5}{8\alpha+5.25}}{2}\cdot\frac{3.5}{2}}\\
&<\gamma\cdot\sigma\cdot\lrA{\alpha+1+\ln2-0.03}\\
&\leq\gamma\cdot\max\{\sigma,\ 1\}\cdot\lrA{\alpha+1+\ln2-0.029},
\end{split}
\end{equation}
where the first equality follows from $\lambda=3.5\gamma/\alpha$, the second equality from (\ref{thep}) and $\theta=0.5$,
the first inequality from $\alpha+\frac{1+\frac{8\alpha+3.5}{8\alpha+5.25}}{2}\cdot\frac{3.5}{2}\leq \alpha+1+\ln2-0.029$ for any $1\leq \alpha\leq 1.5$, and the second inequality from $\sigma\leq\max\{\sigma,\ 1\}$.

On the other hand, by (\ref{thep}), we have
\begin{equation}\label{ratio2-2}
\begin{split}
&\gamma\cdot\lrA{p\cdot\frac{2}{\lambda}+(1-p)\cdot\frac{1}{\theta\cdot\lambda}\cdot\frac{2\lambda-\theta\cdot\lambda}{\lambda-\theta\cdot\lambda}}+\lrA{p\cdot\frac{1}{2}+(1-p)\cdot\frac{1}{4}\cdot\frac{3\lambda-2\theta\cdot\lambda}{\lambda-\theta\cdot\lambda}}\\
&=\frac{\alpha}{3.5}\cdot\lrA{2p+(1-p)\cdot\frac{1}{\theta}\cdot\frac{2-\theta}{1-\theta}}+p\cdot\frac{1}{2}+(1-p)\cdot\frac{1}{4}\cdot\frac{3-2\theta}{1-\theta}\\
&=\frac{\alpha}{3.5}\cdot\lrA{2+\frac{7}{8\alpha+5.25}}+\frac{4\alpha+3.5}{8\alpha+5.25}\\
&<\frac{1}{2}\cdot\lrA{\alpha+1+\ln2 -0.029},
\end{split}
\end{equation}
where the first equality follows from $\lambda=3.5\gamma/\alpha$, the second equality from (\ref{thep}) and $\theta=0.5$, and
the inequality from $\frac{\alpha}{3.5}\cdot\lrA{2+\frac{7}{8\alpha+5.25}}+\frac{4\alpha+3.5}{8\alpha+5.25}\leq \frac{1}{2}\cdot\lrA{\alpha+1+\ln2 -0.029}$ for any $1\leq \alpha\leq 1.5$.

Therefore, by (\ref{ratio2}), (\ref{ratio2-1}), and (\ref{ratio2-2}), we have 
\begin{align*}
&R(\sigma)<\frac{\gamma\cdot\max\{\sigma,\ 1\}\cdot\lrA{\alpha+1+\ln2-0.029}+\frac{1}{2}\cdot\lrA{\alpha+1+\ln2 -0.029}}{\gamma\cdot\max\lrC{\sigma,\ 1}+0.5}=\alpha+1+\ln2 -0.029,
\end{align*}
which finishes the proof.
\end{proof}

This finishes the proof.
\end{proof}

\subsection{Proof of Lemma~\ref{sl2}}
\begingroup
\def\thelemma{\ref{sl2}}
\begin{lemma}
Conditioned on $\chi=d$, when serving each customer $v_i$ in $ALG.S(\lambda)$, the expected cumulative cost of the vehicle due to the possible additional visit(s) to $v_0$ is $a\cdot 2\cdot \frac{d_i}{\lambda}\cdot l_i+b\cdot d_i\cdot l_i$.
\end{lemma}
\endgroup
\begin{proof}
We first consider $d_i\leq\lambda$. In this case, the vehicle can satisfy $v_i$ using at most one additional visit to $v_0$.

By Lemma~\ref{sl0}, the vehicle carries $L_{i-1}$ units of goods during the vehicle's travel from $v_{i-1}$ to $v_i$, where $L_{i-1}\sim U[0,\lambda)$. Hence, the vehicle incurs one additional visit to $v_0$ with a probability of $\PP{L_{i-1}=x<d_i}=\int_{0}^{d_i}\frac{1}{\lambda}dx=\frac{d_i}{\lambda}$.

If the vehicle incurs an additional visit to $v_0$, by $ALG.S(\lambda)$, the vehicle will carry $0$ units of goods from $v_i$ to $v_0$ and $\lambda$ units of goods from $v_0$ to $v_i$. So, the cumulative cost of this additional visit will be $a\cdot 2\cdot l_i + b\cdot \lambda\cdot l_i$. Since $L_{i-1}\sim U[0,\lambda)$ and $d_i\leq \lambda$, the expected cumulative cost of this additional visit conditioned on the demand realization is 
$
\int_{0}^{d_i}\frac{a\cdot 2\cdot l_i + b\cdot \lambda\cdot l_i}{\lambda}dx=a\cdot 2\cdot \frac{d_i}{\lambda}\cdot l_i+b\cdot d_i\cdot l_i.
$

Next, we consider $d_i>\lambda$.

In this case, we may implicitly split each customer $v_i$ with $d_i>\lambda$ into a set of customers at the same place, denoted as $V_i$, with each having a demand of at most $\lambda$ and all having a total demand of $d_i$. Define $l_v\coloneqq w(v_0,v)$ for each $v\in V_i$. Hence, $l_i=l_v$ for each $v\in V_i$ and $d_i\cdot l_i=\sum_{v\in V_i} d_v\cdot l_v$. 
Denote the original instance and the new instance by $I$ and $I'$, respectively. The TSP tour in $I'$, denoted as $v_0p_1\dots p_nv_0$ where $p_i$ can be any permutation of customers in $V_i$, has the same cost as the TSP tour $T^*=v_0v_1\dots v_nv_0$ in $I$.
It can be verified that based on the greedy strategy, $ALG.S(\lambda)$ generates two equivalent solutions for $I$ and $I'$ with the same cumulative cost. 
Moreover, the expected cumulative cost of the possible additional visits when visiting $v_i$ in $I$ equals to the total expected cumulative cost of the possible additional visits when visiting all customers in $V_i$ in $I'$.

Since the demand of each customer in $I'$ is at most $\lambda$, by the previous analysis of the case where $d_i\leq\lambda$, the expected cumulative cost of the possible additional visits to $v_i$ conditioned on the demand realization is 
\[
\sum_{v\in V_i}\lrA{a\cdot 2\cdot \frac{d_v}{\lambda}\cdot l_v+b\cdot d_v\cdot l_v}=a\cdot 2\cdot\sum_{v\in V_i} \frac{d_v}{\lambda}\cdot l_i+b\cdot \lambda\cdot\sum_{v\in V_i} d_v\cdot l_i=a\cdot 2\cdot \frac{d_i}{\lambda}\cdot l_i+b\cdot d_i\cdot l_i,
\]
which finishes the proof.
\end{proof}

\subsection{Proof of Lemma~\ref{t9}}
\begingroup
\def\thelemma{\ref{t9}}
\begin{lemma}
For splittable Cu-VRPSD with any $\lambda\in(0,1]$, conditioned on $\chi=d$, $ALG.S(\lambda)$ generates a solution $\T$ with an expected cumulative cost of
\[
\frac{\gamma\cdot\lrA{\alpha\cdot\sigma+\frac{1}{\lambda}}+\frac{\lambda}{2}\cdot\lrA{\alpha\cdot\sigma+\frac{1}{\lambda}}}{\gamma\cdot\max\lrC{\sigma,\ 1}+0.5}\cdot LB.
\]
\end{lemma}
\endgroup
\begin{proof}
By Lemma~\ref{sl1}, the expected cumulative cost of the vehicle for traveling the edges on the TSP tour $T^*$ is 
$
\sum_{e\in E(T^*)}\lra{a\cdot w(e)+b\cdot\frac{\lambda}{2}\cdot w(e)}=a\cdot w(T^*)+b\cdot\frac{\lambda}{2}\cdot w(T^*)\leq a\cdot\alpha\cdot\tau+b\cdot\frac{\lambda}{2}\cdot\alpha\cdot\tau,
$
where the inequality follows from $w(T^*)\leq \alpha\cdot\tau$ since $T^*$ is an $\alpha$-approximate TSP tour.

By Lemma~\ref{sl2}, when serving customers in $V'$, the expected cumulative cost due to the possible additional visit(s) to $v_0$ is
\[
\sum_{v_i\in V'}\lrA{a\cdot 2\cdot \frac{d_i}{\lambda}\cdot l_i+b\cdot d_i\cdot l_i}=a\cdot\sum_{v_i\in V'} 2\cdot \frac{d_i}{\lambda}\cdot l_i+b\cdot\sum_{v_i\in V'} d_i\cdot l_i.
\]

Therefore, $ALG.S(\lambda)$ generates a solution $\T$ with an expected cumulative cost of
\begin{align*}
&\frac{a\cdot\lrA{\alpha\cdot\tau+\sum_{v_i\in V'} 2\cdot \frac{d_i}{\lambda}\cdot l_i}+b\cdot\lrA{\frac{\lambda}{2}\cdot\alpha\cdot\tau+\sum_{v_i\in V'} d_i\cdot l_i}}{LB}\cdot LB\\
&=\frac{\gamma\cdot\lrA{\alpha\cdot\sigma+\int_{0}^{1} \frac{1}{\lambda}dF(x)}+\lrA{\frac{\lambda}{2}\cdot\alpha\cdot\sigma+\int_{0}^{1} \frac{1}{2}dF(x)}}{\gamma\cdot\max\lrC{\sigma,\ 1}+0.5}\cdot LB\\
&=\frac{\gamma\cdot\lrA{\alpha\cdot\sigma+\frac{1}{\lambda}}+\frac{\lambda}{2}\cdot\lrA{\alpha\cdot\sigma+\frac{1}{\lambda}}}{\gamma\cdot\max\lrC{\sigma,\ 1}+0.5}\cdot LB
\end{align*}
where the equality follows from dividing both the numerator and denominator by $b\cdot\eta$, using $\gamma=a/b$, $\sigma=\tau/\eta$, and $LB=a\cdot\max\{\tau, \eta\}+b\cdot\frac{1}{2}\cdot\eta$ by Lemma~\ref{lb}, and equitation (\ref{int}).
\end{proof}

\section{Numerical Evaluation of the Approximation Ratio of APPROX.2}\label{tightness}
Let $N$ be a multiple of 3, saying $N=300$. 
For any $i\in\{0,1,2\}$ and $j\in[N]$, we set 
\begin{equation}\label{N1}
r^i_j=\int_{0}^{\frac{j}{N}}x^idF(x).
\end{equation}

By (\ref{int}), for any $i\in\{1,2\}$ and $j\in[t]$, we have 
\begin{equation}\label{N2}
\frac{j-1}{N}\cdot (r^{i-1}_j-r^{i-1}_{j-1})\leq (r^i_j-r^i_{j-1})\leq \frac{j}{N}\cdot (r^{i-1}_j-r^{i-1}_{j-1})\quad\quad\mbox{and}\quad\quad r^1_N=1.
\end{equation}

By the proof of Theorem~\ref{t4}, the approximation ratio of $ALG.1(1,1/3)$ is at most
\begin{equation}\label{N3}
\begin{split}
&\frac{\gamma\cdot\lrA{\alpha\cdot\sigma+\int^{\frac{1}{3}}_0\frac{3x}{2}dF(x)+\int^\frac{2}{3}_{\frac{1}{3}}\frac{6x-1}{2}dF(x)+\int^1_\frac{2}{3}\frac{3x+1}{2}dF(x)}}{\gamma\cdot\max\lrC{\sigma,\ 1}+0.5}\\
&\quad+\frac{\lrA{\frac{2}{3}\cdot\alpha\cdot\sigma+\int^{\frac{1}{3}}_0\frac{4x-3x^2}{4}dF(x)+\int^\frac{2}{3}_{\frac{1}{3}}\frac{3x^2+2x}{4}dF(x)+\int^1_\frac{2}{3}\frac{9x^2-6x+4}{6}dF(x)}}{\gamma\cdot\max\lrC{\sigma,\ 1}+0.5}\\
&=\frac{\gamma\cdot\lrA{\alpha\cdot\sigma+\frac{3}{2}r^1_{\frac{1}{3}N}+3(r^1_{\frac{2}{3}N}-r^1_{\frac{1}{3}N})-\frac{1}{2}(r^0_{\frac{2}{3}N}-r^0_{\frac{1}{3}N})+\frac{3}{2}(r^1_{N}-r^1_{\frac{2}{3}N})+\frac{1}{2}(r^0_{N}-r^0_{\frac{2}{3}N})}}{\gamma\cdot\max\lrC{\sigma,\ 1}+0.5}\\
&\quad+\frac{\lrA{\frac{2}{3}\cdot\alpha\cdot\sigma-\frac{3}{4}r^2_{\frac{1}{3}N}+r^1_{\frac{1}{3}N}+\frac{3}{4}(r^2_{\frac{2}{3}N}-r^2_{\frac{1}{3}N})+\frac{1}{2}(r^1_{\frac{2}{3}N}-r^1_{\frac{1}{3}N})+\frac{3}{2}(r^2_{N}-r^2_{\frac{2}{3}N})-(r^1_{N}-r^1_{\frac{2}{3}N})+\frac{2}{3}(r^0_{N}-r^0_{\frac{2}{3}N})}}{\gamma\cdot\max\lrC{\sigma,\ 1}+0.5},
\end{split}
\end{equation}
where the equality follows from (\ref{N1}).

Moreover, the approximation ratio of $ALG.2(1,1/3)$ is at most
\begin{equation}\label{N4}
\begin{split}
&\frac{\gamma\cdot\lrA{\alpha\cdot\sigma+\int^{\frac{1}{3}}_0\frac{3x}{2}dF(x)}+\lrA{\frac{2}{3}\cdot\alpha\cdot\sigma+\int^{\frac{1}{3}}_0\frac{4x-3x^2}{4}dF(x)}}{\gamma\cdot\max\lrC{\sigma,\ 1}+0.5}+\min\lrC{\frac{\int^{1}_\frac{1}{3}\frac{2\gamma+x}{2}dF(x)}{\gamma\cdot\max\lrC{\sigma,\ 1}+0.5},\ 1}\\
&=\frac{\gamma\cdot\lrA{\alpha\cdot\sigma+\frac{3}{2}r^1_{\frac{1}{3}N}}+\lrA{\frac{2}{3}\cdot\alpha\cdot\sigma-\frac{3}{4}r^2_{\frac{1}{3}N}+r^1_{\frac{1}{3}N}}}{\gamma\cdot\max\lrC{\sigma,\ 1}+0.5}+\min\lrC{\frac{
\frac{1}{2}(r^1_{N}-r^1_{\frac{1}{3}N})+\gamma\cdot(r^0_{N}-r^0_{\frac{1}{3}N})}{\gamma\cdot\max\lrC{\sigma,\ 1}+0.5},\ 1},
\end{split}
\end{equation}
where the equality follows from (\ref{N1}).

We have the following two cases.

\textbf{Case~1: $\frac{1}{2}(r^1_{N}-r^1_{\frac{1}{3}N})+\gamma\cdot(r^0_{N}-r^0_{\frac{1}{3}N})\leq \gamma\cdot\max\lrC{\sigma,\ 1}+0.5$.}

Recall that in $APPROX.2$ we call $ALG.1(1,1/3)$ (resp., $ALG.2(1,1/3)$) with a probability of $0.5$ (resp., $0.5$). By (\ref{N1})-(\ref{N4}), the approximation ratio of $APPROX.2$ is at most the value of the following linear program, where $\gamma$, $\sigma$, and $N$ are all fixed parameters.
{\small
\begin{alignat}{2}\label{lp1}
\max\quad & y_\sigma \nonumber \\
\mbox{s.t.}\quad 
& y_\sigma\leq \frac{1}{2}(A+B),\nonumber\\
&A\leq\frac{\gamma\cdot\lrA{\alpha\cdot\sigma+\frac{3}{2}r^1_{\frac{1}{3}N}+3(r^1_{\frac{2}{3}N}-r^1_{\frac{1}{3}N})-\frac{1}{2}(r^0_{\frac{2}{3}N}-r^0_{\frac{1}{3}N})+\frac{3}{2}(r^1_{N}-r^1_{\frac{2}{3}N})+\frac{1}{2}(r^0_{N}-r^0_{\frac{2}{3}N})}}{\gamma\cdot\max\lrC{\sigma,\ 1}+0.5}\nonumber\\
&\quad\quad+\frac{\lrA{\frac{2}{3}\cdot\alpha\cdot\sigma-\frac{3}{4}r^2_{\frac{1}{3}N}+r^1_{\frac{1}{3}N}+\frac{3}{4}(r^2_{\frac{2}{3}N}-r^2_{\frac{1}{3}N})+\frac{1}{2}(r^1_{\frac{2}{3}N}-r^1_{\frac{1}{3}N})+\frac{3}{2}(r^2_{N}-r^2_{\frac{2}{3}N})-(r^1_{N}-r^1_{\frac{2}{3}N})+\frac{2}{3}(r^0_{N}-r^0_{\frac{2}{3}N})}}{\gamma\cdot\max\lrC{\sigma,\ 1}+0.5},\nonumber\\
&B\leq \frac{\gamma\cdot\lrA{\alpha\cdot\sigma+\frac{3}{2}r^1_{\frac{1}{3}N}}+\lrA{\frac{2}{3}\cdot\alpha\cdot\sigma-\frac{3}{4}r^2_{\frac{1}{3}N}+r^1_{\frac{1}{3}N}}+\lrA{\frac{1}{2}(r^1_{N}-r^1_{\frac{1}{3}N})+\gamma\cdot(r^0_{N}-r^0_{\frac{1}{3}N})}}{\gamma\cdot\max\lrC{\sigma,\ 1}+0.5},\nonumber\\
&\frac{1}{2}(r^1_{N}-r^1_{\frac{1}{3}N})+\gamma\cdot(r^0_{N}-r^0_{\frac{1}{3}N})\leq \gamma\cdot\max\lrC{\sigma,\ 1}+0.5,\nonumber\\
&0\leq \frac{j-1}{N}\cdot (r^{i-1}_j-r^{i-1}_{j-1})\leq (r^i_j-r^i_{j-1})\leq \frac{j}{N}\cdot (r^{i-1}_j-r^{i-1}_{j-1}), \forall i\in\{1,2\}, j\in[N], \nonumber\\
&r^1_N=1.\nonumber
\end{alignat}}

\textbf{Case~2: $\frac{1}{2}(r^1_{N}-r^1_{\frac{1}{3}N})+\gamma\cdot(r^0_{N}-r^0_{\frac{1}{3}N})\geq \gamma\cdot\max\lrC{\sigma,\ 1}+0.5$.}

Similarly, by (\ref{N1})-(\ref{N4}), the approximation ratio of $APPROX.2$ is at most the value of the following linear program, where $\gamma$, $\sigma$, and $N$ are all fixed parameters.
{\small
\begin{alignat}{2}
\max\quad & y_\sigma \nonumber \\
\mbox{s.t.}\quad 
& y_\sigma\leq \frac{1}{2}(A+B),\nonumber\\
&A\leq\frac{\gamma\cdot\lrA{\alpha\cdot\sigma+\frac{3}{2}r^1_{\frac{1}{3}N}+3(r^1_{\frac{2}{3}N}-r^1_{\frac{1}{3}N})-\frac{1}{2}(r^0_{\frac{2}{3}N}-r^0_{\frac{1}{3}N})+\frac{3}{2}(r^1_{N}-r^1_{\frac{2}{3}N})+\frac{1}{2}(r^0_{N}-r^0_{\frac{2}{3}N})}}{\gamma\cdot\max\lrC{\sigma,\ 1}+0.5}\nonumber\\
&\quad\quad+\frac{\lrA{\frac{2}{3}\cdot\alpha\cdot\sigma-\frac{3}{4}r^2_{\frac{1}{3}N}+r^1_{\frac{1}{3}N}+\frac{3}{4}(r^2_{\frac{2}{3}N}-r^2_{\frac{1}{3}N})+\frac{1}{2}(r^1_{\frac{2}{3}N}-r^1_{\frac{1}{3}N})+\frac{3}{2}(r^2_{N}-r^2_{\frac{2}{3}N})-(r^1_{N}-r^1_{\frac{2}{3}N})+\frac{2}{3}(r^0_{N}-r^0_{\frac{2}{3}N})}}{\gamma\cdot\max\lrC{\sigma,\ 1}+0.5},\nonumber\\
&B\leq \frac{\gamma\cdot\lrA{\alpha\cdot\sigma+\frac{3}{2}r^1_{\frac{1}{3}N}}+\lrA{\frac{2}{3}\cdot\alpha\cdot\sigma-\frac{3}{4}r^2_{\frac{1}{3}N}+r^1_{\frac{1}{3}N}}}{\gamma\cdot\max\lrC{\sigma,\ 1}+0.5}+1,\nonumber\\
&\frac{1}{2}(r^1_{N}-r^1_{\frac{1}{3}N})+\gamma\cdot(r^0_{N}-r^0_{\frac{1}{3}N})\geq \gamma\cdot\max\lrC{\sigma,\ 1}+0.5,\nonumber\\
&0\leq \frac{j-1}{N}\cdot (r^{i-1}_j-r^{i-1}_{j-1})\leq (r^i_j-r^i_{j-1})\leq \frac{j}{N}\cdot (r^{i-1}_j-r^{i-1}_{j-1}), \forall i\in\{1,2\}, j\in[N], \nonumber\\
&r^1_N=1.\nonumber
\end{alignat}}

Setting $\gamma=1.444$ and $N=300$, the values of the above linear programs under $\sigma\in[1,3]$ can be found in Figure~\ref{fig4}. 

\begin{figure}[ht]
    \centering
    \includegraphics[scale=0.6]{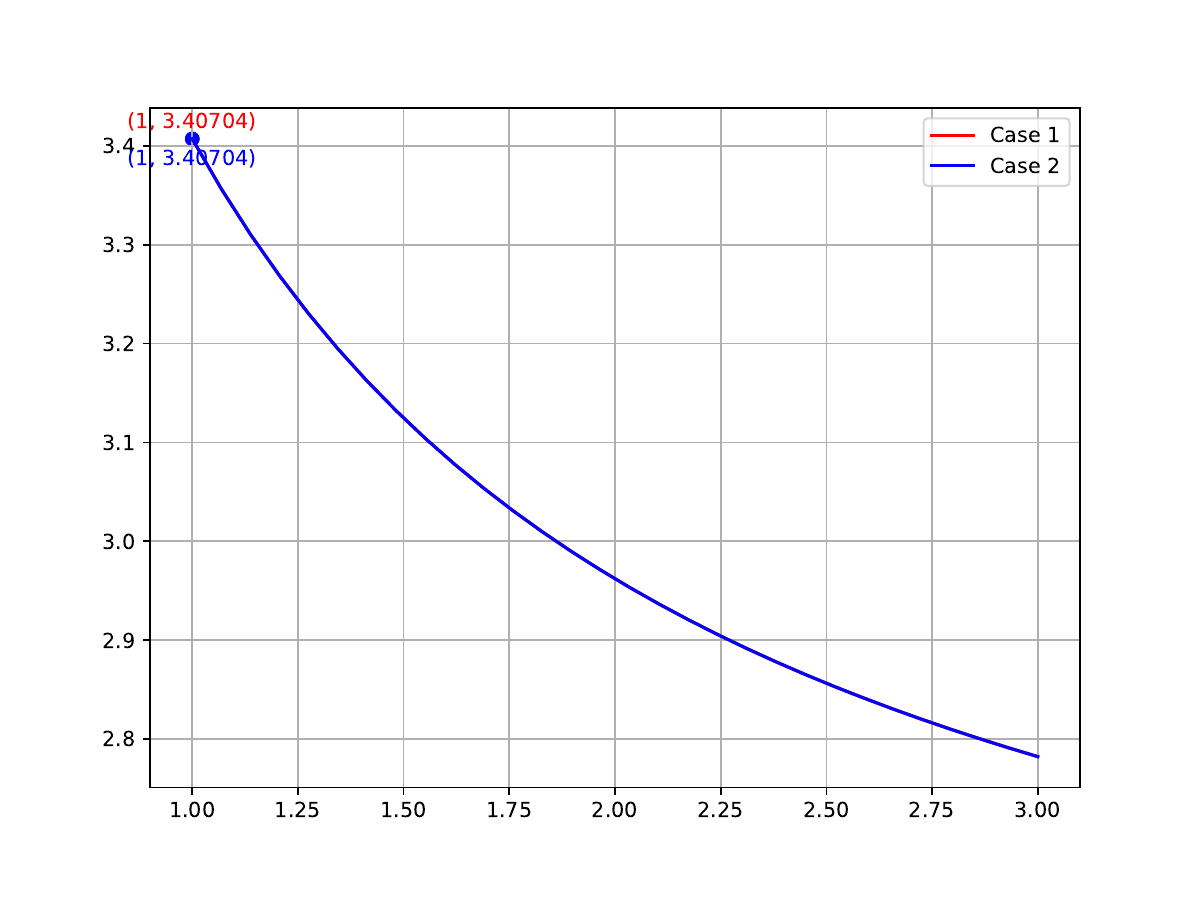}
    \caption{The values of the linear programs under $\sigma\in[1,3]$, where the red line denotes the value of the linear program in the first case and the blue line denotes the value of the linear program in the second case.}
    \label{fig4}
\end{figure}

We can see that their values are the same, i.e., $\frac{
\frac{1}{2}(r^1_{N}-r^1_{\frac{1}{3}N})+\gamma\cdot(r^0_{N}-r^0_{\frac{1}{3}N})}{\gamma\cdot\max\lrC{\sigma,\ 1}+0.5}=1$. Moreover, it suggests that when $\gamma\in [1.444,+\infty)$, the approximation ratio of $APPROX.2$ may actually achieve $3.408$ and in the worst case $\sigma=1$.
We remark that increasing $N$ may yield values closer to the approximation ratio. For example, setting $N=3000$, the maximum value obtained from the linear programs under $\sigma\in[1,3]$ achieve $3.404$, which is slightly better than $3.408$.


Similarly, we can use the linear programs to evaluate the approximation ratio of $APPROX.2$ under different values of $\gamma$.
Note that for each choice of $\gamma$, we need to consider the maximum value obtained from the linear programs for all $\delta\geq 1$.
However, as suggested in Figure~\ref{fig4}, in the worst, we may have $\sigma=1$ when $\gamma$ is large. Therefore, we first evaluate the effect of $\sigma$. 
The maximum values obtained from the linear programs under $\sigma\in[1,3]$ for each $\gamma\in\{0.4,0.5,...,2\}$ are shown in Figure~\ref{fig5}.

\begin{figure}[ht]
    \centering
    \includegraphics[scale=0.6]{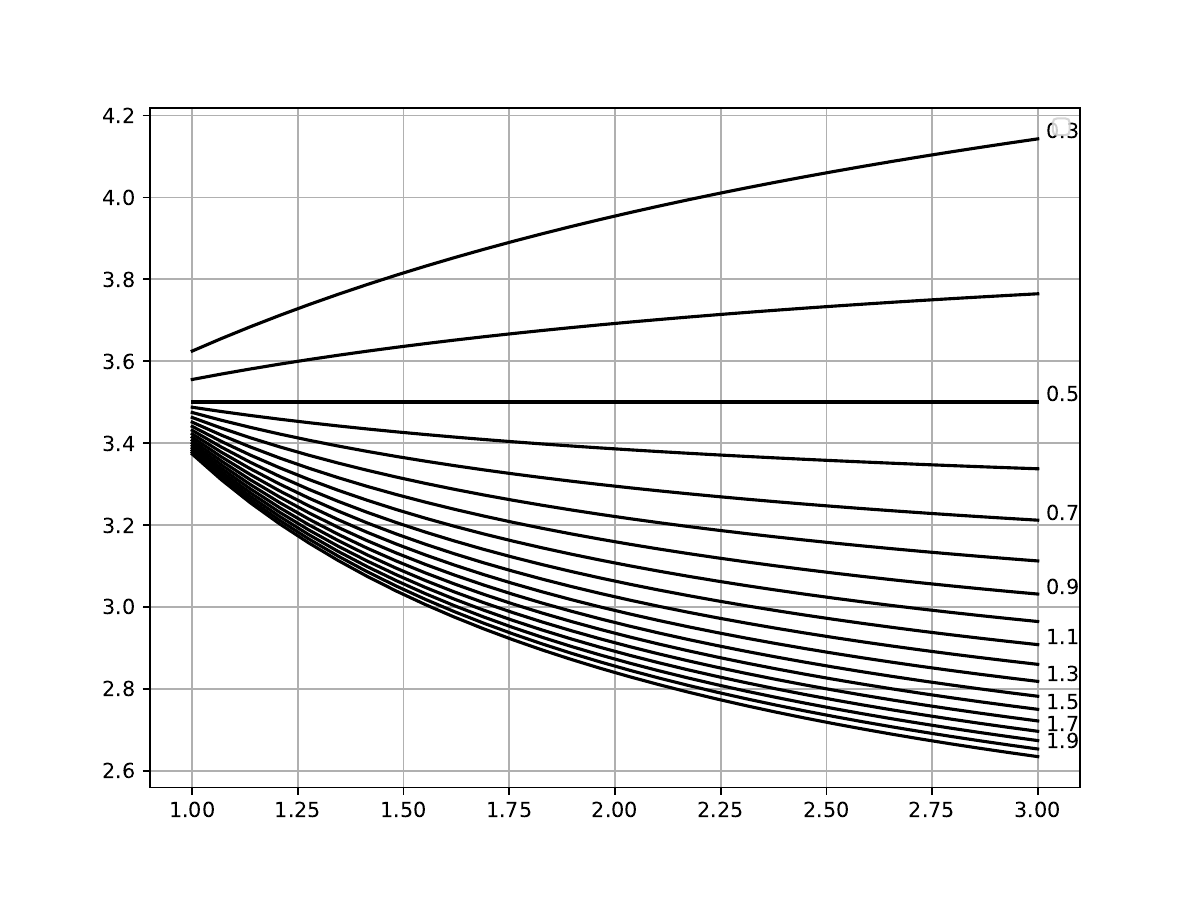}
    \caption{The maximum values obtained from the linear programs under $\sigma\in[1,3]$ for each $\gamma\in\{0.4,0.5,...,2\}$.}
    \label{fig5}
\end{figure}

This suggests that in the worst case of $APPROX.2$ we may have $\sigma=1$ for $\gamma\geq 0.6$ and $\sigma=\infty$ for $\gamma\leq 0.4$. Therefore, we may simply consider $\sigma=1$ when using the linear programs for $\gamma\geq 0.6$.

Based on the above observation, we evaluate the approximation ratio of $APPROX.2$ under different values of $\gamma\geq 0.6$ by simply considering $\sigma=1$. 
The results are shown in Figure~\ref{fig6}.

\begin{figure}[ht]
    \centering
    \includegraphics[scale=0.6]{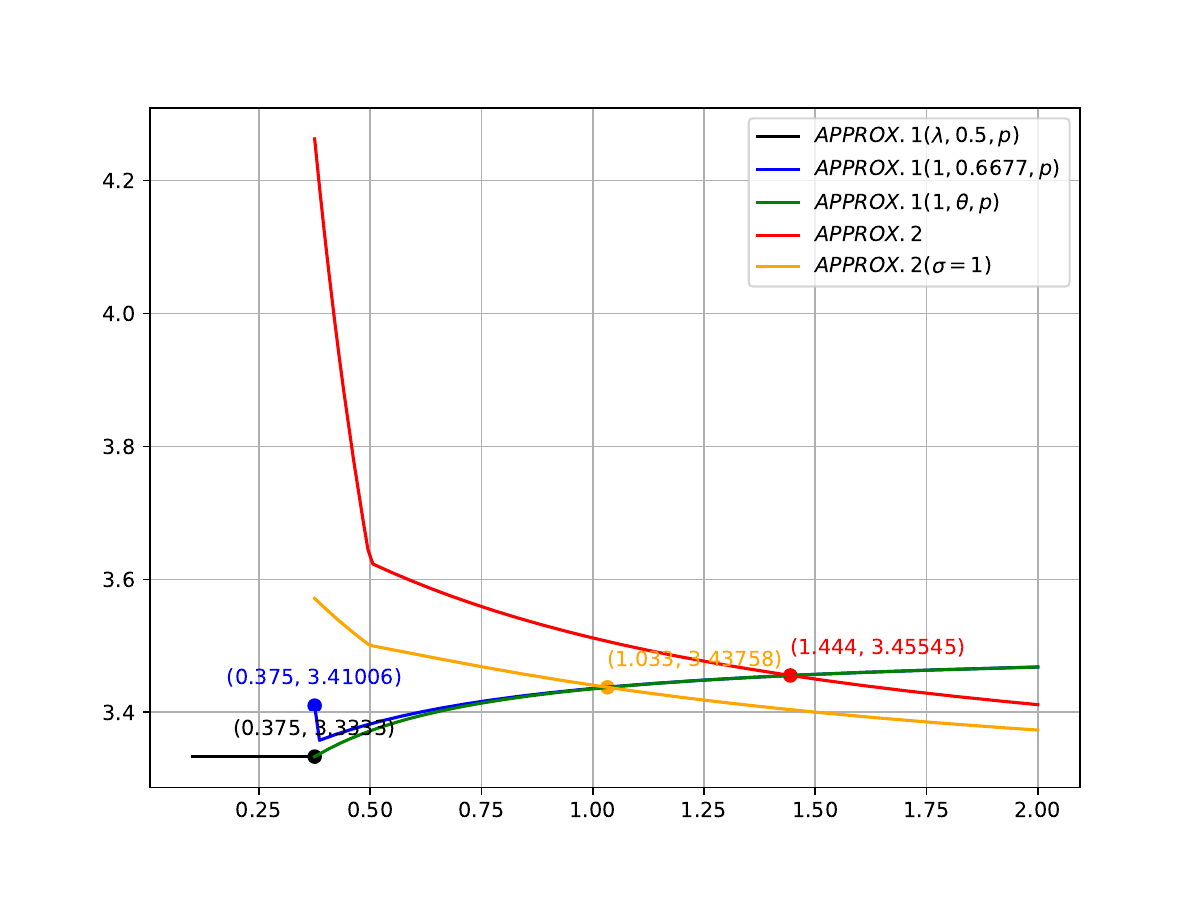}
    \caption{The approximation ratios of $APPROX.1(\lambda,\theta,p)$ and $APPROX.2$ under $\gamma\in(0,2)$, where the orange line denotes the maximum values obtained from the linear programs.}
    \label{fig6}
\end{figure}

Therefore, by using $APPROX.1$ when $\gamma\leq 1.033$ and $APPROX.2$ otherwise, we may obtain a randomized $3.438$-approximation algorithm for unsplittable Cu-VRPSD.

\bibliographystyle{abbrvnat}
\bibliography{zmain}

\end{document}